\documentclass[11pt]{article}
\usepackage[utf8]{inputenc}
\usepackage{graphicx}
\usepackage{subcaption}
\usepackage{amsfonts}
\usepackage{amsmath}
\usepackage{amssymb}
\usepackage{bm}
\usepackage{url}
\usepackage{fancyhdr}
\usepackage{float}
\usepackage{indentfirst}
\usepackage{enumerate}
\usepackage{dsfont,footmisc}
\usepackage{comment}
\usepackage{caption}
\usepackage{multirow}
\usepackage{tabularx}
\usepackage{booktabs}
\usepackage{threeparttable}
\usepackage{soul}

\usepackage{setspace}

\usepackage[misc]{ifsym}

\usepackage{subcaption}
\usepackage{caption}
\usepackage[font=footnotesize]{caption}

\usepackage{float}
\usepackage{amsthm}
\usepackage{color}
\usepackage{natbib}
\usepackage[british]{babel}
\usepackage[colorlinks=true,citecolor=blue]{hyperref}
\usepackage[right,pagewise,mathlines]{lineno}
\usepackage{tikz}

\newcommand{\E}{\mathbb{E}}
\newcommand{\R}{\mathbb{R}}

\renewcommand{\ge}{\geqslant}
\renewcommand{\le}{\leqslant}
\renewcommand{\geq}{\geqslant}
\renewcommand{\leq}{\leqslant}
\renewcommand{\epsilon}{\varepsilon}

\theoremstyle{plain}
\newtheorem{theorem}{Theorem}[section]

\newtheorem{lemma}{Lemma}[section]
\newtheorem{proposition}{Proposition}[section]

\theoremstyle{definition}
\newtheorem{definition}{Definition}[section]

\theoremstyle{remark}

\theoremstyle{definition}

\numberwithin{equation}{section} \numberwithin{theorem}{section}
\renewcommand{\cite}{\citet}

\DeclareMathOperator*{\argmin}{arg\,min}
\newcommand{\be}{\begin{eqnarray}}
\newcommand{\ee}{\end{eqnarray}}
\newcommand{\by}{\begin{eqnarray*}}
\newcommand{\ey}{\end{eqnarray*}}
\newcommand{\bn}{\begin{enumerate}}
\newcommand{\en}{\end{enumerate}}
\newcommand{\bi}{\begin{itemize}}
\newcommand{\ei}{\end{itemize}}

\title{Robust  enhanced    index tracking portfolio selection under distributional  uncertainty}
\author{Jun Cai\thanks{Department of Statistics and Actuarial Science, University of Waterloo, 200 University Ave. West, Waterloo, N2L 3G1, ON, Canada. \Letter~{\url{jcai@uwaterloo.ca}}}, \  \  \ \  Zhiqiao Song\thanks{Department of Statistics and Actuarial Science, University of Waterloo, 200 University Ave. West, Waterloo, N2L 3G1, ON, Canada. \Letter~{\url{z22song@uwaterloo.ca}}},  \ \ \  \
}
\date{\today}

\begin{document}

\maketitle

	\begin{abstract}

The enhanced index tracking (EIT) portfolio selection problem aims to construct a portfolio that is expected to 
outperform a benchmark index. In practice, investors face uncertainty in the joint distribution of asset and index 
losses, as the true distribution is typically unknown and only partial information is available. Moreover, portfolio selection 
requires balancing the trade-off between mean loss and downside risk while targeting higher returns. Effectively modeling and 
managing these factors in EIT portfolio selection is therefore of both practical and research interest.

To address these challenges, this paper proposes robust EIT portfolio selection models under distributional uncertainty, guided 
by Pareto-optimality theory. The models minimize a convex combination of worst-case mean loss and downside risk, subject to 
a constraint on worst-case mean return. Two uncertainty sets are considered: one assuming a known mean loss vector 
and covariance matrix, and a more general set with unknown moments. Closed-form solutions are derived for the former, 
while the latter reduces to tractable convex optimization problems. Empirical results based on real market data show that the proposed 
models outperform not only the tracked index but also the equally weighted portfolio, robust mean-variance models, 
and robust variance models in terms of out-of-sample cumulative wealth, as well as Sharpe, Sortino, and Omega ratios. 
Moreover, portfolios that minimize a convex combination of worst-case mean loss and downside risk consistently 
outperform those that minimize downside risk alone. We also provide practical guidance for selecting optimal 
weighting coefficients to effectively balance risk and return.
\\

  {\bf Keywords}: Enhanced index tracking, distributional uncertainty, robust portfolio selection, downside risk. 
  	\end{abstract}

\maketitle

\newpage

\section{Introduction}

 Portfolio selection has been a central problem in finance since the seminal work of \citet{markowitz1952}. 
 Markowitz was the first to formulate the portfolio selection problem as an optimization problem by minimizing 
 the variance of portfolio returns (or, equivalently, portfolio losses), subject to a constraint on the mean return
  (or mean loss). This framework captures the fundamental trade-off between risk and return. Since then, 
  numerous portfolio optimization models have been developed from various perspectives, employing 
  different objective functions and constraints.

In financial markets, an index, composed of a representative set of assets,  is widely used as a benchmark to 
evaluate portfolio performance. Investors, including individuals, mutual fund managers, and pension fund managers, 
typically seek to construct portfolios that outperform a benchmark index.

Enhanced index tracking (EIT) aims to construct portfolios that outperform a given benchmark. A common modelling 
approach is to minimize the downside (right-tail) risk of portfolio loss relative to the index loss. Intuitively, investors 
are less concerned with outcomes where the portfolio outperforms the index (left tail), while greater downside risk 
indicates poorer performance. Downside risk is often measured using upper partial moments, defined as moments 
of portfolio loss exceeding the index loss or a specified threshold. These measures have been widely used in portfolio 
optimization, insurance, and economics (see, e.g., \cite{roy1952safety, bawa1975optimal, MTXY93, chen2011tight, 
KLMY17, LCLX19, LSW20, YHLH21, LMYZ22}).

EIT models have been extensively studied in the literature. For example, \citet{beasley2003evolutionary} consider 
index tracking with transaction costs, while \citet{canakgoz2009mixed} propose mixed-integer linear programming 
formulations. \citet{guastaroba2016linear} develop models based on the Omega ratio, and \citet{HLY18} study 
an index tracking model with a conditional value-at-risk (CVaR) constraint. \citet{guastaroba2020enhanced} propose 
bi-criteria models using weighted CVaR measures, and \citet{beraldi2022enhanced} analyze chance-constrained 
formulations. More recently, \citet{HLY22} develop models that balance upper and lower partial moments of 
portfolio loss. \citet{CP2025} propose models that not only stochastically dominate a given benchmark but also 
generate excess returns relative to a benchmark index.

Two key challenges arise in EIT modeling. The first concerns the computation of objective functions and constraints,  
such as partial moments, value-at-risk (VaR), and CVaR,  which typically require full knowledge of the joint 
distribution of asset returns and the benchmark index. In practice, this distribution is unknown, and only partial 
information is available. Existing approaches include parametric assumptions (e.g., multivariate normal or $t$ distributions; 
 see \cite{KLMY17, LMYZ22}) and empirical approximations based on historical data (e.g., \citet{LSW20, YHLH21, HLY18, HLY22}). 
  An increasingly important alternative is distributionally robust optimization (DRO), which minimizes worst-case objective 
  values over a set of plausible distributions (an uncertainty set). This framework has attracted significant
   attention in operations research, insurance, and finance. See, for example, 
    \cite{07PI, chen2011tight, lejeune2012game, ling2014robust, YYY16, xu2018sparse, 
    kang2022robust, CJM2025a}; and the references therein.
      The second challenge concerns the trade-off among multiple objectives. In EIT models, objectives 
such as mean loss, variance, downside risk, and VaR/CVaR often conflict, making joint optimization difficult. 
 Traditional approaches either minimize one objective subject to constraints on others (as in the Markowitz framework) 
 or minimize a convex combination of multiple objectives. This raises important questions regarding their relative 
 effectiveness and practical relevance:  
   including whether one approach leads to better portfolio performance than the other, which approach is more appropriate in practice, and what their respective advantages and limitations are.

Motivated by these challenges, this paper studies a robust EIT problem under distributional uncertainty. 
We show that minimizing a convex combination of mean loss and downside risk has a solid theoretical foundation, 
as the resulting portfolios are Pareto-optimal. Empirically, we demonstrate that this approach outperforms models that minimize 
downside risk alone (subject to a return constraint) in terms of cumulative wealth and Sharpe, Sortino, and Omega ratios.

When the uncertainty set specifies a given mean vector and covariance matrix, we derive explicit closed-form 
solutions for the optimal portfolio. For a more general mean-covariance-based uncertainty set (as in \cite{kang2019data}), 
we show that the problem can be reformulated as a tractable convex optimization problem. Empirical results further confirm 
that the proposed models outperform the equally weighted portfolio, as well as robust mean-variance and robust variance models, 
in out-of-sample cumulative wealth and risk-adjusted performance. We also provide practical guidance for selecting optimal 
weighting coefficients to balance risk and return.

The remainder of this paper is organized as follows. Section \ref{sec-2} introduces the robust EIT framework. 
Section \ref{sec-3} studies the case with known mean vector and covariance matrix and derives closed-form solutions. 
Section \ref{sec-4} considers a generalized uncertainty set and reformulates the models as convex optimization problems.
  Section \ref{sec-5} presents empirical results based on real market data. Section \ref{sec-6} concludes the paper.  The proofs of the main results derived in this paper are presented in the Appendix.

      \section{Robust EIT portfolio selection framework}
      \label{sec-2}

Suppose an investor aims to allocate wealth through a portfolio of $n$ risky assets. Let $\bm{X}=(X_1, ..., X_n)^{\top}\in\mathbb{R}^n$,
  where $X_i$ denotes the loss of the  $i$-th asset  at the end of the investment period with  $\mu_i=\mathbb{E}[X_i]$ and $\sigma_i^2={\rm Var}(X_i)$. Here $\R=(-\infty, \infty)$. Let $\bm{\omega}=(\omega_1,..., \omega_n)^{\top}\in\mathbb{R}^n$ be the vector of decision variables, which represents the weights/proportions of wealth allocated to each of the $n$ assets. We call $\bm{\omega}$ a portfolio or an allocation.   Without loss of generality, we assume the total wealth invested at the 
beginning  of 
the investment period  is 1. It holds that $\bm{\omega}^{\top}\bm{e}=\sum_{i=1}^n\omega_i=1$, with $\bm{e}^{\top}=(1,...,1)\in\mathbb{R}^n$ being an $n$-dimensional unit vector. 
 
Let $X_0$ be  the  loss of an index at the end of the investment period. 
Denote the mean  and the variance of $X_0$ by $\mu_0=\mathbb{E}[X_0]$ and $\sigma_0^2={\rm Var}(X_0)$.
 Further,  define  
 random  vector $\boldsymbol{\xi}$ by
 $\boldsymbol{\xi}=
    (\boldsymbol{X}, X_0)^\top, $  and denote the mean (loss) vector and  covariance matrix of $\boldsymbol{\xi}$ by
$\bm{\mu}=\mathbb{E}[\bm{\xi}]=(\boldsymbol{\mu}_n, \mu_0)^\top 
$
and
\begin{equation}
\boldsymbol{\Sigma}=\E[(\boldsymbol{\xi} -  \mathbb{E}[\bm{\xi}]) (\boldsymbol{\xi} -  \mathbb{E}[\bm{\xi}])^\top] 
   =\begin{bmatrix}
        \boldsymbol{\Sigma}_n & \boldsymbol{b}\\
        \boldsymbol{b}^{\top} & \sigma_0^2
    \end{bmatrix}\in\mathcal{S}^{+}_{n+1},
\end{equation}
where  $ \boldsymbol{\mu}_n=\E[\bm{X}]$ and  $\boldsymbol{\Sigma}_n =\E[(\boldsymbol{X} -  \mathbb{E}[\bm{X}]) (\boldsymbol{X} -  \mathbb{E}[\bm{X}])^\top]  $ are the $n$-dimensional mean vector and 
the
$n\times n$ covariance    matrix  of 
$\bm{X}$, 
and $\boldsymbol{b} = \E[(\boldsymbol{X} -  \mathbb{E}[\bm{X}]) (X_0 -  \E[X_0])] $ 
is  an $n$-dimensional column vector representing the covariance between $\bm{X}$ and  $X_0$, and   $\mathcal{S}^{+}_{n+1}$ denotes the set of $(n+1)$-dimensional 
symmetric positive definite matrices throughout this paper.
 
In an EIT problem, the investor aims to construct a portfolio that is expected to  achieve a 
higher    return than the index, that is, to outperform the benchmark index return $-X_0$ 
 by a pre-specified excess return  $-\alpha$;   see, for example,  \citet{guastaroba2016linear}. 
In other words, the investor seeks to select a portfolio that beats the (random)   target  return $-X_{\alpha} \triangleq  -X_0- \alpha$. 

We point out that,  in our paper, positive values of loss random variables represent  losses,  whereas  
negative values represent   returns or profits.  
Hence,  $(-\boldsymbol{\omega}^{\top}\bm{X} -(-X_\alpha))_+ = (\boldsymbol{\omega}^{\top}\bm{X}-X_\alpha)_-$ 
represents the excess return over the  target return $-X_\alpha$.  The larger the excess  return 
$(\boldsymbol{\omega}^{\top}\bm{X} -X_\alpha)_-$, the higher  the portfolio return.  Here, 
  $(x)_+=\max\{x, 0\}$ and 
$(x)_- =\max\{-x, 0\}$. It holds that $(x)_+=(-x)_-$ for any $x \in \R$. 
    On the other hand,  $(\boldsymbol{\omega}^{\top}\bm{X} -X_{\alpha})_+$ represents 
the excess  loss relative to the loss threshold $X_\alpha$. The larger 
 the excess loss  $(\boldsymbol{\omega}^{\top}\bm{X} -X_\alpha)_+$, the worse   the portfolio performance. 
Therefore, it is of interest to minimize a   
 quantification  of $(\boldsymbol{\omega}^{\top}\bm{X} -X_\alpha)_+$, which is  often referred to  as 
  the downside risk of the 
portfolio.

  A common measure of $(\boldsymbol{\omega}^{\top}\bm{X} -X_\alpha)_+$ is
   $\E[(\boldsymbol{\omega}^{\top}\bm{X} -X_\alpha)_+^m]$, which corresponds to 
    the 
$m$th  upper  partial  moment  of the portfolio loss 
 $\boldsymbol{\omega}^{\top}\bm{X}$ above  the loss threshold   $X_\alpha$. In particular,  the 
first-order upper  partial   moment  $\E[(\boldsymbol{\omega}^{\top}\bm{X} -X_\alpha)_+]$, referred to as the expected regret (ER), and the  
 second-order   upper partial moment  $\E[(\boldsymbol{\omega}^{\top}\bm{X} -X_\alpha)_+^2]$,  
referred to as the target semi-variance (TSV), have been widely  
employed in optimal portfolio selection problems. See, for example, 
 \citet{fishburn1977mean},    \citet{testuri2004relation},  \citet{siegmann2005discrete}, 
\citet{farinelli2008sharpe}, \cite{chen2011tight}, 
 \cite{KLMY17}, 
     \citet{LCLX19},       \citet{LSW20},  
  \citet{YHLH21}, \cite{LMYZ22}, \cite{CJM2025a}, and references therein. Specifically, 
 the following optimal  EIT portfolio selection model has appeared in the literature: 
 \begin{equation}
    \begin{aligned}
    \label{min-upper-moments}
        &
    \min_{\boldsymbol{\omega} \in \mathbb{R}^n}  \,  \ \mathbb{E}[(\boldsymbol{\omega}^{\top}\bm{X}-X_\alpha)_+^m] \\
    & \ \text{s.t.} \quad \boldsymbol{\omega}^{\top}\boldsymbol{e}=1.
    \end{aligned}
\end{equation}

 Model \eqref{min-upper-moments} focuses on  minimizing the downside risk 
  of the portfolio relative to the index. Minimizing  the downside 
risk  alone does not explicitly  account for the optimization of  portfolio returns. Conversely, if the objective is solely to maximize the portfolio return, which is equivalent to minimizing the portfolio loss,  
 then  the downside risk   will be ignored.  
  
  Indeed, there exists a trade-off between
 minimizing the portfolio loss  and minimizing  the downside risk. To balance this trade-off, a common 
 approach is to minimize a convex  combination of the portfolio loss and the downside risk. This  leads to the
  following optimal EIT  portfolio selection  model:
   \begin{equation}
   \label{Pareto-model-upper-W}
    \begin{aligned}
        &
    \min_{\boldsymbol{\omega} \in \cal W}  \,  (1-\lambda) \ \mathbb{E}[\bm{\omega}^{\top}\bm{X}]+\lambda \ 
    \mathbb{E}[(\boldsymbol{\omega}^{\top}\bm{X}-X_\alpha)_+^m],
  \end{aligned}
\end{equation}
 where $\lambda \in (0, 1)$ and 
 ${\cal W}$
  denotes the set of feasible portfolios under consideration by   the investor. 
  
  Moreover, model \eqref{min-upper-moments} is primarily motivated by heuristic considerations. In contrast, model 
   \eqref{Pareto-model-upper-W} can be derived from the theory of Pareto optimality. In particular, an optimal solution 
   $\boldsymbol{\omega}^*_\lambda$ to problem  \eqref{Pareto-model-upper-W} is a Pareto-optimal portfolio, in the sense that the portfolio loss cannot be further minimized without worsening  the downside risk. 
 A formal definition of Pareto-optimal portfolios for optimal EIT portfolio selection problems is provided below.  In particular, an optimal solution $\bm{\omega}^*_\lambda$ to problem \eqref{Pareto-model-upper-W} is Pareto-optimal in the sense that the portfolio loss cannot be further reduced without increasing the downside risk. A formal definition of Pareto optimality in the context of EIT portfolio selection is provided below.

  Let $\rho=\rho(\boldsymbol{\omega}, \bm{X}, X_0)$ and 
  $\varphi=\varphi(\boldsymbol{\omega}, \bm{X}, X_0)$ denote two objective functions of interest 
  to  an investor. For example, $\rho$ and $\varphi$ may represent, respectively, upper partial moments relative to a loss threshold,  
   the expected portfolio loss, the variance of portfolio loss, or 
   quantile-based risk measures such as value-at-risk (VaR) and  conditional value-at-risk (CVaR).

   Let $\cal W$ denote the set of feasible portfolios under consideration by   the investor.    Suppose that
     under the objective function $\rho$, the investor prefers 
     $\boldsymbol{\omega}_1 \in \cal W $ over $\boldsymbol{\omega}_2 \in \cal W$  if 
   $ \rho(\boldsymbol{\omega}_1, \bm{X}, X_0) \le \rho(\boldsymbol{\omega}_2, \bm{X}, X_0)$, 
    and under the objective 
   function $\varphi$, the investor prefers $\boldsymbol{\omega}_1 \in \cal W$ over $\boldsymbol{\omega}_2 \in \cal W$  if 
   $ \varphi(\boldsymbol{\omega}_1, \bm{X}, X_0) \le \varphi(\boldsymbol{\omega}_2, \bm{X}, X_0)$. 
   That is, the investor prefers smaller values of both objective functions and seeks an optimal portfolio $\boldsymbol{\omega}^*$ that minimizes 
    both objectives.
    In many applications, however, the two objective functions may exhibit a trade-off,  typically between risk and return. For example, between mean and variance, mean and VaR, mean and CVaR, or mean and upper partial moments. Consequently, it is generally not possible to minimize both objectives simultaneously. 
    
    To address this trade-off, a common heuristic is to minimize one objective function subject to a constraint on the other, leading to the following optimization problem:
      \be
   \label{non-Pareto-rho}
    \min_{\bm{\omega} \in {\cal W}_{\nu}} \ \rho(\boldsymbol{\omega}, \bm{X}, X_0) 
   \ee
   where
   $
   {\cal W}_\nu =\{ \bm{\omega} \in \mathbb{R}^n :  \varphi(\boldsymbol{\omega}, \bm{X}, X_0) \le \nu, \  
   \boldsymbol{\omega}^{\top}\boldsymbol{e}=1\}.  
     $
     A classical example of model \eqref{non-Pareto-rho} is the following 
      Markowitz model:     
    \begin{equation}
   \label{V-model-0}
\begin{aligned}    
    \min_{\bm{\omega}\in\mathbb{R}^n} \quad & 
    {\rm Var}(\bm{\omega}^{\top} \bm{X}) \\
    \, s.t.    \quad \ &  \E[- \bm{\omega}^{\top} \bm{X}]
     \ge r_0, \ \ \boldsymbol{\omega}^{\top}\boldsymbol{e}=1, 
            \end{aligned}
\end{equation}
where $r_0 \in \mathbb{R}$ represents the minimum target expected return of the portfolio. Models of the form \eqref{non-Pareto-rho}, which minimize one objective while imposing a constraint on another, have been widely used in optimal portfolio selection.

Another approach, based on the theory of Pareto optimality, can also be used to address this trade-off. In this framework, the problem is formulated by  minimizing a convex  combination of the two objective functions. This leads to the following model:  
\be
\label{Pareto-model-lambda}
  \min_{\bm{\omega} \in {\cal W}} \ (1-\lambda) 
  \rho(\boldsymbol{\omega}, \bm{X}, X_0) +\lambda \varphi(\boldsymbol{\omega}, \bm{X}, X_0)
   \ee        
  for some $\lambda \in (0, 1)$, or equivalently, 
   \be
   \label{Pareto-model-kappa}
    \min_{\bm{\omega} \in {\cal W}} \ 
  \rho(\boldsymbol{\omega}, \bm{X}, X_0) +\kappa \,  \varphi(\boldsymbol{\omega}, \bm{X}, X_0)
  \ee
  for some $\kappa >0$. 
  A well-known example of model \eqref{Pareto-model-kappa} is the following variant   of the Markowitz  model:
    \begin{equation}
\label{M-V-model}
\begin{aligned}    
    \min_{\bm{\omega}\in\mathbb{R}^n} &  \ \  \E[\bm{\omega}^{\top} \bm{X}]     + \kappa \,  {\rm Var} (\bm{\omega}^{\top} \bm{X})   \\
         s.t.  &  \ \   \E[- \bm{\omega}^{\top} \bm{X}]
      \ge r_0, \ \ \  \boldsymbol{\omega}^{\top}\boldsymbol{e}=1.
      \end{aligned}
\end{equation}
Here,  $\kappa > 0$ is usually expressed as $\kappa =\frac{\gamma}{2}$, where  $\gamma>0$  denotes  
the risk-aversion coefficient.

Model \eqref{M-V-model} balances portfolio loss and variance by minimizing a convex combination of the two, subject to the constraint that the expected return of the portfolio is at least $r_0$. Equivalently, the model minimizes portfolio loss while penalizing variance, subject to a minimum expected return requirement.
  Models of the form \eqref{Pareto-model-kappa}, which minimize a convex combination of two objective functions while imposing constraints on portfolio return or other quantities of interest to the investor, 
such as VaR or CVaR, have also been widely used in optimal portfolio selection.

     \begin{definition}   Consider a portfolio selection problem characterized by $(\bm{X}, X_0, \rho, \varphi, \boldsymbol{\omega}, \cal W)$, where  $\rho$ and 
    $ \varphi$ are two objective functions. A portfolio $\bm{\omega}^* \in \cal W$ is said to be Pareto-optimal under 
    this    setting if there does not exist  $ \bm{\omega} \in {\cal W}$  such that 
    $ \rho(\boldsymbol{\omega}, \bm{X}, X_0) \le \rho(\boldsymbol{\omega}^*, \bm{X}, X_0)$ and 
     $\varphi(\boldsymbol{\omega}, \bm{X}, X_0) \le \varphi(\boldsymbol{\omega}^*, \bm{X}, X_0),$
    with  at least one of the inequalities  being  strict.  \hfill $\Box$  
        \end{definition}

 \begin{proposition} 
 \label{pro-Pareto}   
 Under a  portfolio selection setting
  $(\bm{X}, X_0, \rho, \varphi, \boldsymbol{\omega}, \cal W)$, if  
  $$
  \bm{\omega}^*_\lambda \in  \argmin_{\bm{\omega} \in {\cal W}} \ (1-\lambda) 
  \rho(\boldsymbol{\omega}, \bm{X}, X_0) +\lambda \varphi(\boldsymbol{\omega}, \bm{X}, X_0)
   $$      
  for some $\lambda \in (0, 1)$, or equivalently, if 
   $$
  \bm{\omega}^*_\kappa \in  \argmin_{\bm{\omega} \in {\cal W}} \ 
  \rho(\boldsymbol{\omega}, \bm{X}, X_0) +\kappa \,  \varphi(\boldsymbol{\omega}, \bm{X}, X_0)
  $$
  for some $\kappa >0$,  
   then  $\bm{\omega}^*_\lambda$ or  $\bm{\omega}^*_\kappa$ is a Pareto-optimal portfolio
   under this  setting.  
    \end{proposition}
      
 According to Proposition \ref{pro-Pareto}, to find a Pareto-optimal portfolio, it suffices to solve the following optimization problem:  
   \be
   \label{Pareto-kappa}
   \min_{\bm{\omega} \in {\cal W}} \ 
  \rho(\boldsymbol{\omega}, \bm{X}, X_0) +\kappa \, \varphi(\boldsymbol{\omega}, \bm{X}, X_0)
   \ee  
   for $\kappa >0$. 
   
   Note that solving problem \eqref{Pareto-kappa} usually requires knowledge of the joint distribution of $(\bm{X}, X_0)$. In this context, it is often assumed that $(\bm{X}, X_0)$ follows a known joint distribution, such as a multivariate normal or a multivariate $t$ distribution. 
 However, in practice, this joint distribution is generally unknown or uncertain, and only partial information is available. One approach to address this issue is to compute $\rho(\boldsymbol{\omega}, \bm{X}, X_0)$ and $\varphi(\boldsymbol{\omega}, \bm{X}, X_0)$ in problem \eqref{Pareto-kappa} using empirical estimates. Alternatively, one may consider the worst-case values of these objective functions to address distributional uncertainty. This leads to the following robust Pareto-optimal EIT portfolio selection model:   
   \be
   \label{Pareto-kappa-rb}
   \min_{\bm{\omega} \in {\cal W}} \ 
  \sup_{(G,F)\in \mathcal{D}_{n+1}} \,  \rho(\boldsymbol{\omega}, \bm{X}^G, X^F_0) +\kappa \,  \sup_{(G,F)\in \mathcal{D}_{n+1}} \,  
   \varphi(\boldsymbol{\omega}, \bm{X}^G, X^F_0).
   \ee
   Here, $\kappa > 0$, and $(G, F)$ denotes the joint distribution of $(\bm{X}, X_0)$, 
   while $\mathcal{D}_{n+1}$ is an $(n+1)$-dimensional uncertainty set of distributions. For any random vector $\bm{X}$ and random variable $X_0$, the notation $\bm{X}^G$ and $X_0^F$ indicates that $\bm{X}$ and $X_0$ have distributions $G$ and $F$, respectively.
    This model aims to minimize a convex  combination of the worst-case values of the two objective functions. It can be expressed as follows:  
   \be
   \label{Pareto-kappa-rb}
   \min_{\bm{\omega} \in {\cal W}} \   \rho_{wc}(\boldsymbol{\omega}, \mathcal{D}_{n+1}) +\kappa \,     \varphi_{wc}(\boldsymbol{\omega}, \mathcal{D}_{n+1}),
   \ee
   where
   $
   \rho_{wc}(\boldsymbol{\omega}, \mathcal{D}_{n+1}) =  \sup_{(G,F)\in \mathcal{D}_{n+1}} \,  \rho(\boldsymbol{\omega}, \bm{X}^G, X^F_0)
     $
   and $
      \varphi_{wc}(\boldsymbol{\omega}, \mathcal{D}_{n+1}) =  \sup_{(G,F)\in \mathcal{D}_{n+1}} \,  \varphi(\boldsymbol{\omega}, \bm{X}^G, X^F_0)
      $
  represent   the worst-case values of the objective functions $\rho$ and $\varphi$, respectively,  over the uncertainty set  $\mathcal{D}_{n+1}$. Similar to Proposition \ref{pro-Pareto}, we have the following result:         
         
 \begin{proposition} 
 \label{pro-Pareto-rb}   
 Under a  portfolio selection setting
  $(\bm{X}, X_0, \rho, \varphi, \boldsymbol{\omega}, {\cal W}, \mathcal{D}_{n+1})$, if 
  $$
  \bm{\omega}^*_\kappa \in  \argmin_{\bm{\omega} \in {\cal W}} \  
  \rho_{wc}(\boldsymbol{\omega}, \mathcal{D}_{n+1}) +\kappa \, \varphi_{wc}(\boldsymbol{\omega}, \mathcal{D}_{n+1})
   $$      
  for some $\kappa > 0$,   then  $\bm{\omega}^*_\kappa$ is a Pareto-optimal portfolio
   under this   setting.   
    \end{proposition}

 By this proposition, we see that the optimal portfolios derived from \eqref{M-V-model} are Pareto-optimal 
 under  the  setting
  $\big (\bm{X}, \, X_0, \,  {\rm Var}(\bm{\omega}^\top X), \, \E[\bm{\omega}^\top X],  \, 
   \boldsymbol{\omega}, \, {\cal W},  \, {\cal M}_{n+1}(\bm{\mu},\bm{\Sigma}) \big )$, 
  where ${\cal M}_{n+1}(\bm{\mu},\bm{\Sigma})$, defined in   \eqref{set-M_n+1},  is the 
  uncertainty  set of joint 
  distributions  with a given mean vector $ \bm{\mu}$   and covariance  matrix $\bm{\Sigma}$.

Portfolio return is one of the primary considerations in investment. Imposing a constraint on the worst-case portfolio return in robust 
Pareto-optimal portfolio models, when the return distributions are uncertain, is of particular interest, as it has the potential to yield an enhanced portfolio  performance compared to models without this constraint.   In this paper, motivated by the above discussion, we investigate the following optimal robust  EIT portfolio selection model:
     \begin{equation}
    \begin{aligned}
    \label{mean-upper-m-rb-k}
        &
    \min_{\boldsymbol{\omega} \in {\cal W}_{r}}  \, 
    \sup_{(G,F)\in \mathcal{D}_{n+1}} \, \mathbb{E}[\bm{\omega}^{\top}\bm{X}^G]+ \kappa \,  \    \sup_{(G,F)\in \mathcal{D}_{n+1}} \,  \mathbb{E}[(\boldsymbol{\omega}^{\top}\bm{X}^G-X_0^F-\alpha)_+^m],\\
  \end{aligned}
\end{equation} 
 where $\kappa>0$ and 
 \by
 {\cal W}_{r} = \{ \boldsymbol{\omega}: \boldsymbol{\omega} \in {\cal W} \subset \R^n, \   \boldsymbol{\omega}^{\top}\boldsymbol{e}=1, \ \   \inf_{(G, F) \in {\cal D}_{n+1}} \mathbb{E}[-\bm{\omega}^{\top}\bm{X}^G]   \ge r\}. 
 \ey
 In particular, if ${\cal W}  = \R^n$, the feasible set ${\cal W}_{r}$ for $\boldsymbol{\omega}$ allows short selling, whereas if 
${\cal W} =A^n$ with $A= [0, \infty)$, the feasible set ${\cal W}_{r}$ restricts $\boldsymbol{\omega}$ to non-negative values, prohibiting short selling.  
In addition, additional constraints can be imposed on $\boldsymbol{\omega}$ within the feasible set ${\cal W}_{r}$ or $\cal W$.

This  constraint $\inf_{(G, F) \in {\cal D}_{n+1}} \mathbb{E}[-\boldsymbol{\omega}^{\top}\boldsymbol{X}^G]\geq r$ 
  ensures that the worst-case expected  portfolio return is at least  $r$. Such a constraint has been 
  widely used in optimal  portfolio selection.  See, for example,  \cite{kang2019data}, \cite{blanchet2022distributionally}, 
   \cite{CJM2025a}, and references therein.

 To explore the impact of EIT portfolio  selection 
modeling on portfolio  performance, we also consider the following robust EIT portfolio selection model, 
which   minimizes 
the worst-case downside risk subject to a constraint on the worst-case portfolio return. Specifically, 
we consider 
    \begin{equation}
    \begin{aligned}
    \label{upper-m-rb-k}
        &
    \min_{\boldsymbol{\omega} \in {\cal W}_{r}}  \,   \sup_{(G,F)\in \mathcal{D}_{n+1}} \,  \mathbb{E}[(\boldsymbol{\omega}^{\top}\bm{X}^G-X_0^F-\alpha)_+^m]. 
  \end{aligned}
\end{equation}

It is also worthwhile to consider the robust versions of the Markowitz model \eqref{V-model-0} and its variant \eqref{M-V-model}, which can be formulated as follows:
  \begin{equation}
\label{V-model-D_n+1}
    \min_{\boldsymbol{\omega} \in {\cal W}_r}  \, \  \sup_{(G,F)\in {\cal D}_{n+1}} \, 
    {\rm Var}(\bm{\omega}^{\top}\bm{X}^G) ,
    \end{equation}
and
 \begin{equation}
\label{MV-model-D_n+1}
    \min_{\boldsymbol{\omega} \in {\cal W}_r}  \, \  \sup_{(G,F)\in {\cal D}_{n+1}} \, \mathbb{E}[\bm{\omega}^{\top}\bm{X}^G]+\kappa \,  \sup_{(G,F)\in {\cal D}_{n+1}} \,   
    {\rm Var}(\boldsymbol{\omega}^{\top}\bm{X}^G) .
    \end{equation}

We use the robust portfolio selection models \eqref{V-model-D_n+1} and \eqref{MV-model-D_n+1}, which do not include index tracking, as benchmark models when evaluating the performance of portfolios derived from the robust EIT portfolio selection models \eqref{mean-upper-m-rb-k} and \eqref{upper-m-rb-k}. These models also provide a useful framework for investigating the effects of distributional uncertainty on the classical Markowitz model \eqref{V-model-0} and its variant \eqref{M-V-model}.
     The robust Markowitz model \eqref{V-model-D_n+1} has been extensively studied in the literature. For example, \cite{blanchet2022distributionally} considers this model, where the uncertainty set for the joint distribution of asset returns is defined as a Wasserstein ball centered at the empirical distribution.
     
Our empirical study shows that portfolios derived from models minimizing a convex combination of the worst-case portfolio loss and downside risk 
 outperform those derived from models that minimize 
  the worst-case downside risk alone, in terms of the out-of-sample  cumulative wealth,  
  Sharpe ratio, Sortino ratio,  and Omega ratio. We also discuss the choice of model parameters and empirically illustrate their impact on portfolio performance.   
 
In this paper, we investigate the optimal  robust EIT portfolio selection problems \eqref{mean-upper-m-rb-k} and \eqref{upper-m-rb-k}, as well as the robust mean-variance problems without index tracking \eqref{V-model-D_n+1} and \eqref{MV-model-D_n+1}. We focus on two uncertainty sets: one is the classical set of distributions with a given mean vector and covariance matrix, and the other is a generalized mean-covariance-based set. 
 
  \section{Robust optimal EIT portfolio  under a given mean  vector and covariance matrix}
  \label{sec-3}

In this section, we assume that the joint distribution of the assets and index returns is unknown, but its mean vector and covariance matrix are given. That is, the possible joint distributions of $(\bm{X}, X_0)$ belong to the following uncertainty set:
  \begin{equation}
\label{set-M_n+1}
   \begin{aligned}
    \mathcal{M}_{n+1}(\bm{\mu},\bm{\Sigma})=\big\{
    &(G, F)\in\mathcal{F}(\mathbb{R}^{n+1}): \ \mathbb{E}[\boldsymbol{X}^G]=\boldsymbol{\mu}_n, 
    \,  \mathbb{E}[X_0^F]=\mu_0,  \, 
     \E[(\boldsymbol{\xi} -  \mathbb{E}[\bm{\xi}]) (\boldsymbol{\xi} -  \mathbb{E}[\bm{\xi}])^\top] = 
    \boldsymbol{\Sigma} \succ 0   
   \big\},
    \end{aligned}   
\end{equation}
where 
 $\mathcal{F}(\mathbb{R}^{n+1})$ is the set of $(n+1)$-dimensional distributions defined on $\mathbb{R}^{n+1}$.  
 
 It is a classical uncertainty set that has been widely used in robust portfolio selection problems, as well as in various applications in 
 finance and operations research  (see, for example, \cite{07PI}, \cite{chen2011tight}, and references therein).
 
In this section, we solve problems \eqref{mean-upper-m-rb-k} and \eqref{upper-m-rb-k} 
for    ${\cal D}_{n+1} = \mathcal{M}_{n+1}(\boldsymbol{\mu}, \boldsymbol{\Sigma})$ and ${\cal W}=\R^n$.  In this case, for  any $(G,F)\in\mathcal{M}_{n+1}(\bm{\mu},\bm{\Sigma})$, we have that $\mathbb{E}[\bm{\omega}^{\top}\bm{X}^G]=\bm{\omega}^{\top}\bm{\mu}_n$.
 Thus,  when ${\cal D}_{n+1} =   \mathcal{M}_{n+1}(\bm{\mu},\bm{\Sigma})$ and   ${\cal W} =\R^n$,
  problems  \eqref{mean-upper-m-rb-k} and \eqref{upper-m-rb-k} 
    reduce to 
  the following problems: 
      \begin{equation}
    \begin{aligned}
    \label{mean-upm-M}
        &
    \min_{\boldsymbol{\omega} \in \mathbb{R}^n}  \, \bm{\omega}^{\top}\bm{\mu}_n + \kappa 
     \sup_{(G,F)\in \mathcal{M}_{n+1}(\bm{\mu},\bm{\Sigma})} \,  \mathbb{E}[(\boldsymbol{\omega}^{\top}\bm{X}^G-X_0^F-\alpha)_+^m] \\
    & \ \text{s.t.}  \quad  \ -\bm{\omega}^{\top}\bm{\mu}_n  \ge r_0, \ \ \ \ \boldsymbol{\omega}^{\top}\boldsymbol{e}=1,
     \end{aligned}
\end{equation}
and 
    \begin{equation}
    \begin{aligned}
    \label{upm-M}
        &
    \min_{\bm{\omega} \in \mathbb{R}^n}  \,  
     \sup_{(G,F)\in \mathcal{M}_{n+1}(\bm{\mu},\bm{\Sigma})} \,  \mathbb{E}[(\bm{\omega}^{\top}\bm{X}^G-X_0^F-\alpha)_+^m] \\
    & \ \text{s.t.}  \quad  \ -\bm{\omega}^{\top}\bm{\mu}_n  \ge r_0, \ \ \ \ \bm{\omega}^{\top}\bm{e}=1.   \end{aligned}
\end{equation}   
   
   We derive closed-form solutions to these two robust EIT portfolio selection problems. To this end, we 
   introduce an uncertainty set of distributions with a given mean and variance and present two lemmas that 
   will be used throughout the paper.

  Let $X$ be a random variable. Suppose that only the first two moments of $X$ are known,  
  while   the complete distribution of $X$ is uncertain. Define 
 \begin{equation}
    \begin{aligned}\label{MV uncertainty set}
    \mathcal{L}(\mu,\sigma)
    &=\left\{F\in\mathcal{F}(\mathbb{R}): \ \mathbb{E}[X^F]=\mu, \ \mathbb{E}[(X^F)^2]=\mu^2+\sigma^2\right\},
    \end{aligned}
\end{equation}
where $\mathcal{F}(\mathbb{R})$ denotes  
 the set of all distributions defined on $\mathbb{R}$. That is, $\mathcal{L}(\mu,\sigma)$ consists of all  distributions with mean $\mu$ and variance $\sigma^2$.

The uncertainty set $\mathcal{L}(\mu,\sigma)$, along with other moment-based uncertainty sets, has been widely used in distributionally robust optimization (DRO) problems in insurance, finance, and operations research; 
 see, for example, \cite{jagannathan1977minimax},   \cite{07PI}, \cite{chen2011tight},   \cite{ZY25},  and references therein.

\begin{lemma}
\label{lemma-jan-1977}
    For any $\mu, d\in\mathbb{R}$ and $\sigma\in\mathbb{R}^+=(0,\infty)$, 
    it holds that  
    \begin{align}
        \sup_{F\in\mathcal{L}(\mu,\sigma)}\mathbb{E}[(X^F-d)_+] &= \frac{1}{2}\big(\mu-d+\sqrt{\sigma^2+(\mu-d)^2} \, \big), \label{lemma-m1} \\
     \sup_{F\in\mathcal{L}(\mu,\sigma)}\mathbb{E}[(X^F-d)_+^2] &=(\mu-d)_{+}^2+\sigma^2. \label{lemma-m2}
            \end{align}
\end{lemma}
\begin{proof}
The proofs  of equations \eqref{lemma-m1} and \eqref{lemma-m2}  are given in  \cite{jagannathan1977minimax} 
   and  \citet{chen2011tight}, respectively. 
\end{proof}


To simplify the presentation of  the results,  
we  introduce the following parameters throughout the paper, namely $\alpha^*$, $u$, $c_i$ for $i=0,1,\dots,5$, and $b_i$ for $i=0,1,2$:
\begin{equation*}
   \alpha^*:=\alpha+\mu_0, \ \ \ \      u := (\bm{e}^{\top}\bm{\Sigma}_n^{-1}\bm{e})(\bm{\mu}_n^{\top}\bm{\Sigma}_n^{-1}\bm{\mu}_n)-(\bm{\bm{e}^{\top}\bm{\Sigma}}_n^{-1}\bm{\mu}_n)^2, 
\end{equation*}
\begin{equation*}
\begin{alignedat}{3}
c_0 &:= \bm{e}^{\top}\bm{\Sigma}_n^{-1}\bm{e}, \qquad &
c_1 &:= \bm{e}^{\top}\bm{\Sigma}_n^{-1}\bm{\mu}_n, \qquad &
c_2 &:= \bm{\mu}_n^{\top}\bm{\Sigma}_n^{-1}\bm{\mu}_n,  \\
c_3 &:= \bm{b}^{\top}\bm{\Sigma}_n^{-1}\bm{\mu}_n, \qquad &
c_4 &:= \bm{b}^{\top}\bm{\Sigma}_n^{-1}\bm{e}, \qquad &
c_5 &:= \bm{b}^{\top}\bm{\Sigma}_n^{-1}\bm{b}, \\
    b_0 &:= \frac{\bm{e}^{\top}\bm{\Sigma}_n^{-1}\bm{e}}{u}, &  \ \  b_1&:= \frac{\bm{\bm{e}^{\top}\bm{\Sigma}}_n^{-1}\bm{\mu}_n}{u},&  
    b_2 &:= \frac{\bm{\mu}_n^{\top}\bm{\Sigma}_n^{-1}\bm{\mu}_n}{u}.
\end{alignedat}
\end{equation*}
We further assume that the mean vector $\bm{\mu}_n$ and the  vector $\bm{e}$ are linearly independent. 
This  assumption ensures  that $u>0$.

\subsection{Minimization of a convex combination of portfolio loss and downside risk}

In this subsection,   we solve problem \eqref{mean-upm-M}. 
To this end,  we first prove the following lemma. 
\begin{lemma}
    \label{theo-w*} For any $s \in \R$, 
    the following optimization problem     \begin{equation}
        \begin{aligned}
           & \min_{\omega\in\mathbb{R}^n} \quad \bm{\omega}^{\top}\bm{\Sigma}_n\bm{\omega}-2\bm{b}^{\top}\bm{\omega}+\sigma_0^2 \\
           & \ \mbox{s.t.}    \quad   \boldsymbol{\omega}^{\top}\boldsymbol{\mu}_n= s, \ \   \   \boldsymbol{\omega}^{\top}\boldsymbol{e}=1, 
        \end{aligned}
    \end{equation} 
    has a  unique solution $\bm{\omega}_s$,   given  by
    \begin{equation}
    \begin{aligned}
    \label{w_s} 
        \boldsymbol{\omega}_s
        &=\boldsymbol{\Sigma}_n^{-1}\boldsymbol{b}-\begin{pmatrix}
            \boldsymbol{\Sigma}_n^{-1}\boldsymbol{\mu}_n & \boldsymbol{\Sigma}_n^{-1}\boldsymbol{e}
        \end{pmatrix}\begin{pmatrix}
            b_0 & -b_1\\-b_1 & b_2
        \end{pmatrix}\begin{pmatrix}
            c_3-s\\ c_4-1
        \end{pmatrix}. 
    \end{aligned}
\end{equation}
\end{lemma}

\begin{proof} The proof of Lemma \ref{theo-w*} follows from a standard Lagrangian 
approach, so the details are omitted.
  \end{proof}

Further, to solve problem \eqref{mean-upm-M}, we define, 
throughout this paper, the  functions $f_m(\mu, \sigma)$ and  $h_m(\mu, \sigma)$ for any $\mu \in \R$ and $\sigma \in \R^+$ as
\be
       f_m(\mu,\sigma)
        &=&
        \sup_{F\in \mathcal{L}(\mu, \sigma)} \mathbb{E}[(X^F-\alpha)_{+}^m],      \label{f-mu-sigma}  \\
 h_{m,\kappa}(\mu,\sigma)
&= &\mu+\kappa\  f_m(\mu,\sigma), \label{h_m_k}
\ee
where $\kappa >0$.  Closed-form expressions for   the function $f_m(\mu,\sigma)$ are provided in \eqref{lemma-m1} and \eqref{lemma-m2} for $m=1$ and $m=2$. 

\begin{proposition} 
\label{pro-f-h}
The functions $  f_m(\mu,\sigma)$ and  $ h_{m,\kappa}(\mu,\sigma)$ for $m=1$ and $m=2$  have the following properties. 
  \bn[(i)] 
\item  Both   $f_m(\mu,\sigma)$ and  $ h_{m,\kappa}(\mu,\sigma)$  are increasing and continuous  in $\mu \in \R$ for fixed $\sigma \in \R^+$, 
 and  increasing and continuous in $\sigma \in \R^+$ for fixed $\mu \in \R$.
    
 \item  Both $f_m(\mu,\sigma)$  and  $h_{m,k}(\mu,\sigma)$  are 
  convex  on $(\mu, \sigma) \in  \R \times \R^+$. 
     \en
        \end{proposition}
    \begin{proof}
The proof of this proposition follows directly from the expressions \eqref{lemma-m1} and \eqref{lemma-m2}, so the details are omitted.
      \end{proof} 
  
\begin{theorem}
\label{theo-mean-upm}  
     \begin{enumerate}[(i)]
   \item    Consider the robust  EIT  portfolio selection  problem \eqref{mean-upm-M} for  $m=1$.
    \begin{enumerate}[(a)]
        \item 
    If $\kappa>\kappa_0 = \frac{2+2\sqrt{b_0+1}}{b_0}$, then problem \eqref{mean-upm-M}   admits a unique  solution $\bm{\omega}_{1,\kappa}^*$,   given by
    \begin{equation}
    \begin{aligned}\label{w1*}
        \bm{\omega}_{1, \kappa}^*
        &=\bm{\Sigma}_n^{-1}\bm{b}-\begin{pmatrix}
            \bm{\Sigma}_n^{-1}\bm{\mu}_n & \bm{\Sigma}_n^{-1}\bm{e}
        \end{pmatrix}\begin{pmatrix}
            b_0 & -b_1\\-b_1 & b_2
        \end{pmatrix}\begin{pmatrix}
       c_3-\min\{-r_0, s_{1,\kappa}\}  \\ c_4-1
        \end{pmatrix},
    \end{aligned}
\end{equation}
where 
\by
  s_{1,\kappa} & = &  \frac{c_3b_0-b_1(c_4-1)+\alpha^*}{b_0+1}-\frac{\kappa+2}{b_0+1}\sqrt{\frac{q}{b_0\kappa^2-4(\kappa+1)}}, \\
q &=& (b_0b_2-b_1^2+b_2)(c_4-1)^2+(c_3-\alpha^*)\big(b_0(c_3-\alpha^*)-2b_1(c_4-1)\big)
-(b_0+1)(c_5-\sigma_0^2).
\ey

\item If $\kappa\leq \frac{2+2\sqrt{b_0+1}}{b_0}$, then   problem \eqref{mean-upm-M}  has no  solution.
\end{enumerate}

 \item  Consider the robust  EIT  portfolio selection  problem \eqref{mean-upm-M} for   $m=2$. Then, 
for any $\kappa>0$,    problem \eqref{mean-upm-M}  has a unique solution    $\bm{\omega}_{2,\kappa}^*$,   given by
\begin{equation}
    \begin{aligned} \label{w2*} 
        \bm{\omega}_{2,\kappa}^*
        &=\bm{\Sigma}_n^{-1}\bm{b}-\begin{pmatrix}
            \bm{\Sigma}_n^{-1}\bm{\mu}_n & \bm{\Sigma}_n^{-1}\bm{e}
        \end{pmatrix}\begin{pmatrix}
            b_0 & -b_1\\-b_1 & b_2
        \end{pmatrix}\begin{pmatrix}
         c_3-\min\{-r_0,  s_{2, \kappa}\}  \\ c_4-1
        \end{pmatrix},
    \end{aligned}
\end{equation}
where 
\by
s_{2,\kappa}=  
s_\kappa -\frac{(s_\kappa-\alpha^*)_+}{b_0+1}
\ \ \ 
\mbox{and} \ \ \ 
s_\kappa = \frac{2b_0c_3\kappa-2(c_4-1)b_1\kappa-1}{2b_0\kappa}.
 \ey
\end{enumerate}
\end{theorem}

\subsection{Minimization of downside risk alone}

In this subsection, we solve problem \eqref{upm-M} to compare models that minimize the trade-off with models that focus exclusively on downside risk, subject to a portfolio return constraint.

\begin{proposition}
\label{pro-upm-M-sol}
    \begin{enumerate}[(i)] 
        \item   The  robust   EIT  portfolio selection  problem \eqref{upm-M} for   $m=1$ has a unique solution  
        $\boldsymbol{\omega}_{RP_1}^*$,   given by 
            \begin{equation}
            \label{w*pr1}
    \begin{aligned}
        \boldsymbol{\omega}_{RP_1}^*
        &=\boldsymbol{\Sigma}_n^{-1}\boldsymbol{b}-\begin{pmatrix}
            \boldsymbol{\Sigma}_n^{-1}\boldsymbol{\mu}_n & \boldsymbol{\Sigma}_n^{-1}\boldsymbol{e}
        \end{pmatrix}\begin{pmatrix}
            b_0 & -b_1\\-b_1 & b_2
        \end{pmatrix}\begin{pmatrix}
            c_3-\min\left\{-r_0, \,  s_1 \right\}\\ c_4-1
        \end{pmatrix},
    \end{aligned}
\end{equation}
where 
\by
s_1 &=& \frac{c_3b_0-b_1(c_4-1)+\alpha^*}{b_0+1}-\frac{1}{\sqrt{b_0}}\frac{\sqrt{q}}{b_0+1}, \\
q &=& (b_0b_2-b_1^2+b_2)(c_4-1)^2+(c_3-\alpha^*)\big(b_0(c_3-\alpha^*)-2b_1(c_4-1)\big)-(b_0+1)(c_5-\sigma_0^2).
\ey

\item  The robust   EIT portfolio  selection problem \eqref{upm-M} for   $m=2$ has a unique solution   $\boldsymbol{\omega}_{RP_2}^*$,   given by 
\begin{equation}
\label{w*pr2}
    \begin{aligned}
        \boldsymbol{\omega}_{RP_2}^*
        &=\boldsymbol{\Sigma}_n^{-1}\boldsymbol{b}-\begin{pmatrix}
            \boldsymbol{\Sigma}_n^{-1}\boldsymbol{\mu}_n & \boldsymbol{\Sigma}_n^{-1}\boldsymbol{e}
        \end{pmatrix}\begin{pmatrix}
            b_0 & -b_1\\-b_1 & b_2
        \end{pmatrix}\begin{pmatrix}
            c_3-\min \{-r_0, \,  s_2 \} \\ c_4-1
        \end{pmatrix},
    \end{aligned}
\end{equation}
where 
\by
s_2 = \frac{b_0c_3-b_1(c_4-1)}{b_0} - \frac{(b_0c_3-b_1c_4+b_1-b_0\alpha^*)_+}{b_0(b_0+1)}. 
\ey
\end{enumerate}
\end{proposition}

\section{Robust  EIT portfolio  under a generalized  
uncertainty set} \label{sec-4} 

In this section, we study problems \eqref{mean-upper-m-rb-k} and \eqref{upper-m-rb-k}, as well as \eqref{MV-model-D_n+1} and \eqref{V-model-D_n+1}, under a more general uncertainty set ${\cal D}_{n+1}$, compared with $\mathcal{M}_{n+1}(\bm{\mu}, \bm{\Sigma})$.
 In practice, the true mean loss vector and covariance matrix of $({\bm X}, X_0)$ are not observable, and practitioners typically rely on their empirical estimates, $\hat{\bm{\mu}}$ and $\hat{\bm{\Sigma}}$, obtained from historical data. However, these estimates often contain significant errors. As noted by \citet{delage2010distributionally}, one can rarely be fully confident in such estimates. This motivates the use of more general uncertainty sets. One of the  common approachs is to consider mean-covariance-based uncertainty sets, in which the true mean vector and covariance matrix are assumed to lie within neighborhoods centered at $\hat{\bm{\mu}}$ and $\hat{\bm{\Sigma}}$. Various forms of such uncertainty sets have been studied in the literature (see, e.g., \citet{delage2010distributionally}, \citet{garlappi2007portfolio}, \cite{XLS18b}, \cite{GSG2024}, and references therein).

In this section, we incorporate both distributional and parameter uncertainty by adopting the uncertainty set proposed in \cite{kang2019data}.
  \begin{equation}\label{set-D}
\mathcal{D}_{n+1}(\gamma_1, \gamma_2)=
\left\{
(G, F)\in\mathcal{F}(\mathbb{R}^{n+1})\;:\;
\begin{array}{l}
\big(\mathbb{E}(\bm{X}^G) - \hat{\bm{\mu}}_n\big)^{\top}
\bm{\Sigma}_n^{-1}
\big(\mathbb{E}(\bm{X}^G) - \hat{\bm{\mu}}_n\big)
\le \gamma_1, \ \mathbb{E}[X_0^F]=\hat{\mu}_0, \\
  \E[(\boldsymbol{\xi} -  \mathbb{E}[\bm{\xi}]) (\boldsymbol{\xi} -  \mathbb{E}[\bm{\xi}])^\top] = 
    \boldsymbol{\Sigma} \succ 0, \   \   \|\boldsymbol{\Sigma}-\hat{\boldsymbol{\Sigma}}\|_{F}\leq \gamma_2. 
\end{array}
\right\},
\end{equation}
where $\gamma_1 \ge 0$ and $\gamma_2 \ge 0$ are constants, $\hat{\bm{\mu}}_n \in \mathbb{R}^n$,  $\hat{\mu}_0 \in \R$, 
and $\hat{\bm{\Sigma}} \in S_{n+1}^+$.  Here, the Frobenius norm $\|\cdot \|_F$ is 
defined as 
 $\|\bm{A}\|_F=\sqrt{\sum_{i=1}^n\sum_{j=1}^n a_{ij}^2}=\sqrt{\operatorname{tr}(\bm{A}^{\top}\bm{A})}$ for an $n\times n$ matrix $\bm{A}=(a_{ij})_{n\times n}$.

In the uncertainty set $\mathcal{D}_{n+1}(\gamma_1, \gamma_2)$,  the condition $ \big(\mathbb{E}(\bm{X}^G) - \hat{\bm{\mu}}_n\big)^{\top}
\bm{\Sigma}_n^{-1}
\big(\mathbb{E}(\bm{X}^G) - \hat{\bm{\mu}}_n\big)
\le \gamma_1$ 
captures the uncertainty in the expected losses of the risky assets, reflecting the fact that these expected losses are more sensitive to sample data. The condition $\mathbb{E}[X_0^F]=\hat{\mu}_0$ imposes a fixed mean on the index loss $X_0$, consistent with the observation that the expected loss of the index is generally more stable, so its sample mean can serve as a reasonable   estimate.
 Finally,  the condition  
$\|\boldsymbol{\Sigma}-\hat{\boldsymbol{\Sigma}}\|_{F}\leq \gamma_2$ describes uncertainty in both the variances of the risky assets and the index, as well as the covariances between the risky assets and the index.

Under this uncertainty set, i.e.,  when    
${\cal D}_{n+1} =\mathcal{D}_{n+1}(\gamma_1, \gamma_2)$,
  problems  \eqref{mean-upper-m-rb-k} and 
 \eqref{upper-m-rb-k} reduce to the following    problems: 
 \begin{equation} 
\label{mean-rho-D}
    \min_{\boldsymbol{\omega} \in {\cal W}_r}  \, \  \sup_{(G,F)\in {\cal D}_{n+1}(\gamma_1,\gamma_2)} \, \mathbb{E}[\bm{\omega}^{\top}\bm{X}^G]+\kappa \,  \sup_{(G,F)\in {\cal D}_{n+1}(\gamma_1,\gamma_2)} \,   \mathbb{E}[(\boldsymbol{\omega}^{\top}\bm{X}^G-X_0^F-\alpha)_+^m] \\
   \end{equation}
and 
\begin{equation}
\label{rho-D}    \min_{\boldsymbol{\omega} \in {\cal W}_r}  \,   \sup_{(G,F)\in {\cal D}_{n+1}(\gamma_1,\gamma_2)} \,   \mathbb{E}[(\boldsymbol{\omega}^{\top}\bm{X}^G-X_0^F-\alpha)_+^m].
\end{equation}  
Similarly,  when    
${\cal D}_{n+1} =\mathcal{D}_{n+1}(\gamma_1, \gamma_2)$,
   problems   \eqref{MV-model-D_n+1} and  \eqref{V-model-D_n+1} reduce 
to  following problems: 
 \begin{equation}
\label{MV-model-D-g1-g2}
    \min_{\boldsymbol{\omega} \in {\cal W}_r}  \, \  \sup_{(G,F)\in {\cal D}_{n+1}(\gamma_1,\gamma_2)} \, \mathbb{E}[\bm{\omega}^{\top}\bm{X}^G]+\kappa \, 
     \sup_{(G,F)\in {\cal D}_{n+1}(\gamma_1,\gamma_2)} \,   {\rm Var}(\boldsymbol{\omega}^{\top}\bm{X}^G) 
\end{equation}
and 
\begin{equation}
\label{V-model-D-g1-g2}
    \min_{\boldsymbol{\omega} \in {\cal W}_r}  \,   \sup_{(G,F)\in {\cal D}_{n+1}(\gamma_1,\gamma_2)} \,   {\rm Var}(\boldsymbol{\omega}^{\top}\bm{X}^G). 
\end{equation}  

To solve problems \eqref{mean-rho-D}-\eqref{V-model-D-g1-g2},  
  for a given $\hat{\bm{\mu}}_n \in \R^n$,
 $\hat{\mu}_0 \in \R$, and $\hat{\bm{\Sigma}}_n \in S_{n}^{+}$,   we define the following sets of vectors  and matrices:
\begin{equation}\label{uncertainty set in short}
\begin{aligned}
      \mathcal{U}_{\hat{\bm{\mu}}} =& \big \{\bm{\mu} = (\bm{\mu}_n, \mu_0)\in \mathbb{R}^{n+1}:\ (\bm{\mu}_n-\hat{\bm{\mu}}_n)^{\top}\hat{\bm{\Sigma}}_n^{-1}(\bm{\mu}_n-\hat{\bm{\mu}}_n)\leq \gamma_{1}, \ \mu_0=\hat{\mu}_0 \big \},\\
      \mathcal{U}_{\hat{\bm{\Sigma}}}  =&  \big \{\bm{\Sigma}\in S^{+}_{n+1}: \  \ \|\bm{\Sigma}-\hat{\bm{\Sigma}}\|_{F}\leq \gamma_2 \big \}, \\
  \mathcal{U}_{(\hat{\bm{\mu}},\hat{\bm{\Sigma}})}  =& \ \, \mathcal{U}_{\hat{\bm{\mu}}}  \times  \mathcal{U}_{\hat{\bm{\Sigma}}}.
  \end{aligned}
 \end{equation}

 
We first show that problem \eqref{mean-rho-D}  is equivalent to 
   a convex optimization problem for $m=1$ and $m=2$.  

\begin{theorem} 
  \label{th-EIT-upm-cop} Assume that ${\cal W}$ is a nonempty convex set of $\mathbb{R}^n$. 
  
    (i) The robust  EIT portfolio selection problem \eqref{mean-rho-D} for  $m=1$ is equivalent 
  to the following  convex optimization problem: 
\begin{equation}
\label{ob-m=1}
\begin{aligned}
        \min_{\bm{\omega} \in {\cal W}, \, s, \, t \in \R^+}  & \ \bm{\omega}^{\top}\hat{\bm{\mu}}_n+\sqrt{\gamma_1} \, t +\frac{\kappa}{2+\kappa} \, 
        \sqrt{(\bm{\omega}^{\top}\hat{\bm{\mu}}_n+\sqrt{\gamma_1} \, t-\alpha^*)^2+s^2} \\
        \textrm{s.t.} & \ \sqrt{\bm{a}^{\top}(\hat{\bm{\Sigma}}+\gamma_2 \boldsymbol{I}_{n+1})\boldsymbol{a}} \leq s,\\
        & \ \sqrt{\bm{\omega}^{\top}\hat{\bm{\Sigma}}_n\bm{\omega}}\leq t,\\
          & \ \sqrt{\gamma_1} \sqrt{\bm{\omega}^{\top}\hat{\bm{\Sigma}}_n\bm{\omega}}\leq -\bm{\omega}^{\top}\hat{\bm{\mu}}_n- r,\ \ \ \boldsymbol{\omega}^{\top}\boldsymbol{e}=1.
    \end{aligned}
\end{equation} 

(ii) The  robust  EIT  portfolio selection problem \eqref{mean-rho-D} for  $m=2$ is equivalent  to the following  convex optimization problem:
 \begin{equation}
 \label{ob-m=2}
\begin{aligned}
        \min_{\boldsymbol{\omega} \in {\cal W}, \, s, \, t  \in \R^+} \ & \ \bm{\omega}^{\top}\hat{\bm{\mu}}_n+\sqrt{\gamma_1} \, t + \kappa \big((\bm{\omega}^{\top}\hat{\bm{\mu}}_n+\sqrt{\gamma_1} \, t-\alpha^*)_+^2+s^2\big) \\
        \textrm{s.t.}     & \ \sqrt{\boldsymbol{a}^{\top}(\hat{\boldsymbol{\Sigma}}+\gamma_2 \boldsymbol{I}_{n+1})\boldsymbol{a}} \leq s,\\
        & \ \sqrt{\bm{\omega}^{\top}\hat{\bm{\Sigma}}_n\bm{\omega}}\leq t,\\
        & \ \sqrt{\gamma_1}\sqrt{\bm{\omega}^{\top}\hat{\bm{\Sigma}}_n\bm{\omega}}\leq-\bm{\omega}^{\top}\hat{\bm{\mu}}_n- r,\  \ \  \boldsymbol{\omega}^{\top}\boldsymbol{e}=1.
    \end{aligned}
\end{equation} 
\end{theorem}

\begin{proposition}
 \label{th-EIT-upm-cop-only}
Assume that $\cal W$ is a nonempty convex set of $\mathbb{R}^n$.

(i) The  robust  EIT portfolio selection problem \eqref{rho-D} for  $m=1$ is equivalent 
  to the following  convex optimization problem: 
\begin{equation}
\label{ob-m=1-only}
\begin{aligned}
        \min_{\bm{\omega} \in {\cal W}, \, s, \, t \in \R^+}  & \  \bm{\omega}^{\top}\hat{\bm{\mu}}_n+\sqrt{\gamma_1} \, t  + \sqrt{(\bm{\omega}^{\top}\hat{\bm{\mu}}_n+\sqrt{\gamma_1} \, t-\alpha^*)^2+s^2} \\
        \textrm{s.t.} & \ \sqrt{\bm{a}^{\top}(\hat{\bm{\Sigma}}+\gamma_2 \boldsymbol{I}_{n+1})\boldsymbol{a}} \leq s,\\
        & \ \sqrt{\bm{\omega}^{\top}\hat{\bm{\Sigma}}_n\bm{\omega}}\leq t,\\
          & \ \sqrt{\gamma_1} \sqrt{\bm{\omega}^{\top}\hat{\bm{\Sigma}}_n\bm{\omega}}\leq -\bm{\omega}^{\top}\hat{\bm{\mu}}_n- r,\  \ \boldsymbol{\omega}^{\top}\boldsymbol{e}=1.
    \end{aligned}
\end{equation} 

(ii) The  robust EIT  portfolio selection problem \eqref{rho-D} for  $m=2$ is equivalent  to the following  convex optimization problem:
 \begin{equation}
 \label{ob-m=2-only}
\begin{aligned}
        \min_{\boldsymbol{\omega} \in {\cal W}, \, s, \, t \in \R^+} \ & (\bm{\omega}^{\top}\hat{\bm{\mu}}_n+\sqrt{\gamma_1} \, t-\alpha^*)_+^2+s^2 \\
        \textrm{s.t.}     & \ \sqrt{\boldsymbol{a}^{\top}(\hat{\boldsymbol{\Sigma}}+\gamma_2 \boldsymbol{I}_{n+1})\boldsymbol{a}} \leq s,\\
        & \ \sqrt{\bm{\omega}^{\top}\hat{\bm{\Sigma}}_n\bm{\omega}}\leq t,\\
        & \ \sqrt{\gamma_1}\sqrt{\bm{\omega}^{\top}\hat{\bm{\Sigma}}_n\bm{\omega}}\leq-\bm{\omega}^{\top}\hat{\bm{\mu}}_n- r,\ \ \ \boldsymbol{\omega}^{\top}\boldsymbol{e}=1.
    \end{aligned}
\end{equation} 
\end{proposition}

Next,  we  show that problems  \eqref{MV-model-D-g1-g2} and   \eqref{V-model-D-g1-g2}   are  also equivalent to 
   convex optimization problems.
\begin{proposition}
 \label{pro-ob-MV-model-D-g1-g2}
 Assume that $\cal W$ is a nonempty convex set of $\mathbb{R}^n$.

 (i) The robust mean-variance  portfolio selection problem \eqref{MV-model-D-g1-g2} 
 is equivalent  
  to the following  convex optimization problem: 
\begin{equation}
\label{ob-MV-D}
\begin{aligned}
        \min_{\bm{\omega} \in {\cal W}, \, s, \, t \in \R^+}  & \ \bm{\omega}^{\top}\hat{\bm{\mu}}_n+\sqrt{\gamma_1} \, t + \kappa s^2 \\
        \textrm{s.t.} \,  & \ \sqrt{\bm{\omega}^{\top}(\hat{\bm{\Sigma}}_n+\gamma_2 \boldsymbol{I}_{n})\boldsymbol{\omega}} \leq s,\\
        & \ \sqrt{\bm{\omega}^{\top}\hat{\bm{\Sigma}}_n\bm{\omega}}\leq t,\\
          & \ \sqrt{\gamma_1} \sqrt{\bm{\omega}^{\top}\hat{\bm{\Sigma}}_n\bm{\omega}}\leq -\bm{\omega}^{\top}\hat{\bm{\mu}}_n- r, \  \ \boldsymbol{\omega}^{\top}\boldsymbol{e}=1.
    \end{aligned}
\end{equation} 

(ii) The  robust  Markowitz's portfolio selection problem \eqref{V-model-D-g1-g2} is equivalent  to the following  convex optimization problem:
 \begin{equation}
 \label{ob-V-D}
\begin{aligned}
        \min_{\boldsymbol{\omega} \in {\cal W}, \, s  \in \R^+}  & \  s^2 \\
        \textrm{s.t.}     & \  \sqrt{\boldsymbol{\omega}^{\top}(\hat{\boldsymbol{\Sigma}}_n+\gamma_2 \boldsymbol{I}_{n})\boldsymbol{\omega}} \leq s,\\
        & \ \sqrt{\gamma_1}\sqrt{\bm{\omega}^{\top}\hat{\bm{\Sigma}}_n\bm{\omega}}\leq-\bm{\omega}^{\top}\hat{\bm{\mu}}_n- r,\ \  \boldsymbol{\omega}^{\top}\boldsymbol{e}=1.
    \end{aligned}
\end{equation} 
       \end{proposition}

\section{Empirical evaluation and comparison of robust EIT portfolio selection models}
\label{sec-5} 

In this section, we conduct  empirical studies to evaluate and compare the robust EIT portfolio selection 
models proposed in this paper. Specifically, we examine how minimizing the trade-off between mean loss 
and downside risk, as opposed to minimizing downside risk alone subject to a portfolio mean-loss constraint, affects 
portfolio performance. Performance is assessed using cumulative wealth, Sharpe ratio, Sortino ratio, and Omega ratio. 
Cumulative wealth is one of the most important performance metrics, while the three ratios are widely used to measure 
portfolio outcomes and the trade-offs between mean return and downside risk (or equivalently, between mean 
loss and downside risk). Additionally, we examine the impact of short selling versus no-short selling 
constraints and discuss the selection of model parameters.

 We conduct  two  empirical studies.    In the first study, we assume that the asset and index returns have a given mean vector and covariance 
matrix; that is, their 
joint distribution belongs to the uncertainty set ${\cal M}_{n+1}(\boldsymbol{\mu}, \boldsymbol{\Sigma})$, 
defined in \eqref{set-M_n+1}. In the second study, we assume that the joint distribution belongs to the generalized uncertainty 
set $\mathbb{D}_{n+1}(\gamma_1,\gamma_2)$, defined in \eqref{set-D}.

For our empirical experiments, we select the five largest companies by market capitalization from each of five sectors: 
Technology (AAPL, NVDA, MSFT, AVGO, ORCL); Healthcare (JNJ, ABBV, LLY, UNH, ABT); Financial 
Services (BRK-B, JPM, V, BAC, MA); Consumer Cyclical (AMZN, TSLA, HD, MCD, BKNG); and 
Consumer Defensive (WMT, COST, KO, PG, PM). This results in a portfolio consisting of 25 stocks.  Selecting top companies 
by market capitalization is a common practice among mutual fund and pension fund managers.

We use monthly loss data for these stocks from Yahoo! Finance over a 10-year period, 
 from January 2015 to December 2024. The loss (rate) of stock $i$ at time $t$ (the end of month $t$), 
 denoted by $x_{i,t}$, is computed as
 $$x_{i, t}=-\frac{P_{i, t+1}-P_{i, t}}{P_{i, t}},$$ 
 where $P_{i,t}$ is the adjusted closing price of stock $i$ at time $t$. Correspondingly, $ -x_{i,t}$ represents the return (rate) 
 of stock $i$ in month $t$. Furthermore, $\bm{x}_t^\top = (x_{1,t}, \ldots, x_{25,t})$ denotes the vector of losses 
 for the 25 stocks at time $t$. 
   
For each model considered in this section, we use the period from January 2015 
to December 2019 as the five-year in-sample (training) period and back-test 
the monthly portfolio returns over the out-of-sample period from January 2020 to December 2024.
 The initial optimal portfolio $\boldsymbol{\omega}_0^*$ for each model is solved using the five-year training data. 
  The corresponding out-of-sample return for month one is then 
   computed as $R_1=-\boldsymbol{\omega}_0^{*\top}\, \bm{x}_1$.
   The five-year window is subsequently rolled forward by one month, and the procedure is repeated. In general, the 
   out-of-sample portfolio return for month $t$ is given by $
R_t=-\boldsymbol{\omega}_{t-1}^{*\top} \, \boldsymbol{x}_t$,  
where $\boldsymbol{\omega}_{t-1}^{*}$ is  obtained by solving the model 
 using the most recent    five-year training data ending in  month $t-1$.
The models considered in this section are as follows: 
        \begin{enumerate}[(a)]    
    \item Mean-Expected Regret (M-ER) model:  This model minimizes a convex combination of the worst-case 
    mean loss and the expected regret relative to a specified portfolio loss threshold, subject to a constraint on 
    the portfolio return. It is formulated in \eqref{mean-upper-m-rb-k}  for    
     $m=1$.

    \item Expected Regret (ER) model: This model minimizes the worst-case expected regret 
     relative to a specified portfolio loss threshold,  subject to a constraint on the portfolio return. 
     It is formulated in 
      \eqref{upper-m-rb-k} for  $m=1$.

    \item Mean-Target Semi-Variance (M-TSV)  model: This  model minimizes a convex combination of the 
    worst-case mean loss and  the target semi-variance 
    relative to  a specified portfolio loss threshold,   subject to a constraint on the portfolio return. 
    It  is formulated in \eqref{mean-upper-m-rb-k} for  $m=2$.

    \item  Target Semi-Variance (TSV) model: This  model minimizes  the 
    worst-case target semi-variance 
    relative to  a specified portfolio loss threshold,   subject to a constraint on the portfolio return. 
    It  is formulated in 
      \eqref{upper-m-rb-k} for  $m=2$.
      
        \item Mean-Variance (M-V) model: This  model 
         minimizes   a convex  combination of the worst-case 
     mean loss and the variance   of a portfolio, subject to a constraint on the portfolio return.    
    It is  formulated in \eqref{MV-model-D_n+1}.

    \item  Variance (V) model:  This model minimizes the worst-case variance of portfolio losses, 
    subject to a constraint on the portfolio return. It is formulated in \eqref{V-model-D_n+1}. 
    We refer to this model as the variance model (denoted by V in the figures and tables presented in this section) 
    to distinguish it from the M-V model in item (e). 
    This notation is consistent with the distinctions 
    between M-ER and ER, and between M-TSV and TSV.

     \item Equal-Weighted Portfolio (EWP) model: This model allocates equal weights to all assets in the portfolio 
     and is often used as a benchmark in portfolio performance comparisons. Many empirical studies have 
     shown that equal-weighted portfolios can achieve superior performance; see, for example, \citet{demiguel2009optimal} for a detailed discussion of the EWP model.
             \end{enumerate}
    
  In our empirical studies,  
  we use models (e), (f), and (g) as benchmark models that do not track any index. Moreover, we use the following metrics to evaluate the portfolio performance: the out-of-sample cumulative wealth (CW), mean return (Mean Return), standard deviation (SD), expected regret (ER), target semi-variance (TSV), Sharpe ratio (Sharpe), Sortino ratio (Sortino), and Omega ratio (Omega). For further details of the Sortino and Omega ratios, we refer to \cite{BBB2021}, \cite{BVY2019}, \cite{guastaroba2016linear}, and the references therein.

 \subsection{Empirical studies  under  a given mean vector and covariance matrix} 
 \label{sec-5.1}

In the first study, we assume that the joint distribution of the stock and index 
losses  belongs to ${\cal M}_{n+1}(\boldsymbol{\mu}, \boldsymbol{\Sigma})$; that is, the losses  have a 
given mean vector $\boldsymbol{\mu}$ and covariance matrix $\boldsymbol{\Sigma}$.
   We first consider the case where short selling is allowed, i.e., $\bm{\omega} \in \mathbb{R}^n$.

Under these settings, model (a)(M-ER) reduces to \eqref{mean-upm-M} for $m=1$, with its optimal solution given in Theorem \ref{theo-mean-upm}(i). Similarly, model (b)(ER) reduces to \eqref{upm-M} for $m=1$, with its solution given in Proposition \ref{pro-upm-M-sol}(i). Model (c)(M-TSV) reduces to \eqref{mean-upm-M} for $m=2$, with its solution given in Theorem \ref{theo-mean-upm}(ii), while model (d)(TSV) reduces to \eqref{upm-M} for $m=2$, with its solution given in Proposition \ref{pro-upm-M-sol}(ii).    
Furthermore, model (e)(M-V) reduces to the following mean-variance model:
    \begin{equation}
\label{MV-model-1}
\begin{aligned}    
    \min_{\bm{\omega}\in\mathbb{R}^n} &  \ \  
    \bm{\omega}^{\top}\bm{\mu}_n      + \kappa \,  \bm{\omega}^{\top} \bm{\Sigma}_n \bm{\omega},   \\
         s.t.  &  \ \  
     -\bm{\omega}^{\top}\bm{\mu}_n     \ge r_0, \ \ \  \boldsymbol{\omega}^{\top}\boldsymbol{e}=1.
      \end{aligned}
\end{equation}    
Using an approach similar to that in the proof for Theorem \ref{theo-mean-upm} (the detailed proof is omitted), we obtain that the optimal solution to model \eqref{MV-model-1} is given by 
 \begin{equation}
    \bm{\omega}^*
        =\begin{pmatrix}
            \bm{\Sigma}_n^{-1}\bm{\mu}_n &  \bm{\Sigma}_n^{-1}\bm{e}
        \end{pmatrix}\begin{pmatrix}
            b_0 & -b_1\\-b_1 & b_2
        \end{pmatrix}\begin{pmatrix}
            \min\left\{-r_0,\frac{b_1}{b_0}-\frac{1}{2b_0\kappa}\right\}\\ 1
            \end{pmatrix}.
\end{equation}
Moreover, model (f) (V)  reduces to the classical   Markowitz model:
           \begin{equation}
   \label{V-model-1}
\begin{aligned}    
    \min_{\bm{\omega}\in\mathbb{R}^n} \quad & 
       \bm{\omega}^{\top} \bm{\Sigma}_n \bm{\omega},\\
    \, s.t.    \quad  &
    -\bm{\omega}^{\top}\bm{\mu}_n 
     \ge r_0, \ \ \boldsymbol{\omega}^{\top}\boldsymbol{e}=1. 
            \end{aligned}
        \end{equation}       
            The optimal solution to model \eqref{V-model-1}  is given by 
  \begin{equation}
  \label{sol-V-1}
    \bm{\omega}^*
        =\begin{pmatrix}
            \bm{\Sigma}_n^{-1}\bm{\mu}_n &  \bm{\Sigma}_n^{-1}\bm{e}
        \end{pmatrix}\begin{pmatrix}
            b_0 & -b_1\\-b_1 & b_2
        \end{pmatrix}\begin{pmatrix}
            \min\left\{-r_0,\frac{b_1}{b_0}\right\}\\ 1
            \end{pmatrix}.
\end{equation}   
Solution \eqref{sol-V-1} is provided in \cite{CJM2025a}. It can also be derived using an approach similar to 
that in the proof for  Theorem \ref{theo-mean-upm} (the detailed proof is omitted).

Note that models (a)-(f) involve several parameters, including $\alpha$,  $r_0$, and $\kappa$. In our empirical studies, these model parameters are chosen as follows:
\bn[(i)] 
\item The parameter $\alpha$ appears in models (a)-(d), where $-\alpha$ represents the excess 
return of the portfolio over the index return $-X_0$, reflecting an enhanced performance objective. 
A natural  and feasible choice of $-\alpha$ is a value  above the average historical excess return 
of the portfolio relative to the index, denoted by $-\bar{\alpha}$. Based on the in-sample data,  we 
 obtain $-\bar{\alpha} = 0.0075$. Accordingly, we set $-\alpha = 0.04$ in our empirical studies.

\item The parameter $r_0$ appears in the constraints of all models (a)-(f) as the minimum expected return requirement. It is natural to set $r_0$ above   the historical average portfolio return,   denoted by  $\bar{r}$,  but below the target return $-\alpha$ to balance an enhanced performance objective with feasibility; in other words, $\bar{r} < r_0 < -\alpha$. Using in-sample data, we find that $\bar{r}= 0.0176$. Accordingly, in our empirical study, we set $r_0 = 0.02$.

\item The parameter $\kappa$ appears in models (a), (c), and (e) as a weight that balances the trade-off between portfolio loss and downside risk. Intuitively, a smaller $\kappa$ places greater emphasis on portfolio loss (or equivalently, portfolio return), potentially leading to more aggressive, higher-return portfolios, whereas a larger $\kappa$ emphasizes downside risk and may result in more conservative, lower-return portfolios.
   For model (a) (M-ER), $\kappa$ must satisfy there $\kappa > \kappa_0$, as shown in Theorem \ref{theo-mean-upm}, to ensure the existence of an optimal solution. For models (c) and (e), $\kappa$  can take any positive value.  
    Based on the in-sample data,  we find  $\kappa_0 = 4.65$, and we therefore set $\kappa = 5$ to guarantee feasibility. For consistency and fair comparison, we use the same value of $\kappa$ across models (a), (c), and (e). The optimal choice of $\kappa$ for models (a), (c), and (e) in Section \ref{sec-optimal-kappa}.
  \en

Based on the above choices of model parameters, the cumulative wealth of each model under comparison, allowing for short selling, is shown in Figure \ref{fig: CW-s-d4v2}(a). We also consider models (a)-(f) when short selling is not allowed, under the same uncertainty set $\mathcal{M}_{n+1}(\boldsymbol{\mu}, \boldsymbol{\Sigma})$. That is, we impose the constraint $\bm{\omega} \ge 0$ in models (a)-(f) and evaluate the resulting portfolio performance. In this case, one of models (a)-(f) corresponds to a special case of models \eqref{mean-rho-D}-\eqref{V-model-D-g1-g2} with $\gamma_1 = 0$, $\gamma_2 = 0$, and ${\cal W} = A^n$, where $A = [0, \infty)$. The optimal solutions are provided in 
Theorem \ref{th-EIT-upm-cop} and Propositions \ref{th-EIT-upm-cop-only} and \ref{pro-ob-MV-model-D-g1-g2}. The cumulative wealth of each model under comparison without short selling is shown in Figure \ref{fig: CW-s-d4v2}(b).

 The other performance metrics,  with and without short selling, 
 are reported in Table \ref{tab:M}. For each performance criterion, the best result among the models 
 under comparison is highlighted in a box. Specifically, we highlight the highest values of the mean return, 
 Sharpe ratio, Sortino ratio, and Omega ratio, as well as the lowest values of the standard deviation, expected regret, and target semi-variance.

 From Figures \ref{fig: CW-s-d4v2}(a) and \ref{fig: CW-s-d4v2}(b), we observe that, under the uncertainty 
set ${\cal M}_{n+1}(\bm{\mu}, \bm{\Sigma})$, models that minimize a convex 
combination of mean loss and downside risk achieve significantly higher out-of-sample cumulative wealth than the corresponding 
models that minimize downside risk alone, regardless of whether short selling is allowed. 
Specifically, the M-ER, M-TSV, and M-V models yield substantially higher cumulative 
wealth than the ER, TSV, and V models, respectively. In addition, when short selling is allowed, the M-TSV and M-V models 
achieve the highest cumulative wealth, whereas when short selling is prohibited, the M-TSV model attains the highest cumulative wealth.

Moreover, Table \ref{tab:M} shows that, with short selling, the out-of-sample mean return and Omega ratio for each model (a)-(f) are 
higher than those of the corresponding models without short selling. However, for each model (a)-(f), the out-of-sample 
standard deviation (SD), expected regret (ER), and target semi-variance (TSV) are lower when short selling is not allowed. 
  Furthermore, for the M-ER and M-TSV models, both the out-of-sample Sharpe ratio and Sortino ratio are higher when short selling is prohibited. In other words, allowing short selling generally leads to higher mean return and Omega ratio but may also result in larger downside risk and lower Sharpe and Sortino ratios. Conversely, prohibiting short selling may limit returns but yields higher Sharpe and Sortino ratios.

Table \ref{tab:M} further indicates that, when short selling is not allowed, models that minimize a convex combination of 
mean loss and downside risk achieve significantly higher out-of-sample mean return, Sharpe, Sortino, and Omega ratios 
than the corresponding models that minimize downside risk alone. Specifically, the M-ER, M-TSV, and M-V models yield 
substantially higher out-of-sample mean return, Sharpe, Sortino, and Omega ratios than the ER, TSV, and V models, respectively. 
 On the other hand, when short selling is allowed, models that minimize a convex combination of mean loss and downside risk also achieve significantly higher out-of-sample mean return and Omega ratio than the corresponding models that minimize 
 downside risk alone, although the Sharpe and Sortino ratios are improved with models that minimize downside risk alone.

Furthermore, Table \ref{tab:M} indicates that models minimizing a convex combination of mean loss and downside risk achieve 
Pareto-optimal performance with respect to these two risk measures. For example, the M-TSV model attains Pareto-optimal 
out-of-sample mean loss and TSV, in the sense that no competing model simultaneously improves both mean loss 
 (equivalently, mean return) and TSV. Similarly, the M-ER model achieves Pareto-optimal out-of-sample mean loss and 
 ER, as no other model simultaneously improves both mean loss and ER. The M-V model also attains Pareto-optimal 
 out-of-sample mean loss and variance, meaning that no competing model simultaneously improves both mean loss and variance.

Overall, under the uncertainty set ${\cal M}_{n+1}(\bm{\mu}, \bm{\Sigma})$, when short selling is allowed, the M-TSV and M-V 
models achieve the highest cumulative wealth, while the EWP, ER, and M-TSV models deliver the best Sharpe ratio, Sortino ratio, 
and Omega ratio, respectively. In contrast, when short selling is prohibited, the M-TSV model achieves not only the highest 
cumulative wealth and mean return but also the best Sharpe ratio, Sortino ratio, and Omega ratio. These results provide 
strong evidence that minimizing a convex combination of mean loss and downside risk is a more effective portfolio selection 
approach than minimizing downside risk alone.

\captionsetup[figure]{font=footnotesize}
\captionsetup[sub]{font=tiny} 

  \begin{figure}[H]
    \centering
    \begin{subfigure}{0.45\textwidth}
        \centering
        \includegraphics[width=\linewidth]{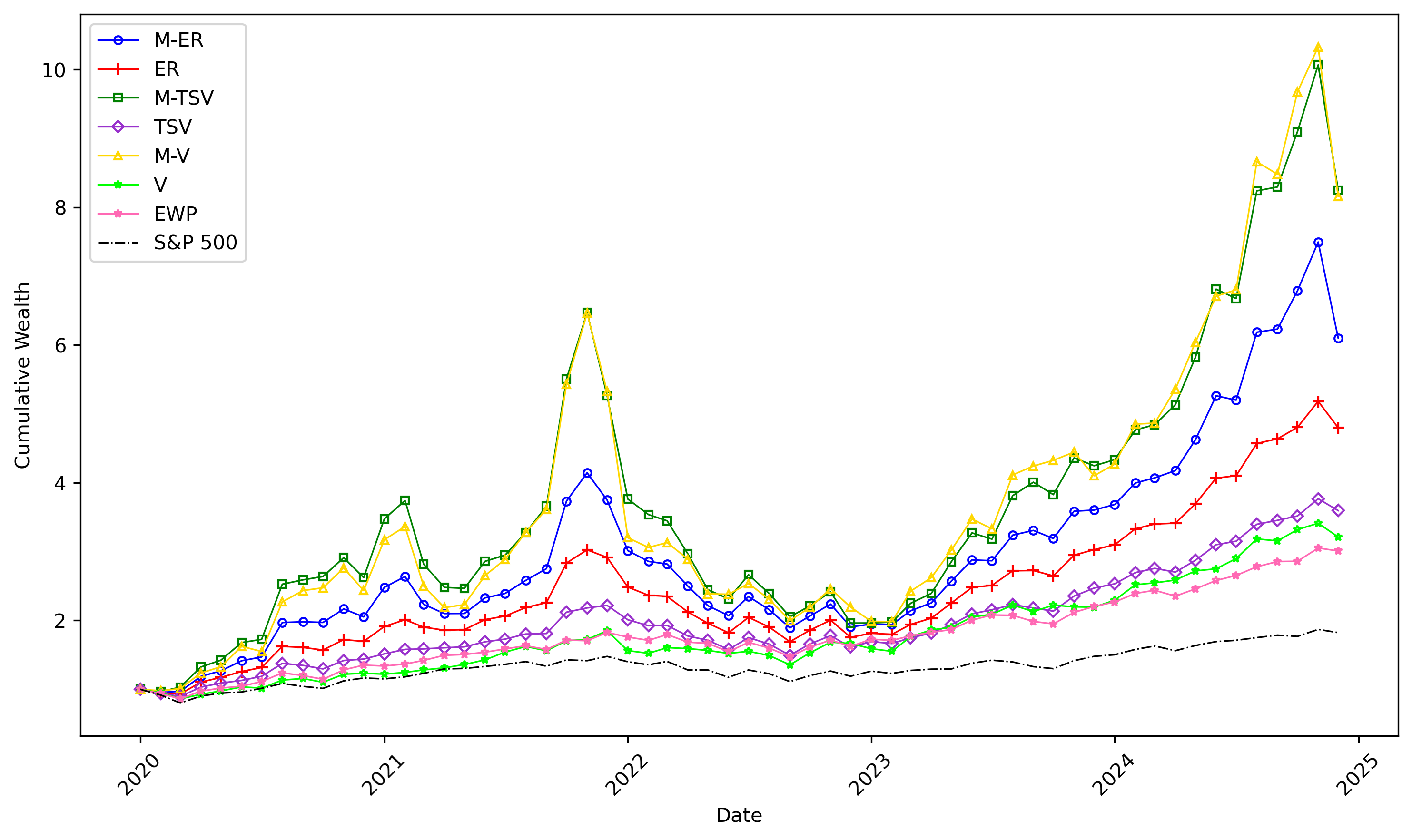}
        \caption{Cumulative wealth of portfolios based on  models (a)-(g)  
        under the uncertainty set 
    ${\cal M}_{n+1}(\bm{\mu}, \bm{\Sigma})$,  with short selling, 
     for $\alpha=-0.04$, $r_0=0.02$,  and $\kappa$=5.}
    \end{subfigure}
    \begin{subfigure}{0.45\textwidth}
        \centering
        \includegraphics[width=\linewidth]{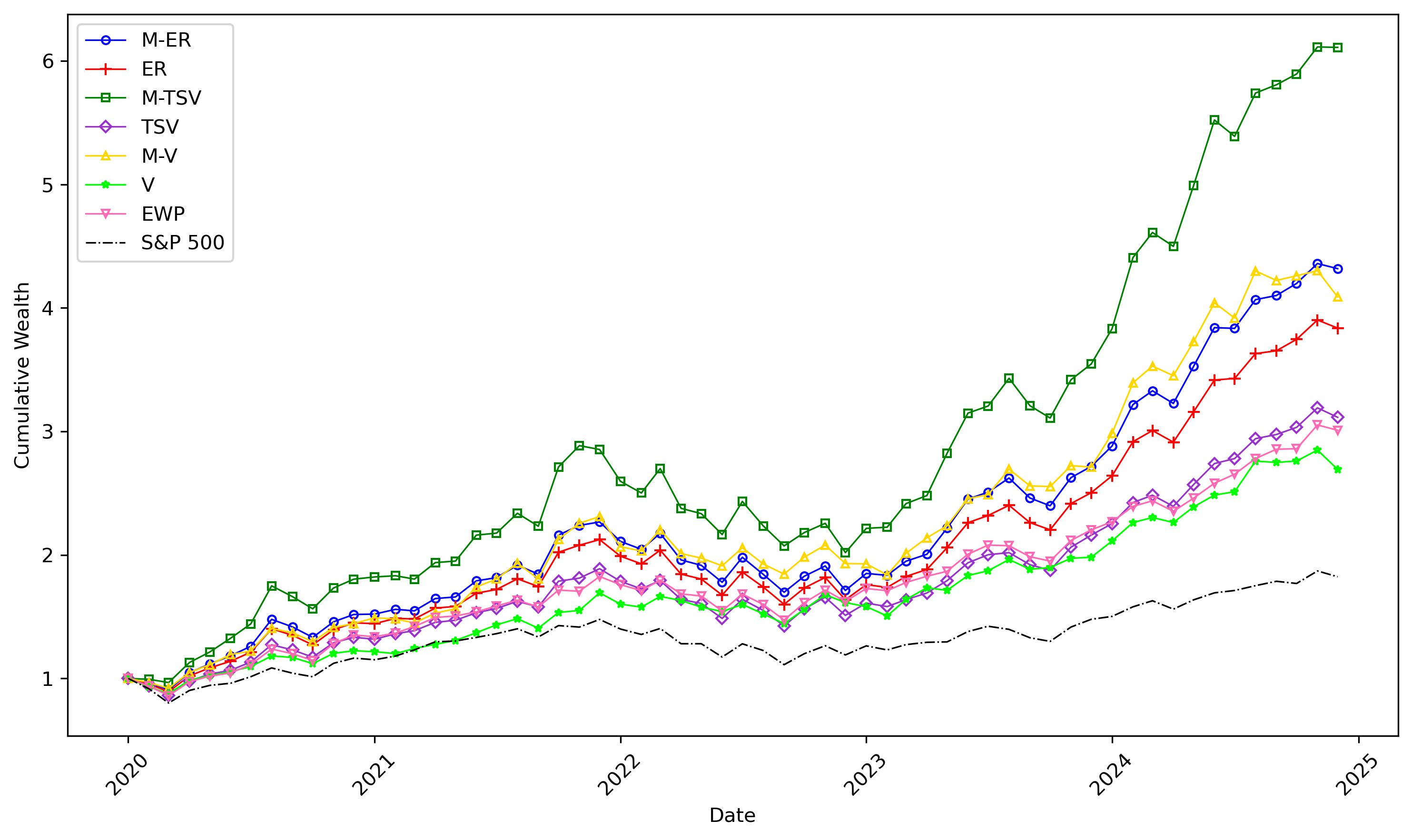}
        \caption{Cumulative wealth  of portfolios based on models (a)-(g) 
        under  the uncertainty set 
    ${\cal M}_{n+1}(\bm{\mu}, \bm{\Sigma})$, without short selling, 
    for 
         $\alpha=-0.04$,  $r_0=0.02$,   and  $\kappa$=5.}
    \end{subfigure}
    
    \begin{subfigure}{0.45\textwidth}
        \centering
        \includegraphics[width=\linewidth]{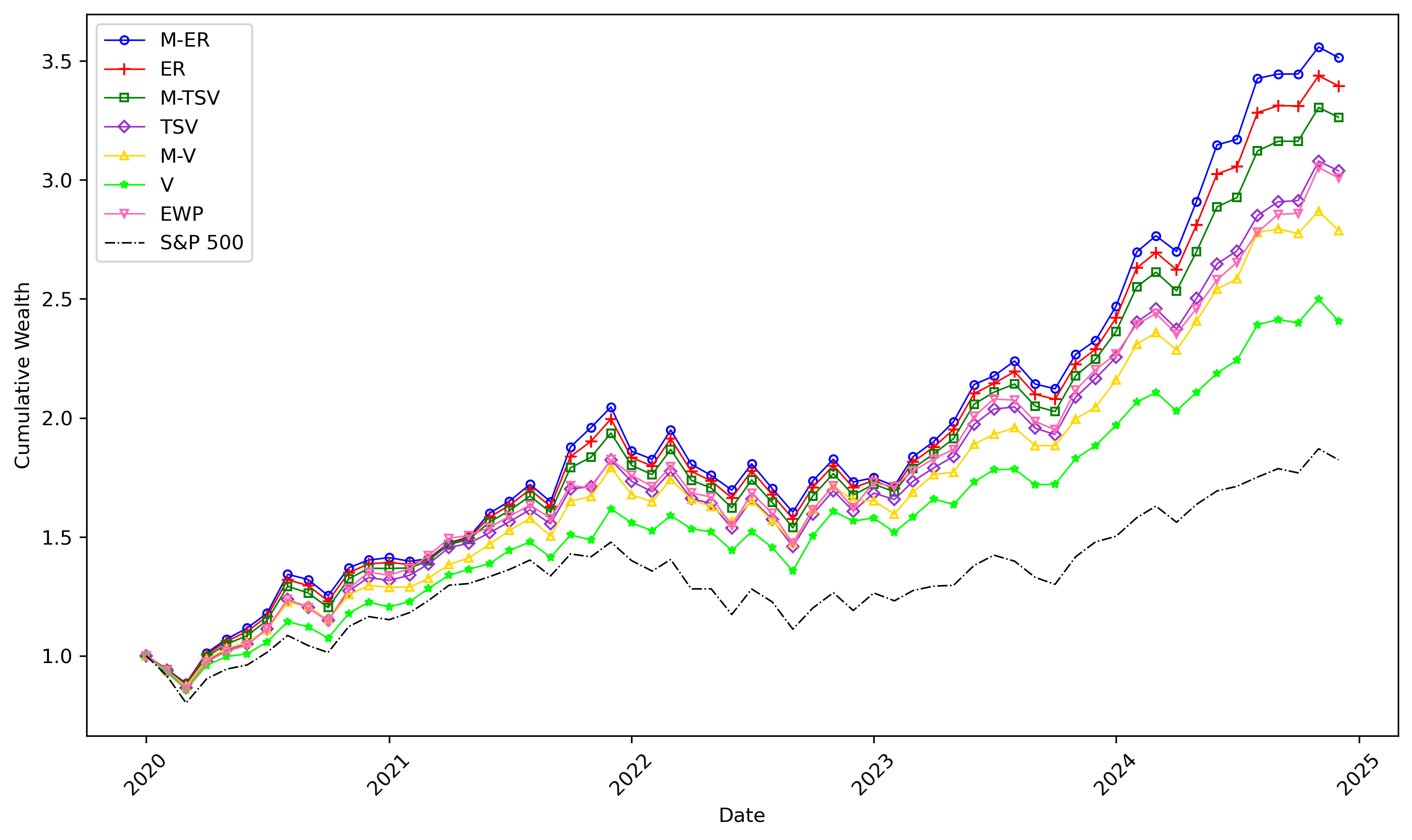}
        \caption{Cumulative wealth of portfolios based on models (a)-(g) 
        under the uncertainty set ${\cal D}_{n+1}(\gamma_1, \gamma_2)$,  with short selling, for  
      $\alpha=-0.04$, $r=-0.01$,  $\gamma_1$ = 0.22, $\gamma_2=0.025$, and $\kappa$=5.}
    \end{subfigure}
    \begin{subfigure}{0.45\textwidth}
        \centering
        \includegraphics[width=\linewidth]{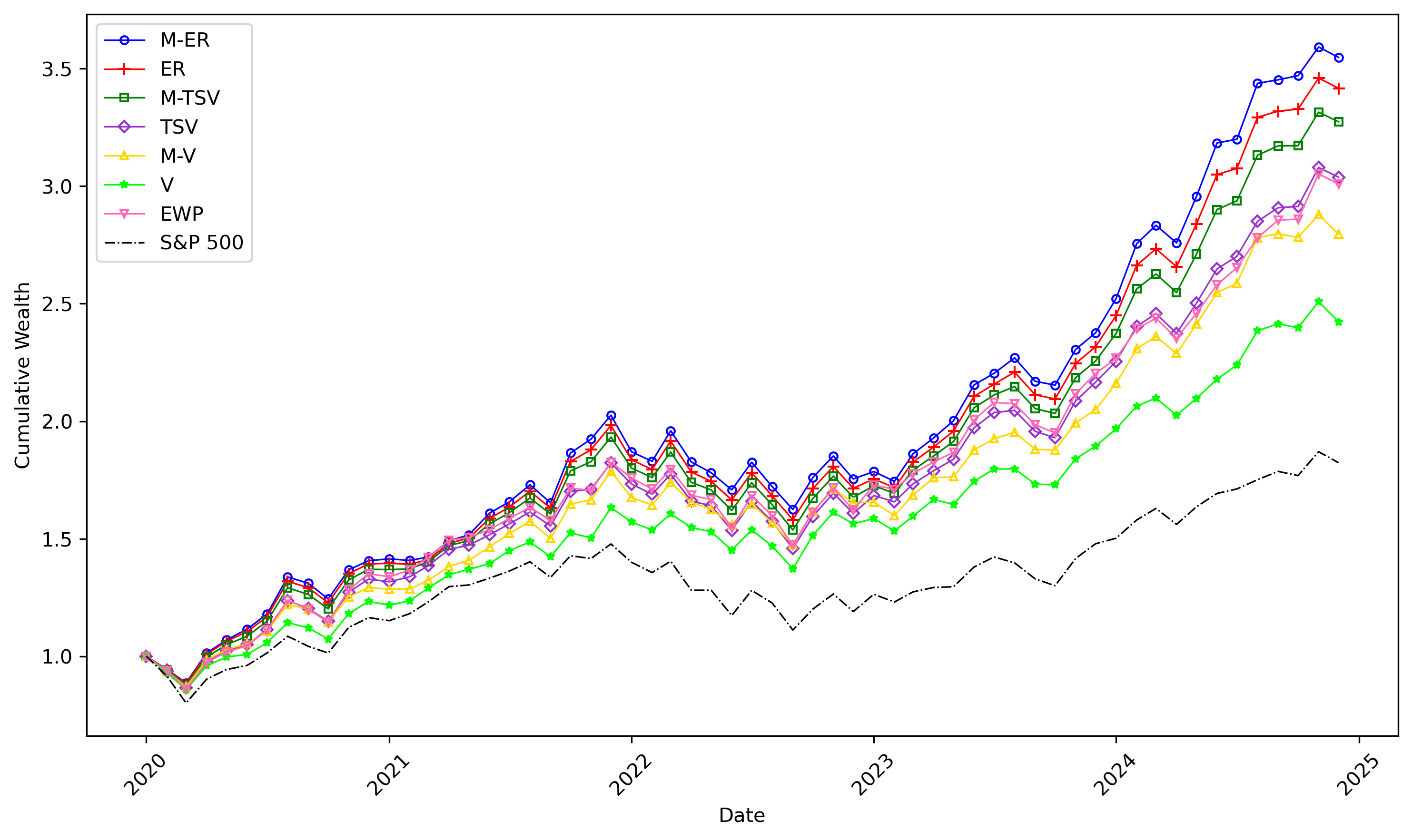}
        \caption{Cumulative wealth  of portfolios based on models (a)-(g) 
        under the uncertainty set ${\cal D}_{n+1}(\gamma_1, \gamma_2)$,  without short selling, for  
      $\alpha=-0.04$, $r=-0.01$,  $\gamma_1$ = 0.22, $\gamma_2=0.025$, and $\kappa$=5.}
    \end{subfigure}
    
    \caption{Cumulative wealth of portfolios based on  models (a)-(g). 
     }   
    \label{fig: CW-s-d4v2}
\end{figure}

\begin{table}[!htbp]
\centering
\footnotesize
\setlength{\tabcolsep}{6pt}
\renewcommand{\arraystretch}{1.15}

\begin{threeparttable}
\begin{tabular}{lccccccc}
\toprule
\textbf{Model} 
& \textbf{Mean Return} 
& \textbf{SD} 
& \textbf{ER} 
& \textbf{TSV} 
& \textbf{Sharpe}
& \textbf{Sortino}
& \textbf{Omega}\\
\midrule

\multicolumn{4}{l}{\textbf{With short selling}} \\
\addlinespace[2pt]

M-ER  & 0.038021 & 0.110156  & 0.043788 & 0.005928  & 0.345155   & 0.493813 & 0.954801\\
ER    & 0.030584 & 0.079341  & 0.036200 & 0.003655   & 0.385482  & \boxed{0.505900}  & 0.739903\\
M-TSV & 0.048939 & 0.150532  & 0.052809 & 0.009546  & 0.325107   & 0.500902 &  \boxed{1.169271} \\
TSV   & 0.023855 & 0.060104  & 0.032059 & 0.002714    & 0.396902   & 0.457940 & 0.496404\\
M-V   & \boxed{0.050084}  & 0.155359 & 0.053195 & 0.010652     & 0.322376  & 0.485271 & 1.189569\\
V     & 0.021873 & 0.056794  & 0.032461 & 0.002751       & 0.385127   & 0.416900 & 0.441582 \\
EWP   & 0.020299 & \boxed{0.050758}   & \boxed{0.031082}  & \boxed{0.002379}    & \boxed{0.399910} & 0.416174 & 0.366154\\
\addlinespace[6pt]

\multicolumn{4}{l}{\textbf{Without  short selling}} \\
\addlinespace[2pt]
M-ER  & 0.027183 & 0.064414  & 0.032721 & 0.002856 & 0.422011   & 0.508637 & 0.608300\\
ER    & 0.024979 & 0.061403  & 0.032544 & 0.002817   & 0.406812   & 0.470651 & 0.538451\\
M-TSV & \boxed{0.033821}  & 0.074642  & 0.033818 & 0.003061  & \boxed{0.453109}   & \boxed{0.611352}  & \boxed{0.817288} \\
TSV   & 0.021018 & 0.055961  & 0.032370 & 0.002723    & 0.375585    & 0.402786 & 0.413597\\
M-V   & 0.026298 & 0.063587   & 0.034140 & 0.002796    & 0.413580    & 0.497338 & 0.598669\\
V     & 0.018312 & \boxed{0.046906}  & 0.031722 &  \boxed{0.002219}      & 0.390402   & 0.388780 & 0.316313\\
EWP   & 0.020299 & 0.050758  & \boxed{0.031082} & 0.002379    & 0.399910   & 0.416174 & 0.366154\\
\bottomrule
\end{tabular}
 \end{threeparttable}
\caption{Out-of-sample performance of portfolios based on models (a)-(g)  under  the uncertainty set 
    ${\cal M}_{n+1}(\bm{\mu}, \bm{\Sigma})$, with and without short selling, for
     $\alpha=-0.04$,  $r_0=0.02$, and  $\kappa=5$.}
     \label{tab:M} 

\vspace{1.5 cm} 
  
\centering
\footnotesize
\setlength{\tabcolsep}{6pt}
\renewcommand{\arraystretch}{1.15}

\begin{threeparttable}
\begin{tabular}{lccccccc}
\toprule
\textbf{Model} 
& \textbf{Mean Return} 
& \textbf{SD} 
& \textbf{ER} 
& \textbf{TSV} 
& \textbf{Sharpe}
& \textbf{Sortino}
& \textbf{Omega}\\
\midrule

\multicolumn{4}{l}{\textbf{With short selling}} \\
\addlinespace[2pt]

M-ER  & \boxed{0.023264}  & 0.053799  & 0.030781 & 0.002346  & \boxed{0.432424}    & \boxed{0.480342} & \boxed{0.456294}\\
ER    & 0.022533 & 0.052189  & 0.030456 & 0.002304    & 0.431760    & 0.469414 & 0.426487\\
M-TSV & 0.021712 & 0.050744   & \boxed{0.030381} & 0.002281  & 0.427870     & 0.454599 & 0.398040\\
TSV   & 0.020297 & 0.049517     & 0.030621 & 0.002316   & 0.409908     & 0.421810 & 0.356556\\
M-V   & 0.018856 & 0.047087  & 0.031181 & \boxed{0.002204}       & 0.400460   & 0.401616 & 0.321904\\
V     & 0.016176 & \boxed{0.045302}  & 0.031771 & 0.002256      & 0.357065    & 0.340532 & 0.250130\\
EWP   & 0.020299 & 0.050758  & 0.031082 & 0.002379     & 0.399910     & 0.416174 & 0.366154\\
\addlinespace[6pt]

\multicolumn{4}{l}{\textbf{Without short selling}} \\
\addlinespace[2pt]
M-ER  & \boxed{0.023280}  & 0.052341  & 0.030257 & 0.002261  & \boxed{0.444768}     & \boxed{0.489637}  & \boxed{0.447392}\\
ER    & 0.022586 & 0.051537  & \boxed{0.030261}  & 0.002262    & 0.438237    & 0.474836 & 0.424523\\
M-TSV & 0.021762 & 0.050701  & 0.030357 & 0.002277   & 0.429233     & 0.456087 & 0.399229\\
TSV   & 0.020297 & 0.049517  & 0.030621 & 0.002316     & 0.409908   & 0.421810 & 0.356556\\
M-V   & 0.018905 & 0.047072  & 0.031067 & \boxed{0.002202}     & 0.401610    & 0.402831 & 0.320983\\
V     & 0.016281 & \boxed{0.045175}   & 0.031731 & 0.002243        & 0.360401   & 0.343776 & 0.252501\\
EWP   & 0.020299 & 0.050758  & 0.031082 & 0.002379    & 0.399910     & 0.416174 & 0.366154\\
\bottomrule
\end{tabular}
\end{threeparttable}
\caption{Out-of-sample performance of portfolios based on models (a)-(g) under the uncertainty set ${\cal D}_{n+1}(\gamma_1, \gamma_2)$, with and without short selling, for   $\alpha=-0.04$, $r=-0.01$, $\gamma_1=0.22$,  $\gamma_2=0.025,$ and $\kappa=5$.
  }
\label{tab:D}
\end{table}

\subsection{Empirical studies   under  a generalized uncertainty set} 

In the second empirical experiment, we assume that the joint distribution of the asset and index losses belongs to the 
 generalized uncertainty set $\mathcal{D}_{n+1}(\gamma_1,\gamma_2)$. 
  Under this uncertainty set, model (a)(M-ER) reduces to model \eqref{mean-rho-D} for $m=1$, with its optimal solution given in Proposition \ref{th-EIT-upm-cop-only}(i) and obtained by solving problem \eqref{ob-m=1}. Similarly, model (b)(ER) reduces to model \eqref{rho-D} for $m=1$, with its optimal solution given in Proposition \ref{pro-ob-MV-model-D-g1-g2}(i) and obtained by solving problem \eqref{ob-m=1-only}.  Model (c)(M-TSV) reduces to model \eqref{mean-rho-D} for $m=2$, with its optimal solution given in Proposition \ref{th-EIT-upm-cop-only}(ii) and obtained by solving problem \eqref{ob-m=2}. Model (d)(TSV) reduces to model \eqref{rho-D} for $m=2$, with its optimal solution given in Proposition \ref{pro-ob-MV-model-D-g1-g2}(ii) and obtained by solving 
  problem \eqref{ob-m=2-only}.
   Furthermore, model (e)(M-V) reduces to model \eqref{MV-model-D-g1-g2}, with its optimal solution given in Proposition \ref{pro-ob-MV-model-D-g1-g2}(i) and obtained by solving problem \eqref{ob-MV-D}. Moreover, 
   model (f)(V) reduces to model \eqref{V-model-D-g1-g2}, with its optimal solution given in Proposition \ref{pro-ob-MV-model-D-g1-g2}(ii) and obtained by solving problem \eqref{ob-V-D}.

Under the uncertainty set $\mathcal{D}_{n+1}(\gamma_1,\gamma_2)$, models (a)-(f) involve the parameters 
$\alpha$, $\kappa$, $r$, $\gamma_1$, and $\gamma_2$. The choices of the model parameters $\alpha$ and 
$\kappa$ have been discussed in Section \ref{sec-5.1}. The remaining model parameters $r$, $\gamma_1$, 
and $\gamma_2$ are chosen as follows:
 
  \bn[(i)] 
 
\item    
 The parameter $r$ appears in the constraints of all models (a)-(f) as the minimum worst-case expected 
 portfolio return over the uncertainty set. Compared to the minimum expected portfolio return $r_0$ discussed in Section \ref{sec-5.1}, 
 it is natural to expect that the minimum worst-case expected portfolio return is less than $r_0$, and therefore we choose $r$ such 
 that $r < r_0$. In Section \ref{sec-5.1}, $r_0 = 0.02$. Under the uncertainty set $\mathcal{D}_{n+1}(\gamma_1,\gamma_2)$, we set $r = -0.01$.
 
\item 
The parameters $\gamma_1$ and $\gamma_2$ represent the size of the uncertainty set $\mathcal{D}_{n+1}(\gamma_1,\gamma_2)$. 
 To select their values, we first estimate $\hat{\bm{\mu}} = (\hat{\bm{\mu}}_n, \hat{\mu}_0)$ and $\hat{\bm{\Sigma}}$ using 
 the sample mean vector and covariance matrix of the in-sample data. To calibrate $\gamma_1$ and $\gamma_2$ in 
 the uncertainty set $\mathbb{D}_n(\gamma_1,\gamma_2)$, we adopt a rolling-window approach to capture the 
 time-varying features of the data.
   Specifically, for each 5-year (60-month) rolling window (window $i$), we compute the sample mean vector $\hat{\bm{\mu}}_i$ 
and sample covariance matrix $\hat{\bm{\Sigma}}_i$. We then measure their deviations from $\hat{\bm{\mu}}$ 
and $\hat{\bm{\Sigma}}$ according to the metrics defining the uncertainty set. For the mean vector, we 
compute $\gamma_{1i}=(\hat{\bm{\mu}}_{n,i}-\hat{\bm{\mu}}_n)^{\top}\hat{\bm{\Sigma}}^{-1}(\hat{\bm{\mu}}_{n,i}-\hat{\bm{\mu}}_n)$ 
and for the covariance matrix, we compute $\gamma_{2i}=\|\hat{\bm{\Sigma}}_i-\hat{\bm{\Sigma}}\|_F$. 
Finally, we set $\gamma_1$ and $\gamma_2$ to be the averages of $\gamma_{1i}$ and $\gamma_{2i}$, which 
represent the typical magnitudes of temporal fluctuations in the mean vector and covariance matrix, respectively. Based on the in-sample data, 
we obtain $\gamma_1=0.22$ and $\gamma_2=0.025$. 
     \en

Based on the above choice of model parameters,   
 the cumulative wealth of each model under comparison with short selling allowed is shown 
  in Figure \ref{fig: CW-s-d4v2}(c). 
   We  also consider  models (a)-(f) when short selling  is not allowed under the same uncertainty set $\mathcal{D}_{n+1}(\gamma_1, \gamma_2)$. 
 That is, we impose the constraint $\bm{\omega} \ge 0$ in models (a)-(f) and evaluate the resulting portfolio performance. In this case, one of models (a)-(f) corresponds to a special case of models \eqref{mean-rho-D}-\eqref{V-model-D-g1-g2} 
 with  ${\cal W} =A^n$, where $A=[0, \infty)$. The optimal solutions are provided in Theorem \ref{th-EIT-upm-cop} and Propositions \ref{th-EIT-upm-cop-only} and \ref{pro-ob-MV-model-D-g1-g2}.
 The cumulative wealth of each model under comparison without short selling is presented 
 in Figure \ref{fig: CW-s-d4v2}(d).  The other performance metrics, both with and without short selling, 
 are reported in Table \ref{tab:D}. 

 From Figures \ref{fig: CW-s-d4v2}(c) and \ref{fig: CW-s-d4v2}(d), we observe that under the uncertainty set 
 ${\cal D}_{n+1}(\gamma_1, \gamma_2)$,  similar to the case under ${\cal M}_{n+1}(\bm{\mu}, \bm{\Sigma})$,  models that minimize a convex combination of mean loss and downside risk achieve significantly higher cumulative wealth than the corresponding models that minimize downside risk alone, regardless of whether short selling is allowed. Specifically, the M-ER, M-TSV, and M-V models yield substantially higher
  cumulative wealth than the ER, TSV, and V models, respectively. Moreover, the M-ER model consistently achieves 
  the highest cumulative wealth in both settings.

Moreover, Table \ref{tab:D} shows that, under the uncertainty set ${\cal D}_{n+1}(\gamma_1, \gamma_2)$, when short selling 
is not allowed, the out-of-sample mean return, standard deviation (SD), ER, TSV, Sharpe ratio, and Sortino ratio for each 
model (a)-(f) are improved relative to those of the corresponding models when short selling is allowed. 
In addition, the out-of-sample 
Omega ratio of the M-TSV model without short selling is higher than that of its counterpart with short selling.  This indicates 
that the M-TSV model effectively balances the trade-off between mean return and downside risk, regardless of whether 
short selling is permitted.

Table \ref{tab:D} also indicates that, regardless of whether short selling is allowed, models that minimize a convex combination 
of mean loss and downside risk achieve significantly higher out-of-sample mean return, Sharpe ratio, Sortino ratio, and 
Omega ratio than models that minimize downside risk alone. In particular, the M-ER, M-TSV, and M-V models consistently 
outperform the ER, TSV, and V models, respectively, across these performance measures.

Furthermore, Table \ref{tab:D} shows that under the uncertainty set ${\cal D}_{n+1}(\gamma_1, \gamma_2)$,  
again consistent with the results under ${\cal M}_{n+1}(\bm{\mu}, \bm{\Sigma})$,  
models that minimize a convex combination of mean loss and downside risk achieve Pareto-optimal performance 
with respect to these two risk measures. Specifically, the M-TSV model is Pareto-optimal in terms of out-of-sample 
mean loss and TSV, meaning that no competing model simultaneously improves both metrics. 
Similarly, the M-ER model is Pareto-optimal with respect to mean loss and ER, and the M-V model is 
Pareto-optimal with respect to mean loss and variance.

Overall, under the uncertainty set ${\cal D}_{n+1}(\gamma_1, \gamma_2)$, regardless of whether short selling is allowed, the M-ER model achieves the highest cumulative wealth, mean return, Sharpe ratio, Sortino ratio, and Omega ratio. Additionally, the  M-ER and M-TSV models also outperform the benchmark models M-V, V, and EWP, in terms of cumulative wealth, mean return, Sharpe ratio, Sortino ratio, and Omega ratio.
 These results further confirm that minimizing a convex combination of mean loss and downside risk is a more effective portfolio selection approach than minimizing downside risk alone.

 \subsection{Optimal weight selection for balancing mean loss and downside risk or variance} 
 \label{sec-optimal-kappa}

The parameter $\kappa$ in models \eqref{mean-upper-m-rb-k} and \eqref{MV-model-D_n+1} represents 
the weight used to balance portfolio mean loss and downside risk (or variance). Assigning a weight of $\kappa$ to 
the downside risk is equivalent to assigning weights $\lambda$ and $1-\lambda$ to the downside risk (or variance) 
and mean loss, respectively, where $\lambda = \frac{\kappa}{1+\kappa}$.

Unlike other model parameters, such as $\alpha$, $r_0$, $r$, $\gamma_1$, and $\gamma_2$, for which 
there is substantial empirical guidance for selection or estimation (as discussed in previous subsections),  
 the parameter $\kappa$ does not have a straightforward empirical approach for its determination.

In this subsection, we conduct a sensitivity analysis to examine how different choices of $\kappa$ affect 
the performance of the M-ER, M-TSV, and M-V models, and whether an optimal choice $\kappa^*$ exists for 
each model. To this end, we plot the out-of-sample cumulative wealth, Sharpe ratio, Sortino ratio, and Omega ratio 
of these models, both with and without short selling. These numerical results allow us to assess the impact of 
$\kappa$ on model performance and to identify the corresponding optimal values of $\kappa$.

 \subsubsection{Optimal weight selection under a given mean vector and covariance matrix} 
 \label{sec-optimal-kappa-M} 
 
As in Subsection \ref{sec-5.1}, we first assume that the joint distribution of stock and index losses belongs to 
${\cal M}_{n+1}(\boldsymbol{\mu}, \boldsymbol{\Sigma})$, and then consider the case where
 it belongs to ${\cal D}_{n+1}(\gamma_1, \gamma_2)$. Under these settings, the out-of-sample cumulative wealth 
 and Sharpe ratios of the M-ER, M-TSV, and M-V models are presented in Figures \ref{fig: CW-Sharpe-diff-k-M-with-SS} and \ref{fig: CW-Sharpe-diff-k-M-no-SS}, corresponding to scenarios with and without short selling, respectively.

From Figures \ref{fig: CW-Sharpe-diff-k-M-with-SS}(a)-(c), under the uncertainty set ${\cal M}_{n+1}(\bm{\mu}, \bm{\Sigma})$ 
with short selling, we observe that the  
cumulative wealth curves generated by the M-ER, M-TSV, and M-V models
decrease as $\kappa$ increases. This indicates that placing greater weight on downside risk, or equivalently less weight on mean loss, 
reduces cumulative wealth. Therefore, the highest cumulative wealth is achieved when $\kappa$ is chosen as small as possible 
within its feasible range. Notably, when $\kappa = 0$, these models do not admit optimal solutions. The feasible range is $(\kappa_0, \infty)$ 
for the M-ER model (where $\kappa_0$ is defined in Theorem \ref{theo-mean-upm}) and $(0, \infty)$ for the M-TSV and M-V models.

In contrast, Figure \ref{fig: CW-Sharpe-diff-k-M-with-SS}(d) shows that the out-of-sample Sharpe ratios of the M-ER and M-TSV 
models increase with $\kappa$, indicating a trade-off between cumulative wealth and the Sharpe ratio. 
For the M-V model, the Sharpe ratio increases first and then decreases as $\kappa$ increases, suggesting the existence of 
an optimal weighting parameter $\kappa^*$ that maximizes the Sharpe ratio. In this dataset, the 
optimal value is $\kappa^* = 20.9$ for the M-V model.

Furthermore,  for small values of $\kappa$ (for our data,  $\kappa < 12.3$), the M-ER model achieves a higher out-of-sample Sharpe 
ratio than the M-TSV and M-V models. In contrast,  for larger values of $\kappa$ (for our data,  $\kappa > 12.3$), the M-TSV model 
outperforms the M-ER and M-V models in terms of the Sharpe ratio within the dataset used for our empirical analysis.

Without short selling, Figures \ref{fig: CW-Sharpe-diff-k-M-no-SS}(a)-(c) shows that the out-of-sample Sharpe ratios of the M-TSV 
and M-V models initially increase with $\kappa$, attain their maximum values, and then decrease or remain below these maxima. 
In contrast, the out-of-sample Sharpe ratio of the M-ER model decreases as $\kappa$ increases. Therefore, under the 
uncertainty set ${\cal M}_{n+1}(\bm{\mu}, \bm{\Sigma})$ without short selling, optimal values of $\kappa$ exist
 for the M-TSV and M-V models that maximize the out-of-sample Sharpe ratios. In our dataset, these optimal values 
 are $\kappa^* = 2.1$ for the M-TSV model and $\kappa^* = 1.7$ for the M-V model. 
 For the M-ER model, smaller values of $\kappa$ are preferred; however, an optimal value is not attainable.

Moreover, when short selling is not allowed,   for small values of $\kappa$ (for our data, $\kappa < 2.2$), the M-V 
model achieves a higher out-of-sample Sharpe ratio than the M-ER and M-TSV models. In contrast, 
for large values of $\kappa$ (for our data,  $\kappa > 13.5$), the M-ER model outperforms the M-TSV and M-V models in terms of the Sharpe ratio within the dataset used for our empirical analysis. For intermediate values of $\kappa$ (for our data,  $2.2 < \kappa < 13.5 $), the M-TSV model achieves the highest out-of-sample Sharpe ratio among the three models.

To examine the effect of the weighting coefficient $\kappa$ on the Sortino and Omega ratios, we plot the out-of-sample 
values of these metrics for the M-ER, M-TSV, and M-V models in Figure \ref{fig: S-Q-diff-k-M}, for both 
scenarios with and without short selling.

With short selling, the Sortino ratios of all three models initially increase with $\kappa$, attain their maxima, and then decrease or stay below their maxima, indicating the existence of optimal values $\kappa^*$. Specifically, the optimal values are $\kappa^*= 7.1$, $7.5$, and $8.1$ 
for the M-ER, M-TSV, and M-V models, respectively. In contrast, the Omega ratios decrease monotonically as $\kappa$ increases.
     Moreover,  for small values of $\kappa$ (for our data,  $\kappa < 5.2$),  
     the M-TSV model achieves higher Sortino ratios than the M-ER and M-V models, 
whereas  for larger $\kappa$ (for our data,  $\kappa > 5.2$),  the M-ER model performs best. A similar pattern is observed for the Omega ratio: the M-V model 
performs best for small $\kappa$, while the M-ER model dominates for larger $\kappa$.

Without short selling, the Sortino and Omega ratios of the M-TSV and M-V models increase with $\kappa$ up to a maximum and then decrease or 
 stay below their maxima,   again indicating the existence of optimal values $\kappa^*$.

Overall, these results provide an empirical basis for selecting the optimal weighting coefficient $\kappa^*$ in the M-ER, M-TSV, and M-V models to balance cumulative wealth and risk-adjusted performance, as measured by the Sharpe,  Sortino, and Omega ratios.

\begin{figure}[h]
    \centering
    \begin{subfigure}{0.45\textwidth}
        \centering
        \includegraphics[width=\linewidth]{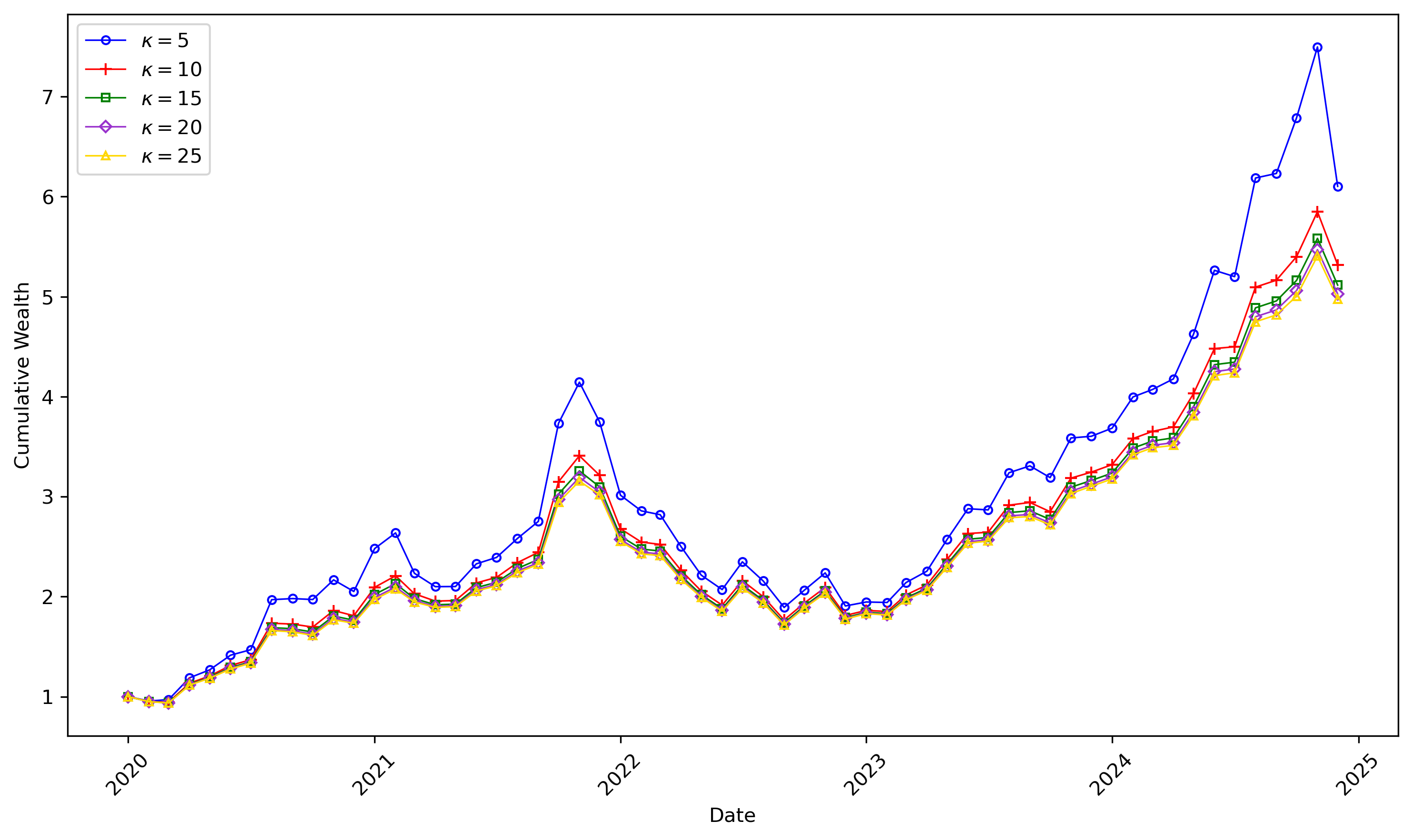}
        \caption{Cumulative wealth of the M-ER model with short selling}
    \end{subfigure}
    \begin{subfigure}{0.45\textwidth}
        \centering
        \includegraphics[width=\linewidth]{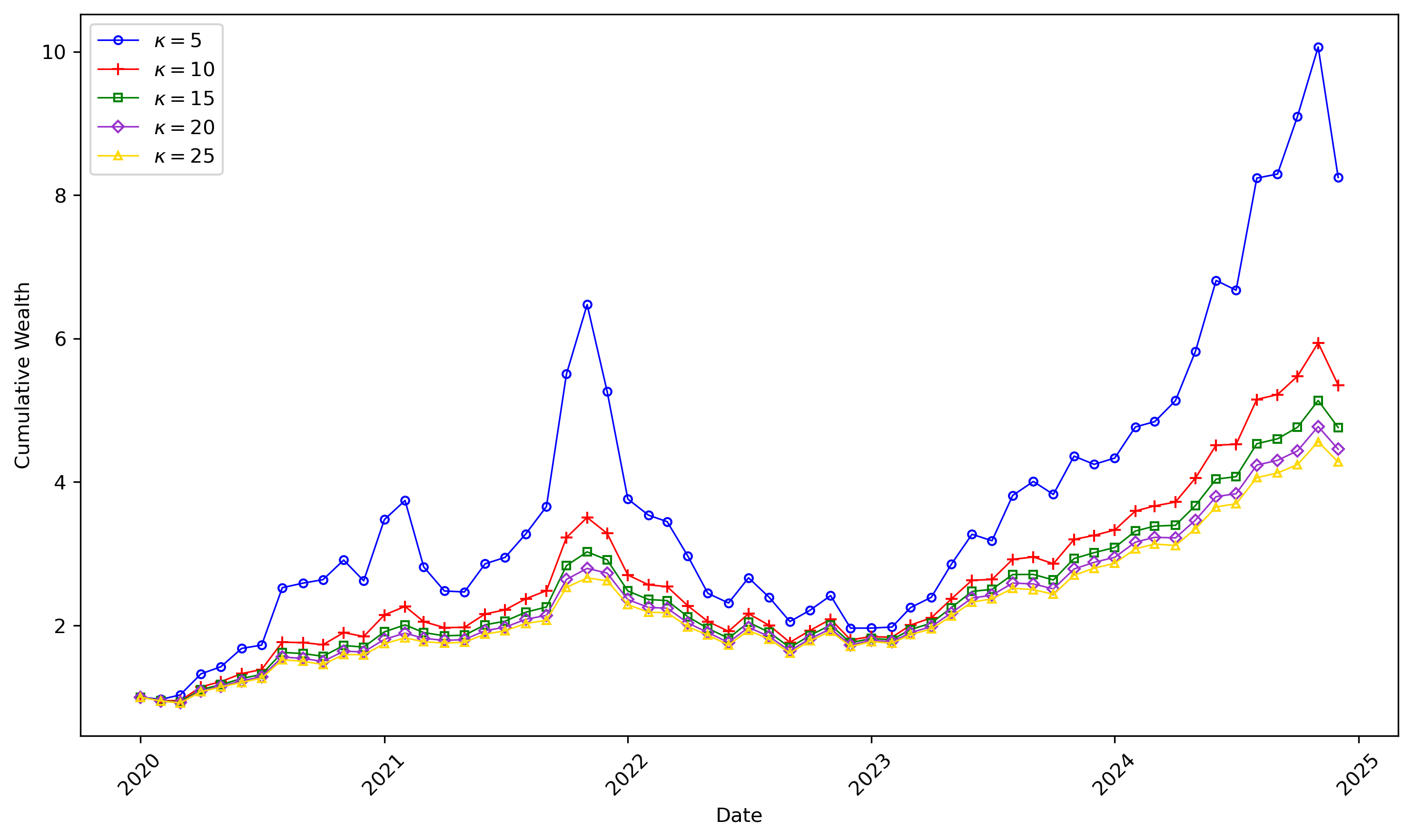}
        \caption{Cumulative wealth of the M-TSV model with short selling}
    \end{subfigure}
    
    \begin{subfigure}{0.45\textwidth}
        \centering
        \includegraphics[width=\linewidth]{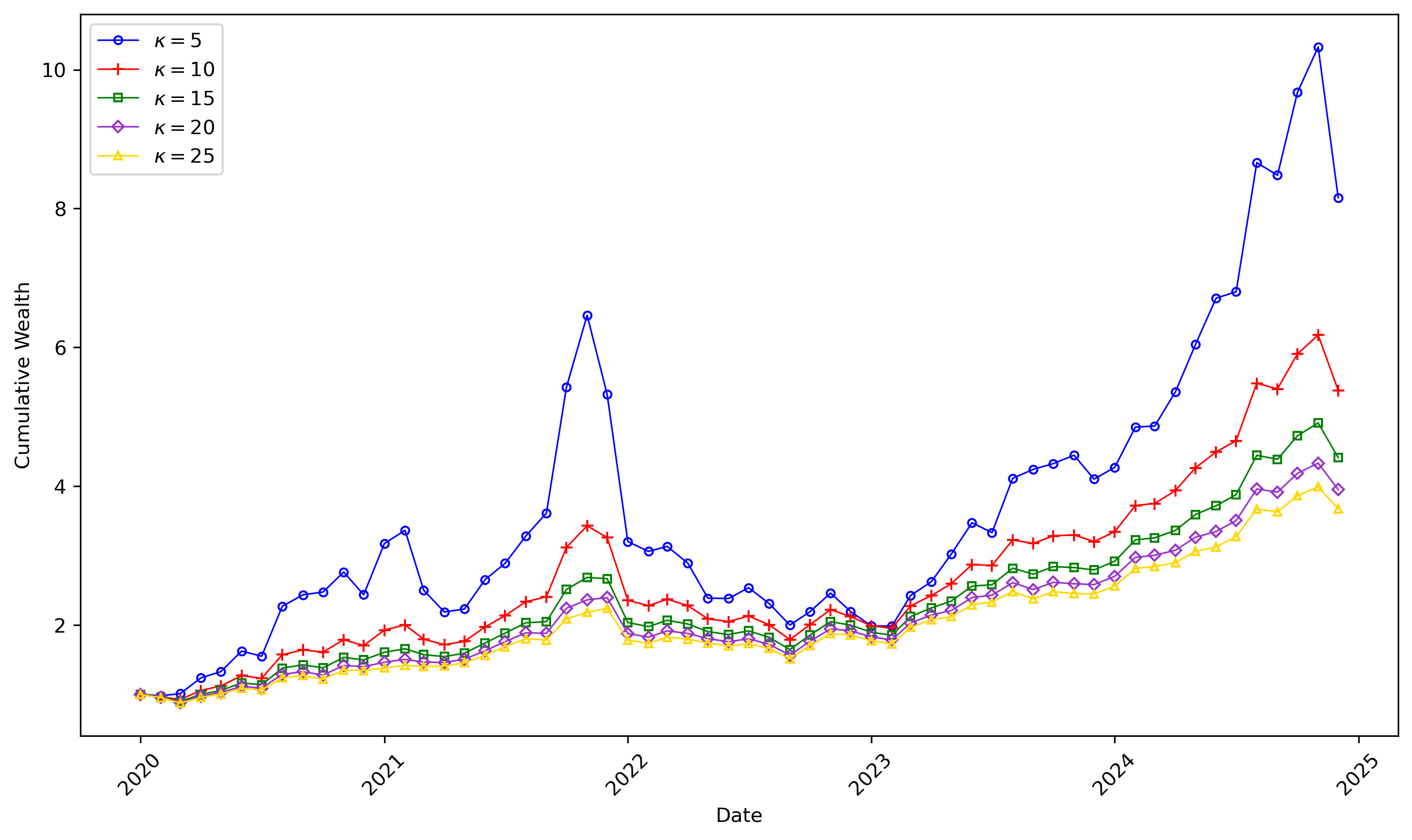}
        \caption{Cumulative wealth of the M-V model with short selling}
    \end{subfigure}
    \begin{subfigure}{0.45\textwidth}
        \centering
        \includegraphics[width=\linewidth]{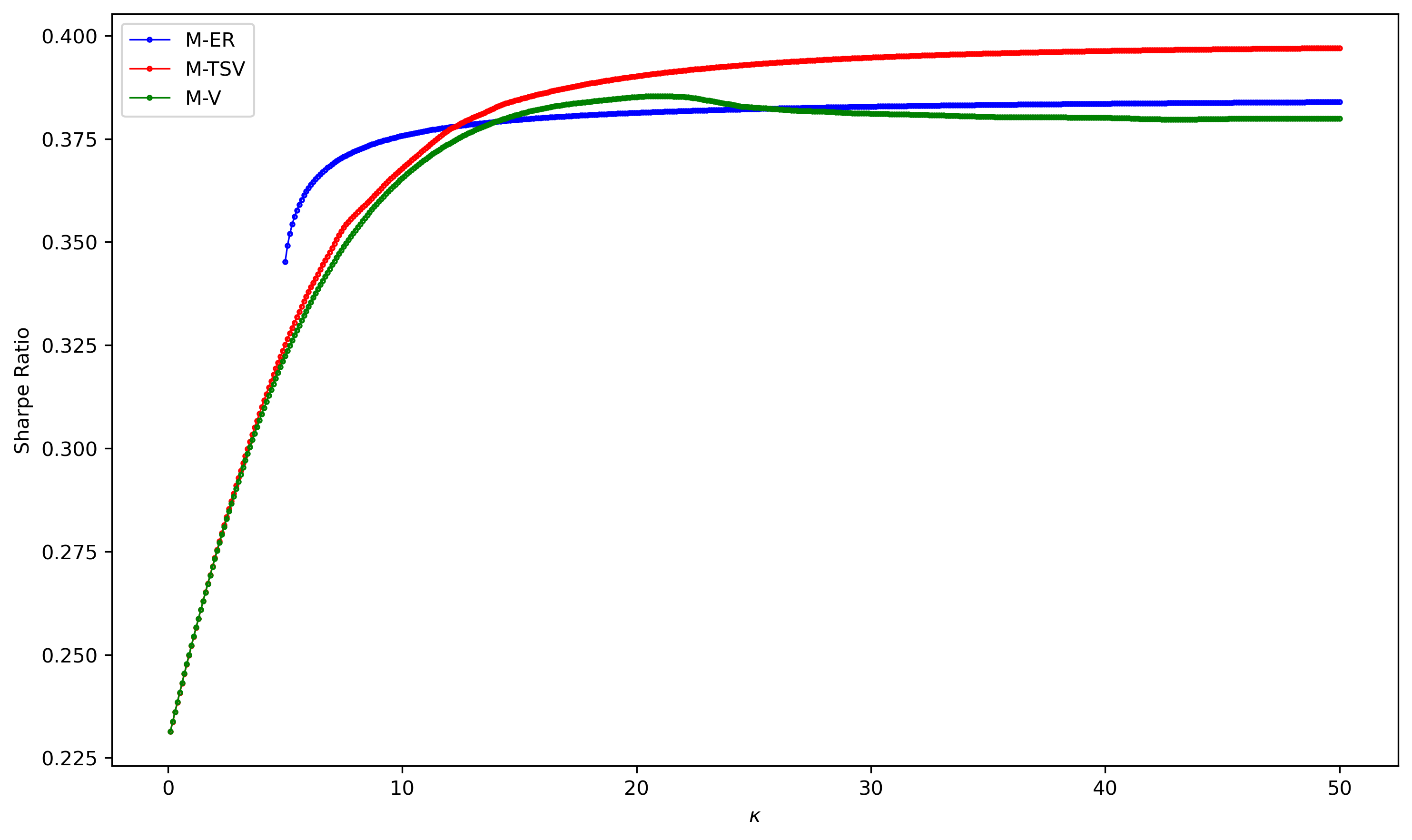}
        \caption{Sharpe ratios of the three models with short selling}
    \end{subfigure}
    
    \caption{Cumulative wealth and Sharpe ratios of portfolios based on the M-ER, M-TSV, and M-V models under 
     the uncertainty set ${\cal M}_{n+1}(\bm{\mu}, \bm{\Sigma})$,   with short selling, plotted against the weighting 
     coefficient $\kappa$, for $\alpha=-0.04$ and $r_0=0.02$.}   
        \label{fig: CW-Sharpe-diff-k-M-with-SS}
\end{figure}


\begin{figure}[H]
    \centering
    \begin{subfigure}{0.45\textwidth}
        \centering
        \includegraphics[width=\linewidth]{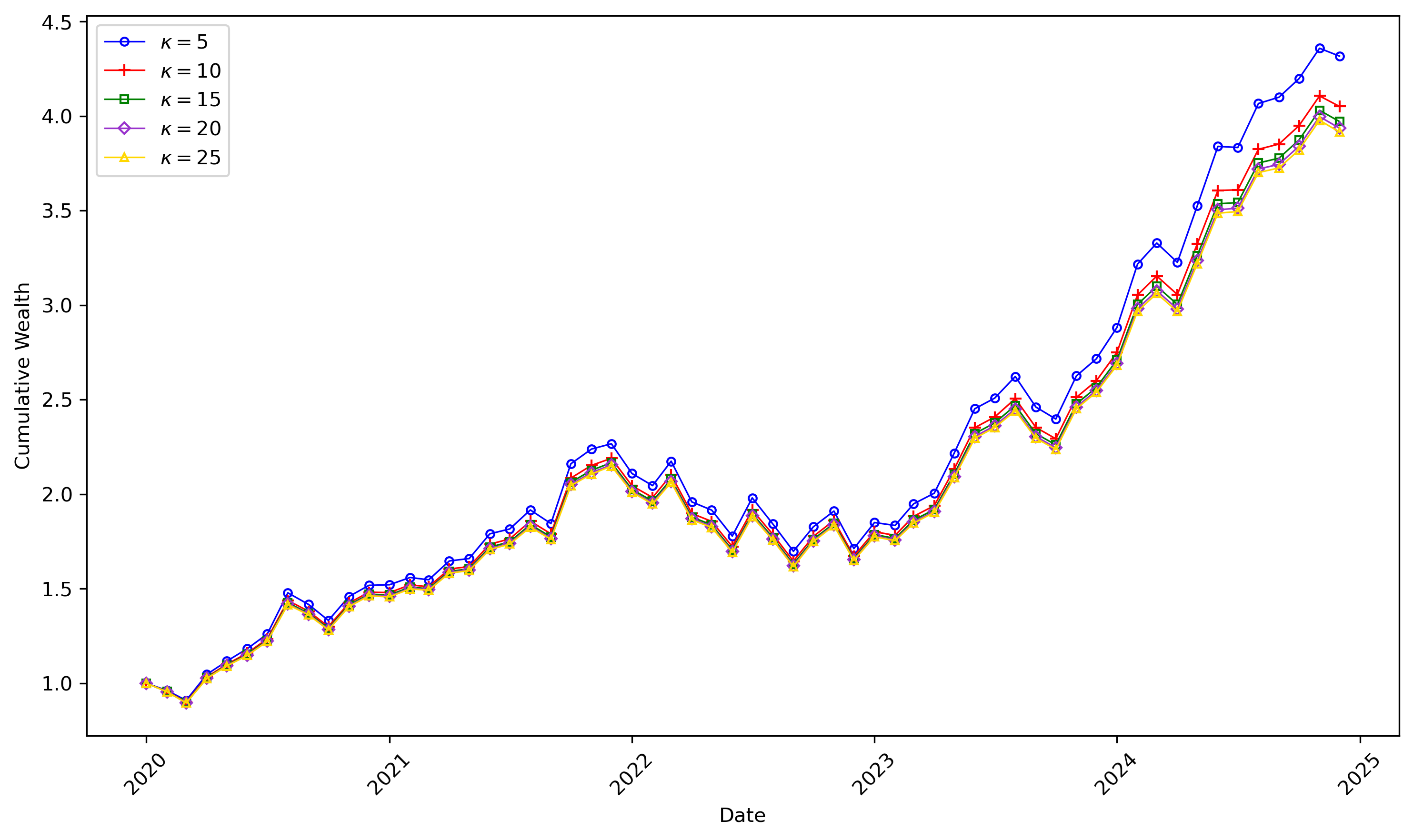}
        \caption{Cumulative wealth of the M-ER model without short selling}
    \end{subfigure}
    \begin{subfigure}{0.45\textwidth}
        \centering
        \includegraphics[width=\linewidth]{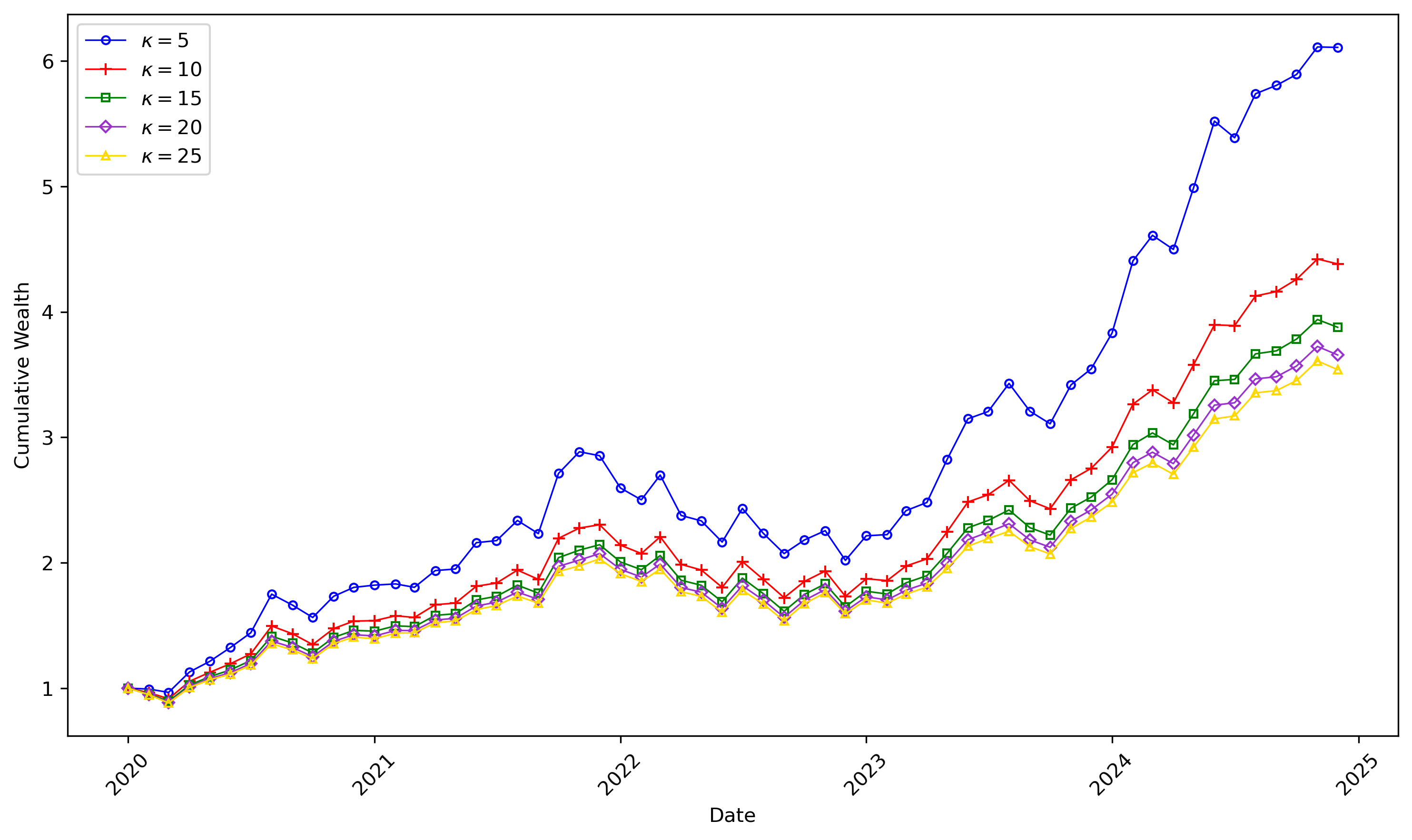}
        \caption{Cumulative wealth of the M-TSV model without short selling}
    \end{subfigure}
    
    \begin{subfigure}{0.45\textwidth}
        \centering
        \includegraphics[width=\linewidth]{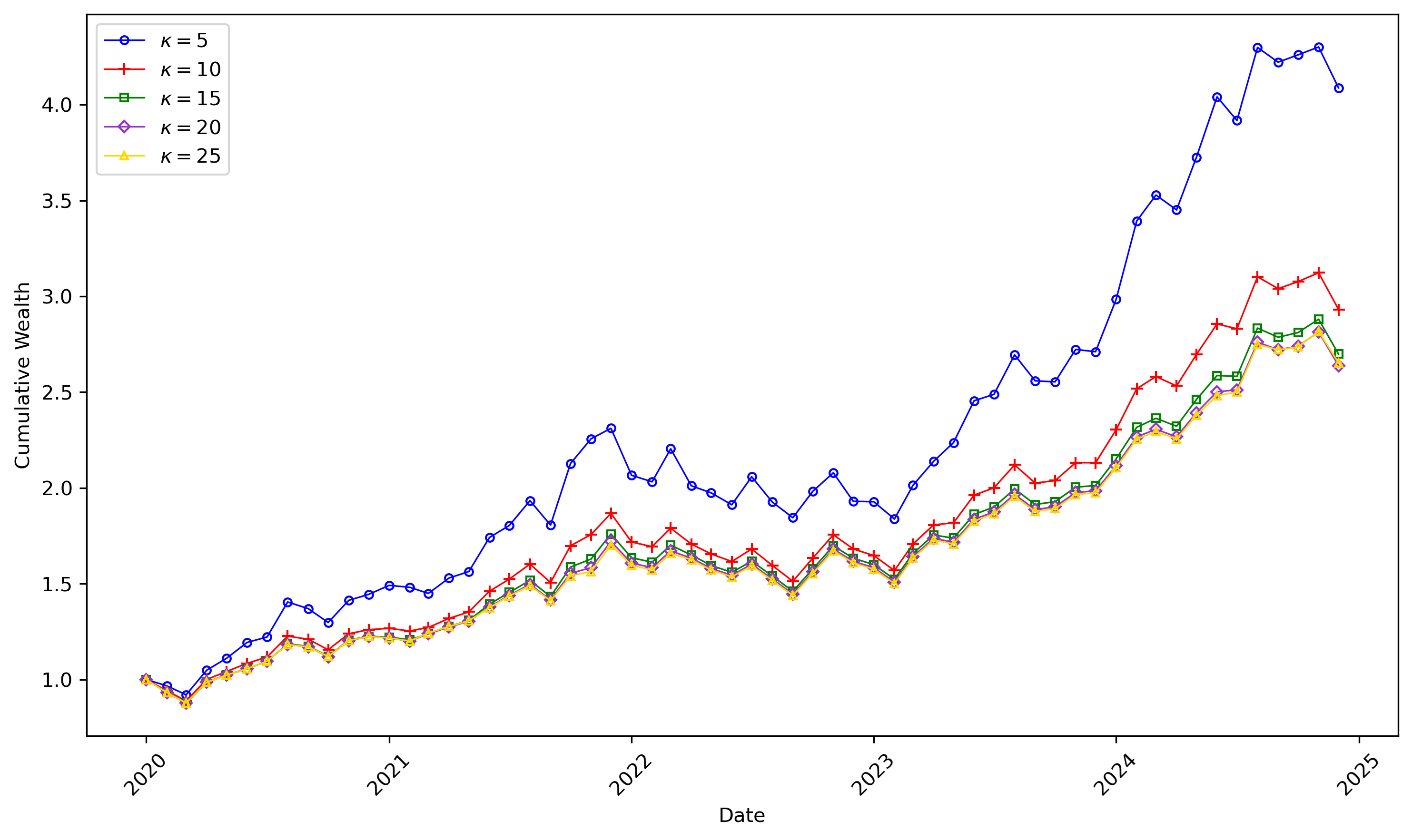}
        \caption{Cumulative wealth of the M-V model without short selling}
    \end{subfigure}
    \begin{subfigure}{0.45\textwidth}
        \centering
        \includegraphics[width=\linewidth]{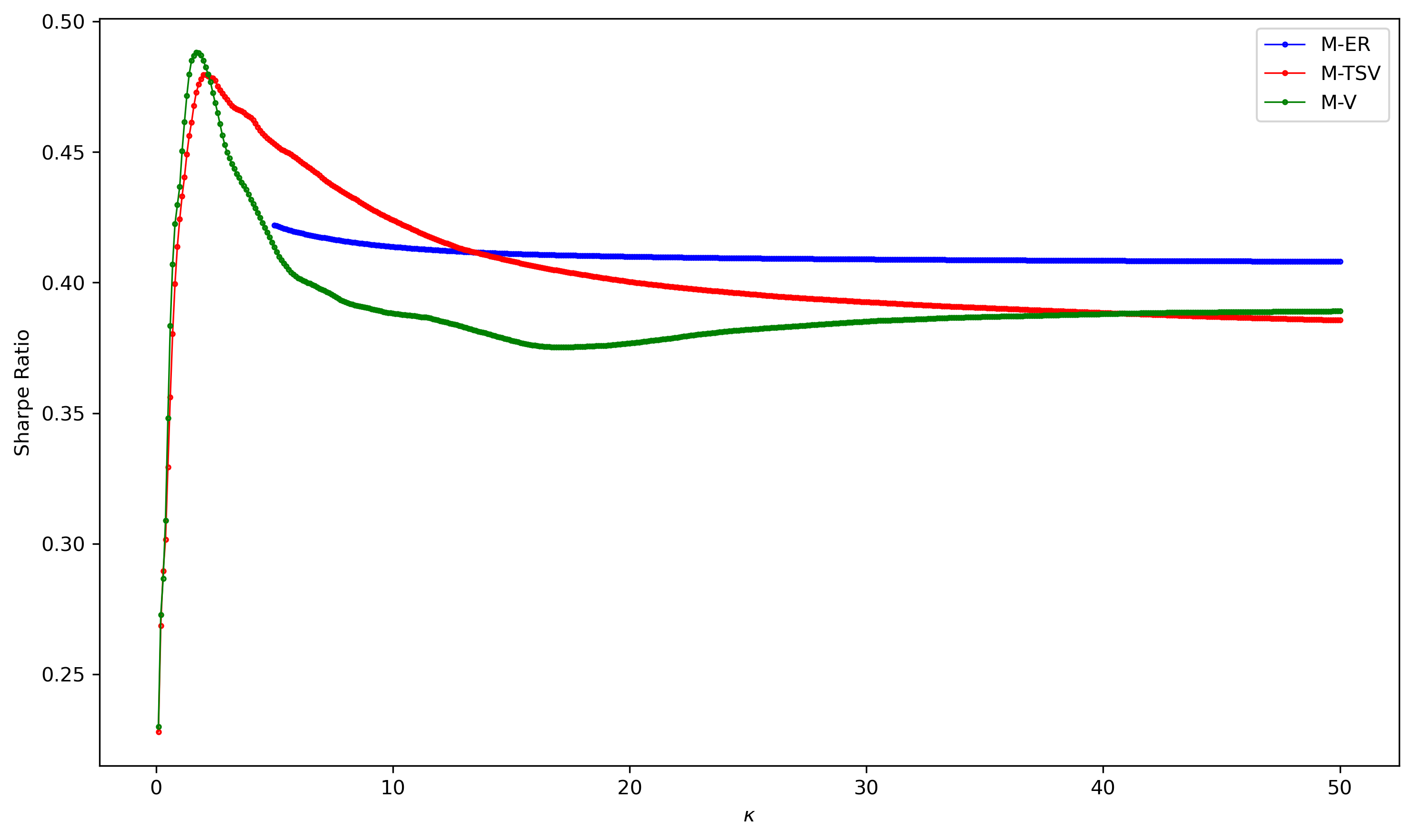}
        \caption{Sharpe ratio of the three models without short selling}
    \end{subfigure}
    
    \caption{Cumulative wealth and Sharpe ratios of portfolios based on the M-ER, M-TSV, and M-V models 
   under the uncertainty set ${\cal M}_{n+1}(\bm{\mu}, \bm{\Sigma})$,  
    without short selling, plotted against the weighting  coefficient $\kappa$, for  $\alpha=-0.04$ and $r_0=0.02$.       }
    \label{fig: CW-Sharpe-diff-k-M-no-SS}
\end{figure}

\begin{figure}[H]
    \centering
    \begin{subfigure}{0.45\textwidth}
        \centering
        \includegraphics[width=\linewidth]{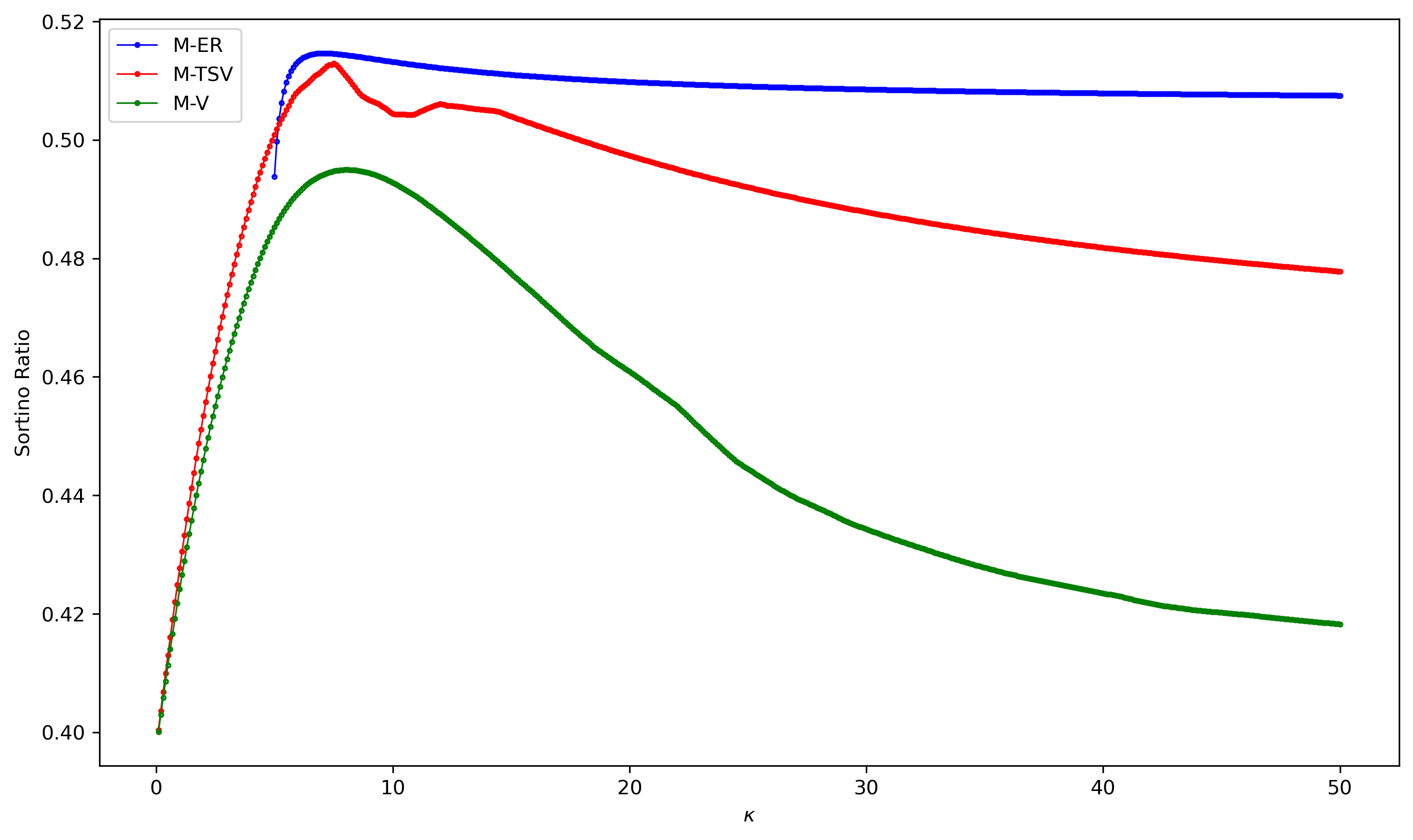}
        \caption{Sortino ratios of the three models with short selling}
    \end{subfigure}
    \begin{subfigure}{0.45\textwidth}
        \centering
        \includegraphics[width=\linewidth]{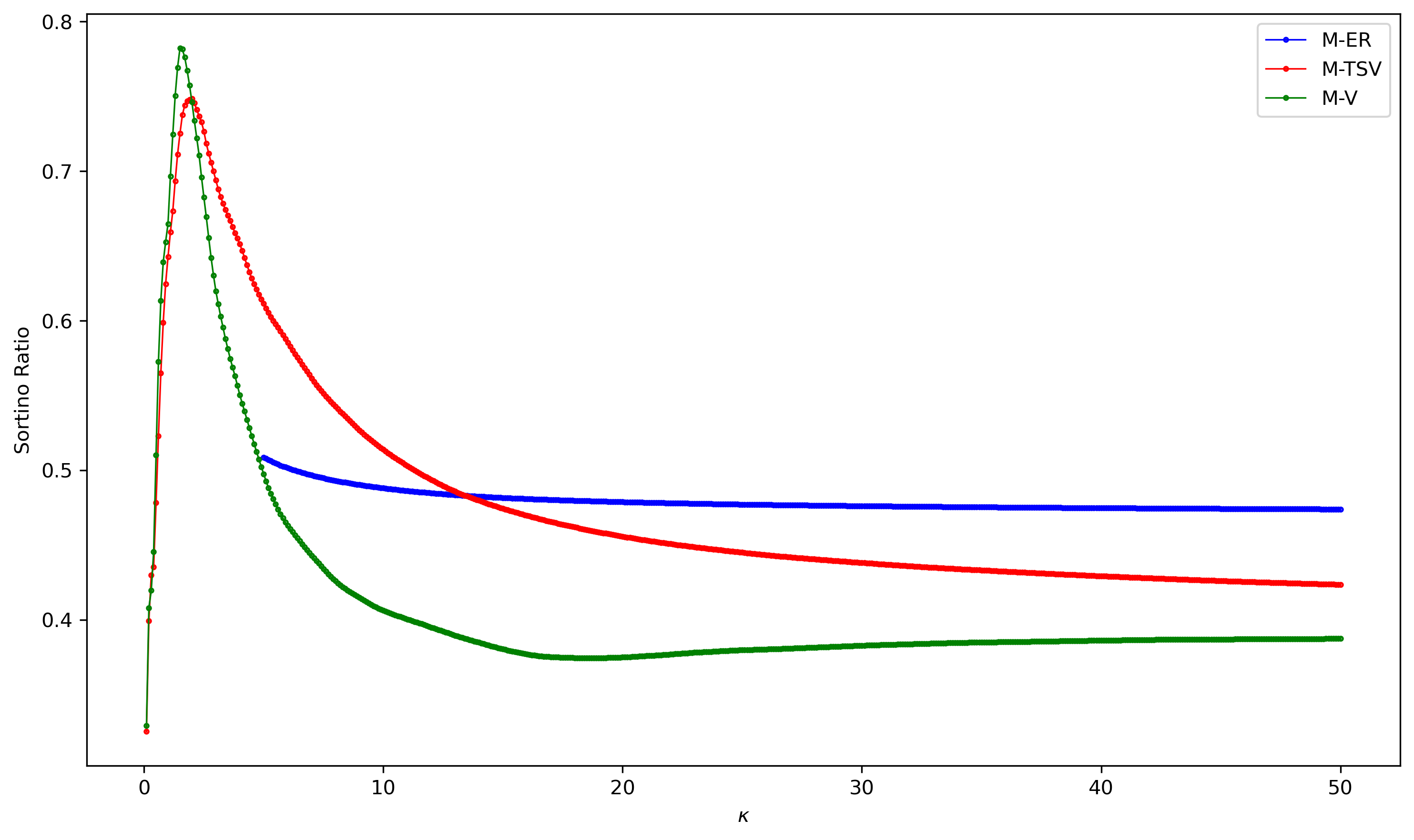}
        \caption{Sortino ratios of the three models without short selling}
    \end{subfigure}
    
    \begin{subfigure}{0.45\textwidth}
        \centering
        \includegraphics[width=\linewidth]{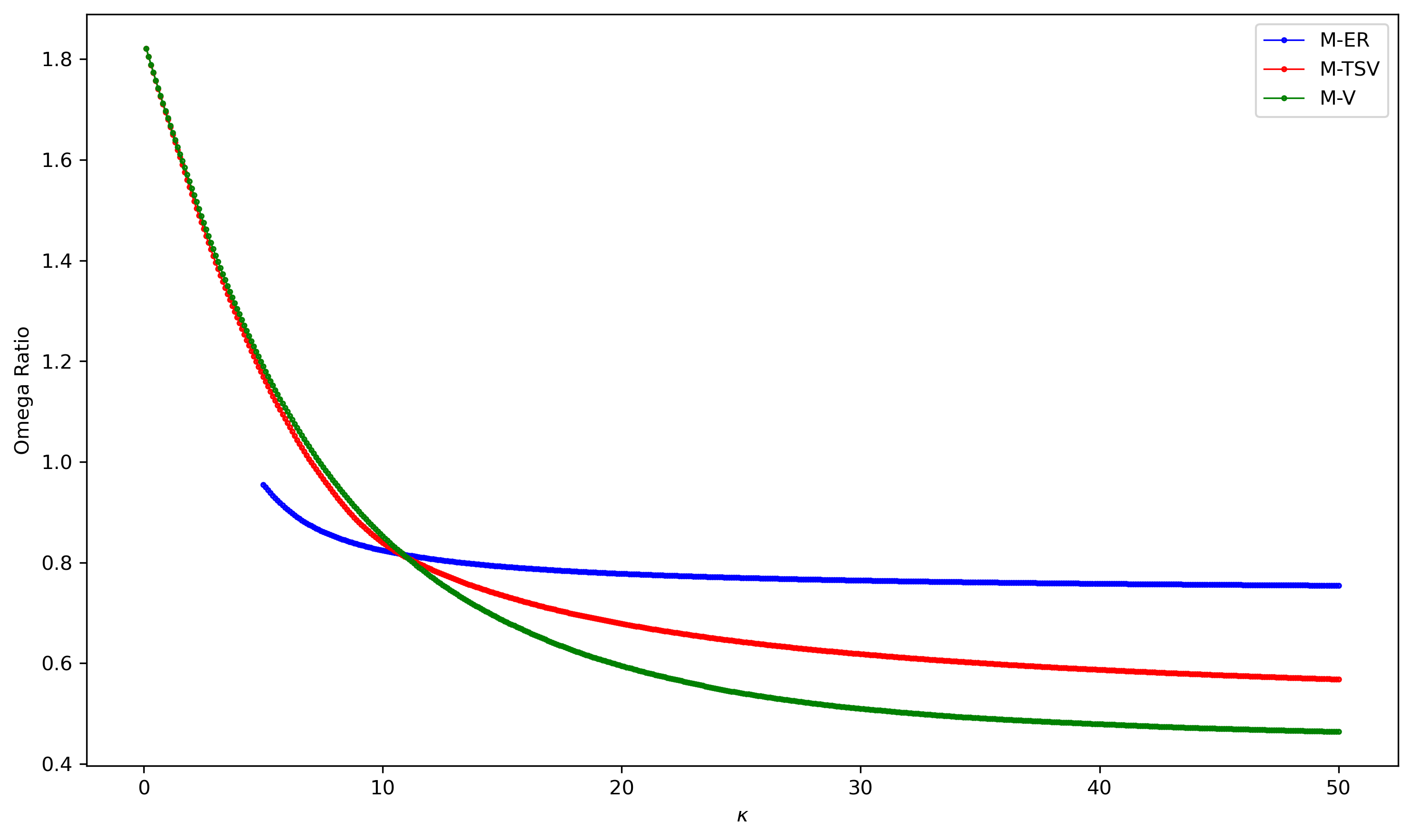}
        \caption{Omega ratios of the three models with short selling}
    \end{subfigure}
    \begin{subfigure}{0.45\textwidth}
        \centering
        \includegraphics[width=\linewidth]{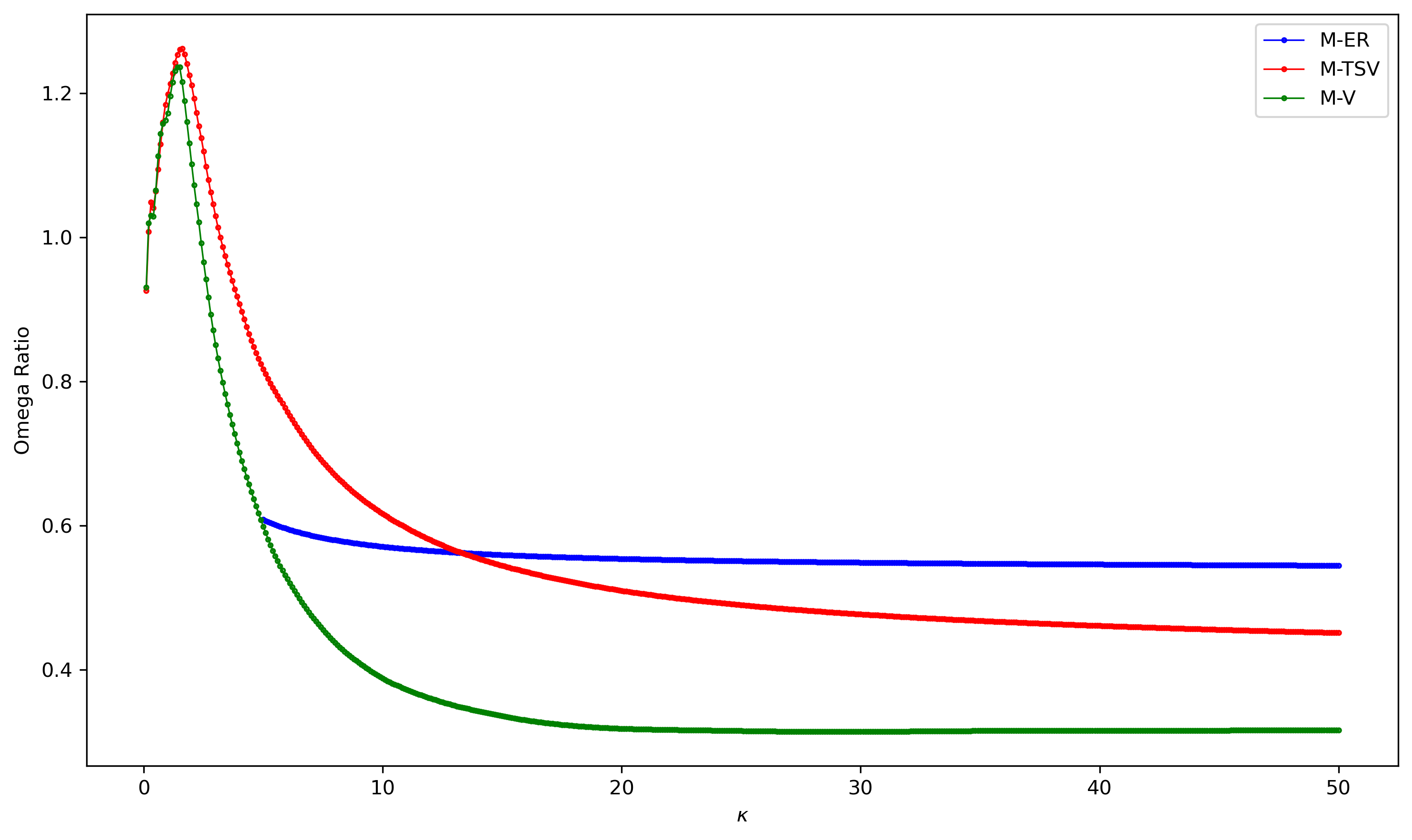}
        \caption{Omega ratios of M-TSV and M-V models without short selling}
    \end{subfigure}

    \captionsetup[figure]{font=normalsize}
\caption{Sortino and Omega ratios of portfolios based on the M-ER, M-TSV, and M-V models under the uncertainty set ${\cal M}_{n+1}(\bm{\mu}, \bm{\Sigma})$,  with and without short selling, plotted against the weighting  coefficient $\kappa$, for $\alpha=-0.04$ and $r_0=0.02$.   }    
     \label{fig: S-Q-diff-k-M}
\end{figure}

 \subsubsection{Optimal weight selection under a generalized uncertainty set} 
 \label{sec-optimal-kappa-D}

We now assume that the joint distribution of the stock and index losses belongs to the uncertainty set 
${\cal D}_{n+1}(\gamma_1, \gamma_2)$.  Under this uncertainty set,
 the out-of-sample cumulative wealth and Sharpe ratios of the M-ER, M-TSV, and M-V models are presented in 
 Figures \ref{fig: CW-Sharpe-diff-k-D-with-SS} and \ref{fig: CW-Sharpe-diff-k-D-no-SS}, while the out-of-sample Sortino ratios 
 and Omega ratios are reported in Figure \ref{fig: S-O-ratio-diff-k-D}.

From these figures, we observe that under the uncertainty set ${\cal D}_{n+1}(\gamma_1, \gamma_2)$, the 
out-of-sample cumulative wealth of the M-ER, M-TSV, and M-V models increase as $\kappa$ increases. 
However, for each model, the differences in cumulative wealth across different values of $\kappa$ are much smaller than those 
observed under the uncertainty set ${\cal M}_{n+1}(\bm{\mu}, \bm{\Sigma})$. In other words, 
the out-of-sample cumulative wealth under ${\cal D}_{n+1}(\gamma_1, \gamma_2)$ is more robust 
to the changes in weighting parameter $\kappa$ than under ${\cal M}_{n+1}(\bm{\mu}, \bm{\Sigma})$. 
In particular, for large values of $\kappa$, the out-of-sample cumulative wealth of each model stabilizes at an 
approximately constant level.

In addition, for each model, there exists an optimal value of the weighting parameter of $\kappa^*$ 
that maximizes the Sharpe ratio, both with and without short selling. Moreover,
 for the M-TSV model, there also exists an optimal  value of $\kappa^*$ that 
 maximizes the Sortino ratio, regardless of whether short selling is allowed.

Furthermore, for each of the M-ER, M-TSV, and M-V models, the Omega ratio decreases as $\kappa$
 increases. This indicates that small values of $\kappa$ are preferred to achieve both 
 higher cumulative wealth and a higher Omega ratio.

These results provide an empirical approach for selecting the optimal weighting  coefficient $\kappa^*$ under the uncertainty
 set ${\cal D}_{n+1}(\gamma_1, \gamma_2)$ for the M-ER, M-TSV, and M-V models.

\begin{figure}[H]
    \centering
    \begin{subfigure}{0.45\textwidth}
        \centering
        \includegraphics[width=\linewidth]{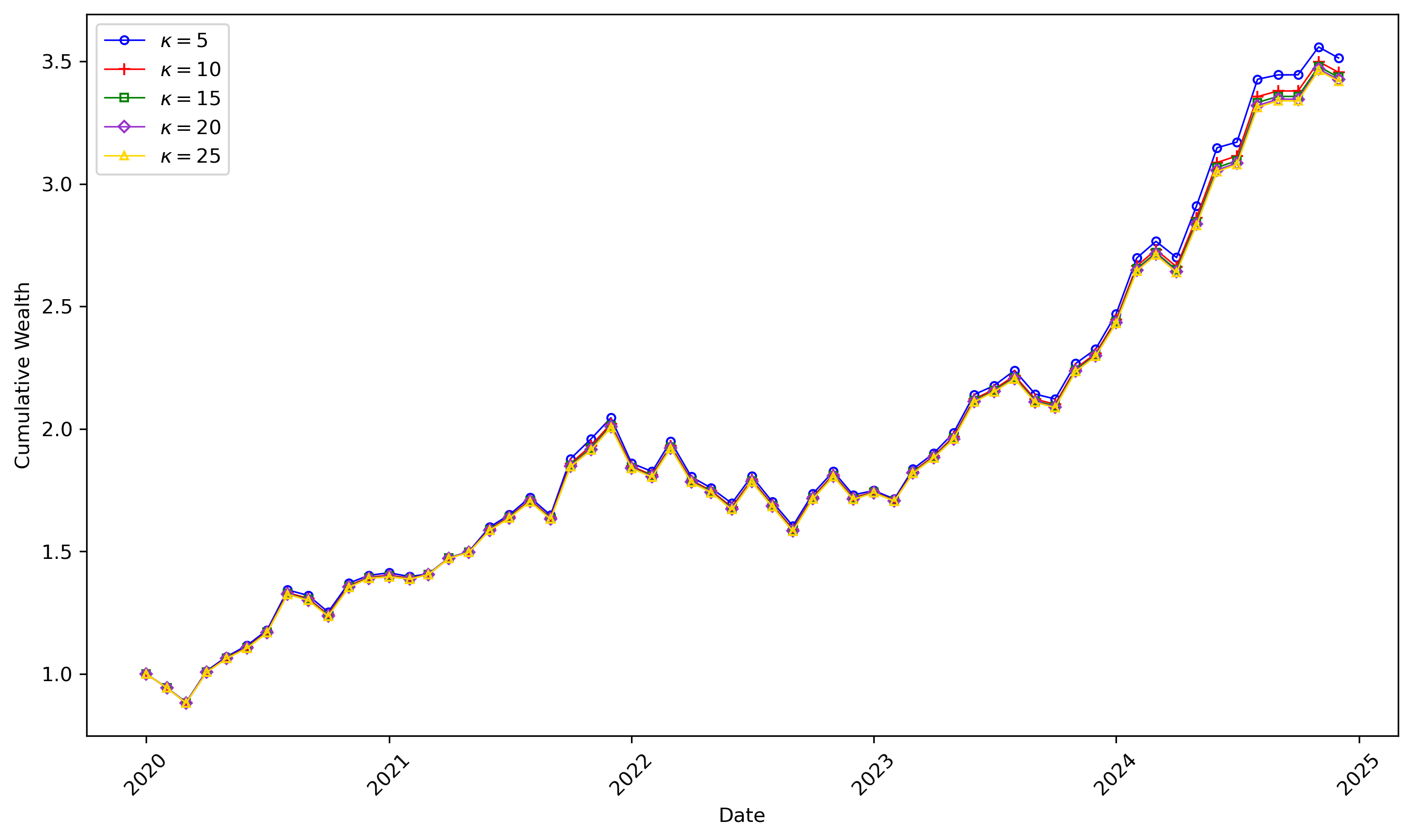}
        \caption{Cumulative wealth of the M-ER model with short selling} 
    \end{subfigure}
    \begin{subfigure}{0.45\textwidth}
        \centering
        \includegraphics[width=\linewidth]{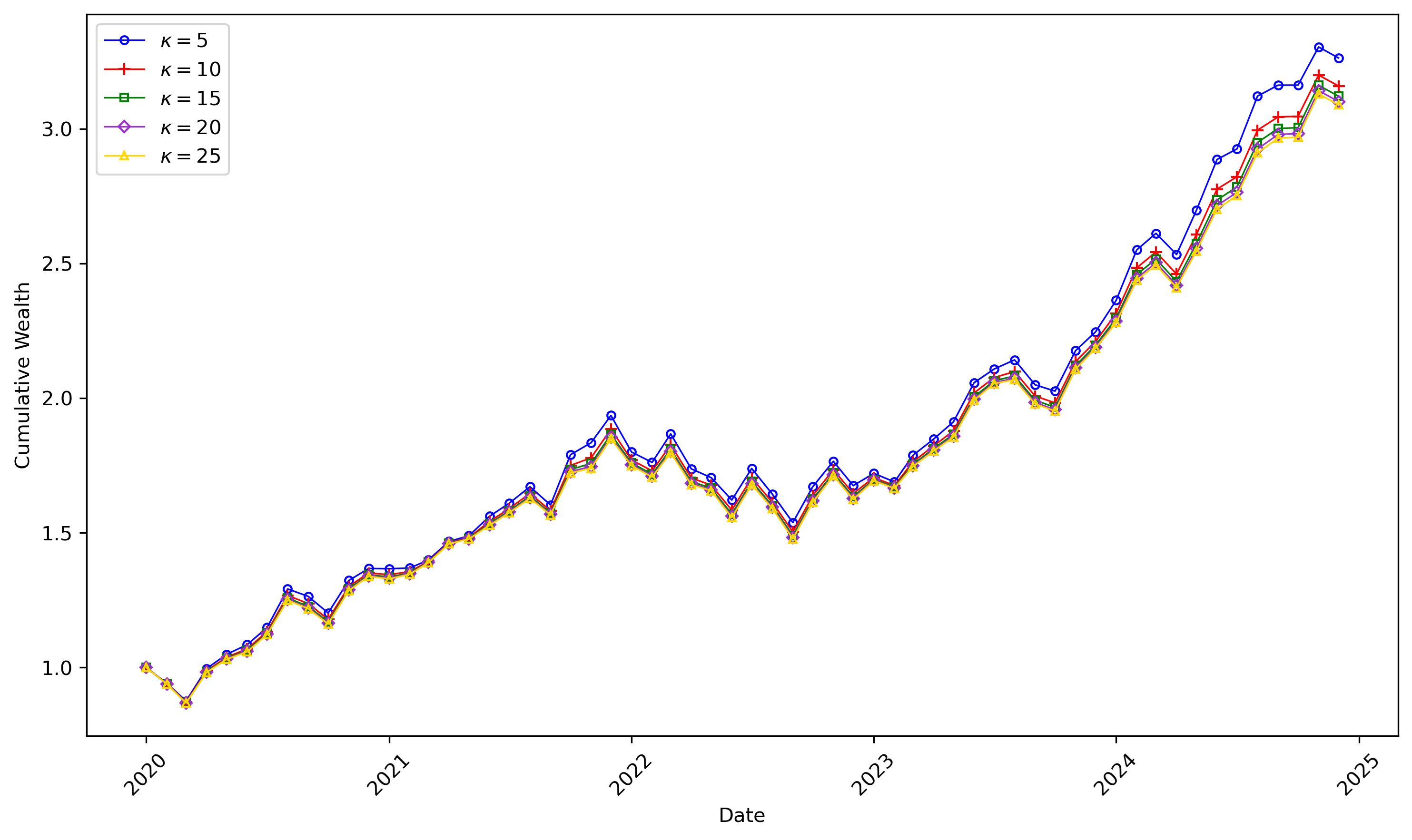}
        \caption{Cumulative wealth of the M-TSV model with short selling}
    \end{subfigure}
    
    \begin{subfigure}{0.45\textwidth}
        \centering
        \includegraphics[width=\linewidth]{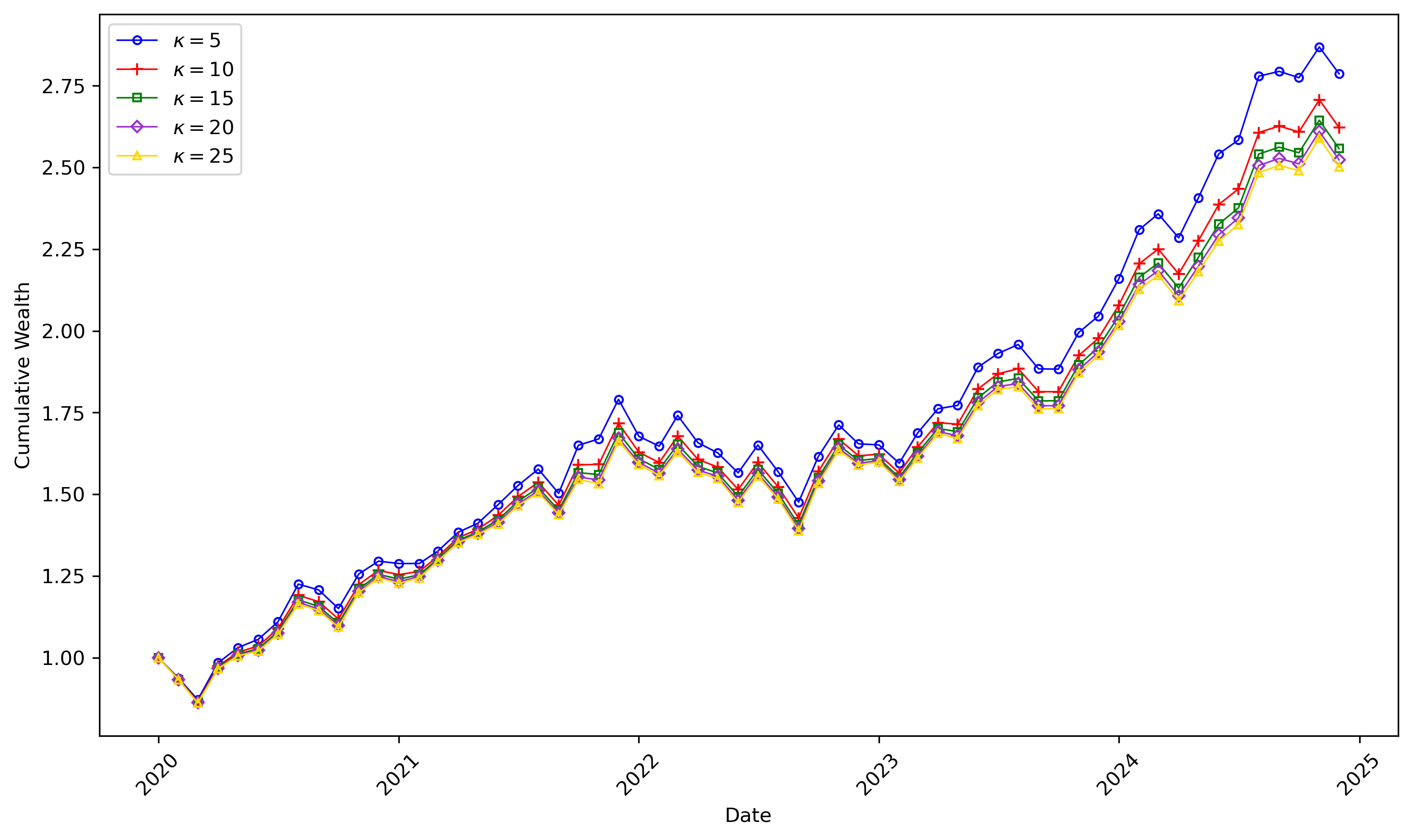}
        \caption{Cumulative wealth of the M-V model with short selling}
    \end{subfigure}
    \begin{subfigure}{0.45\textwidth}
        \centering
        \includegraphics[width=\linewidth]{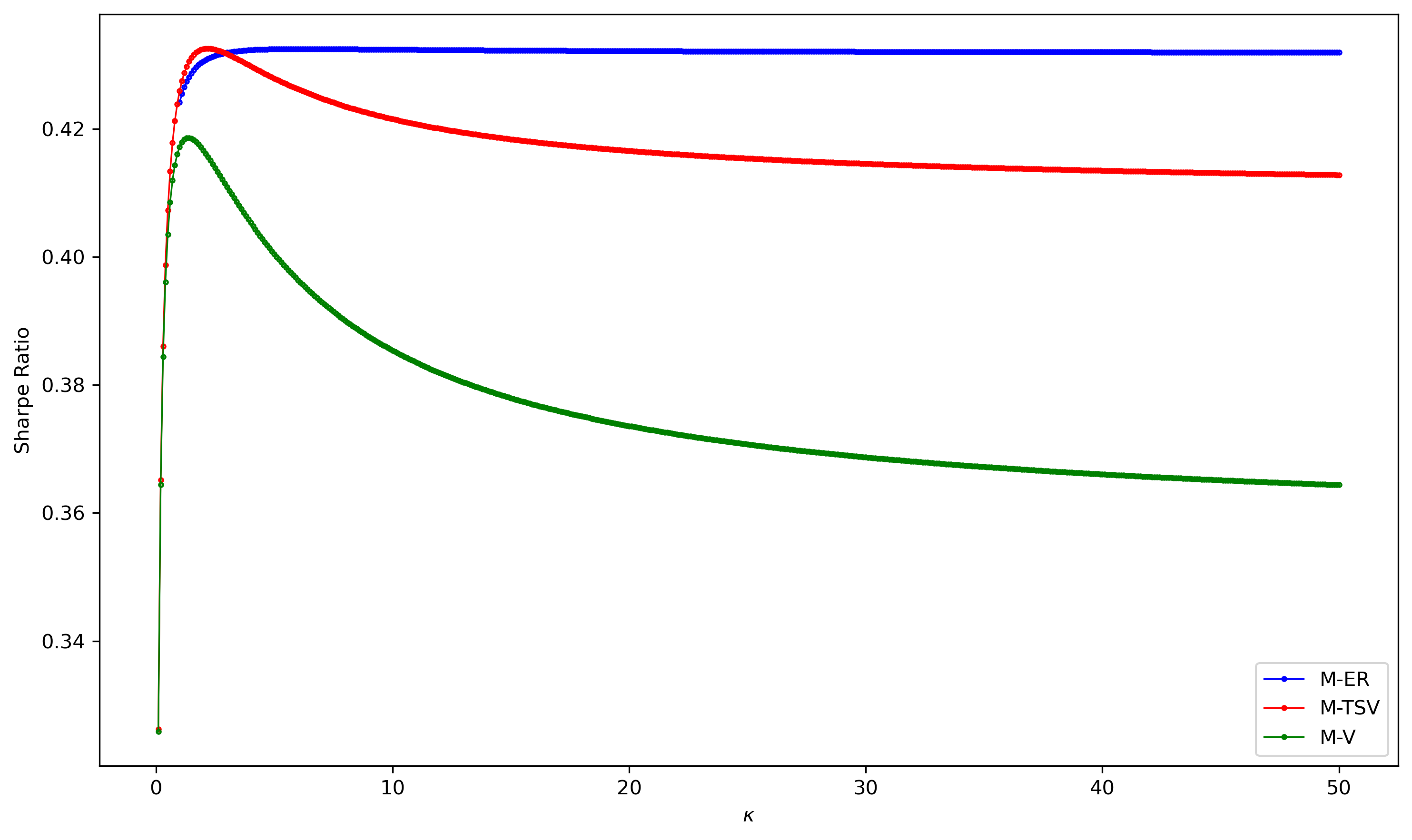}
        \caption{Sharpe ratios of the three models with short selling}
    \end{subfigure}
    
    \caption{The cumulative wealth  and Sharpe ratios 
    of portfolios based on the M-ER, M-TSV and M-V models under the uncertainty set ${\cal D}_{n+1}(\gamma_1, \gamma_2)$,   with short selling, plotted against the risk aversion coefficient $\kappa$, for
     $\alpha=-0.04$, $r=-0.01$,  $\gamma_1=0.21$, and  $\gamma_2=0.025$.
        }
    \label{fig: CW-Sharpe-diff-k-D-with-SS}
\end{figure} 

\begin{figure}[H]
    \centering
    \begin{subfigure}{0.45\textwidth}
        \centering
        \includegraphics[width=\linewidth]{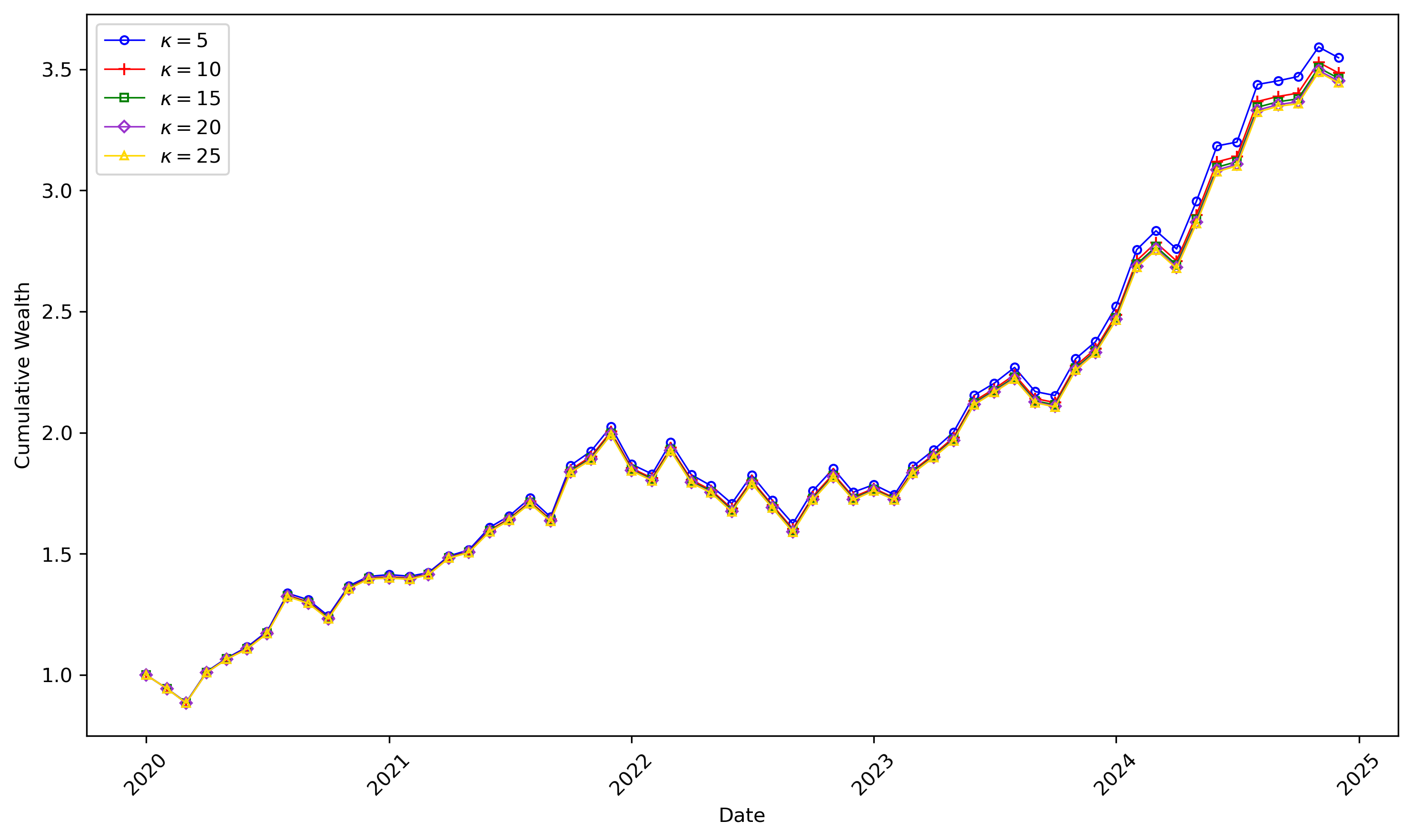}
        \caption{Cumulative wealth of the M-ER model without short selling}
    \end{subfigure}
    \begin{subfigure}{0.45\textwidth}
        \centering
        \includegraphics[width=\linewidth]{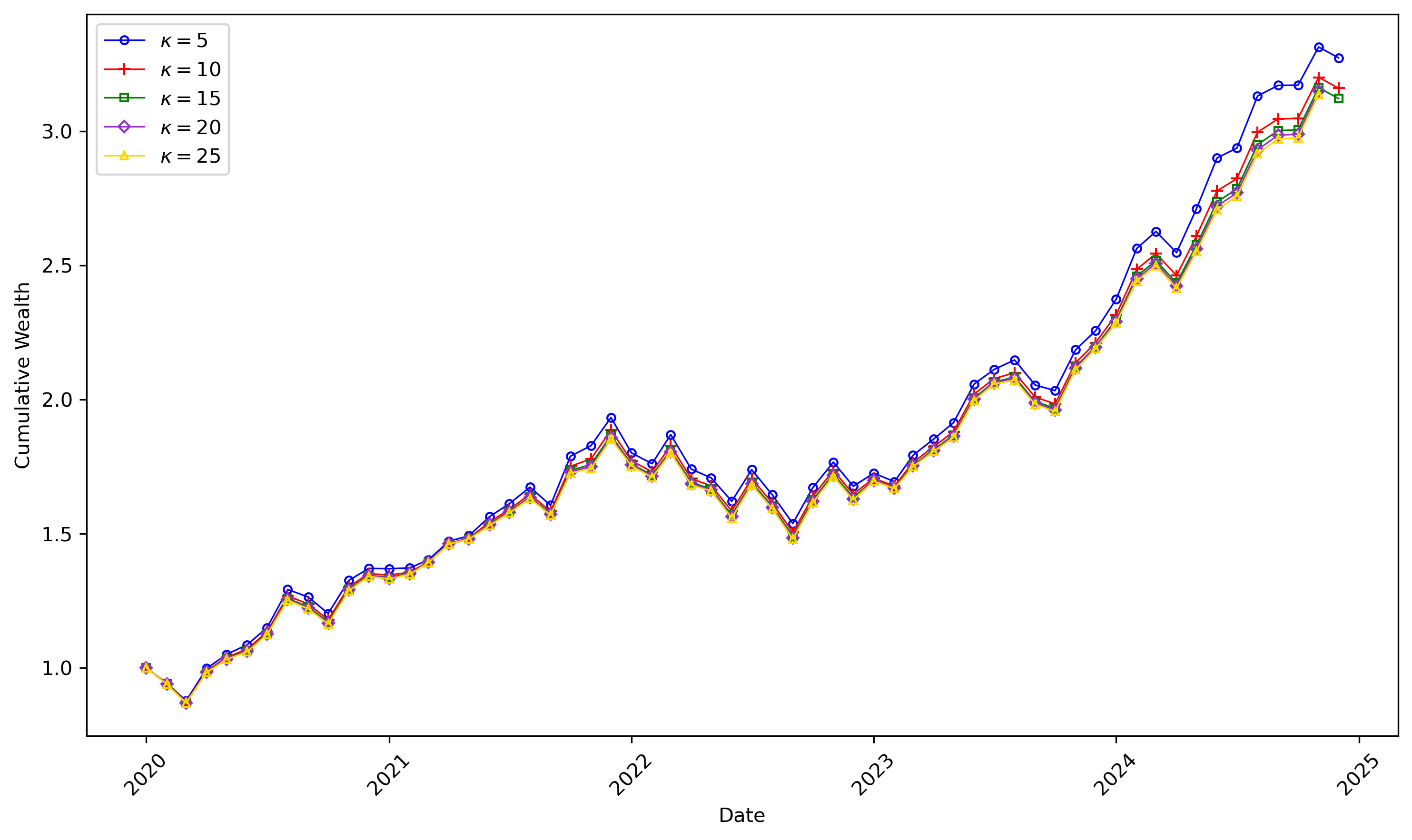}
        \caption{Cumulative wealth of the M-TSV model without short selling}
    \end{subfigure}
    
    \begin{subfigure}{0.45\textwidth}
        \centering
        \includegraphics[width=\linewidth]{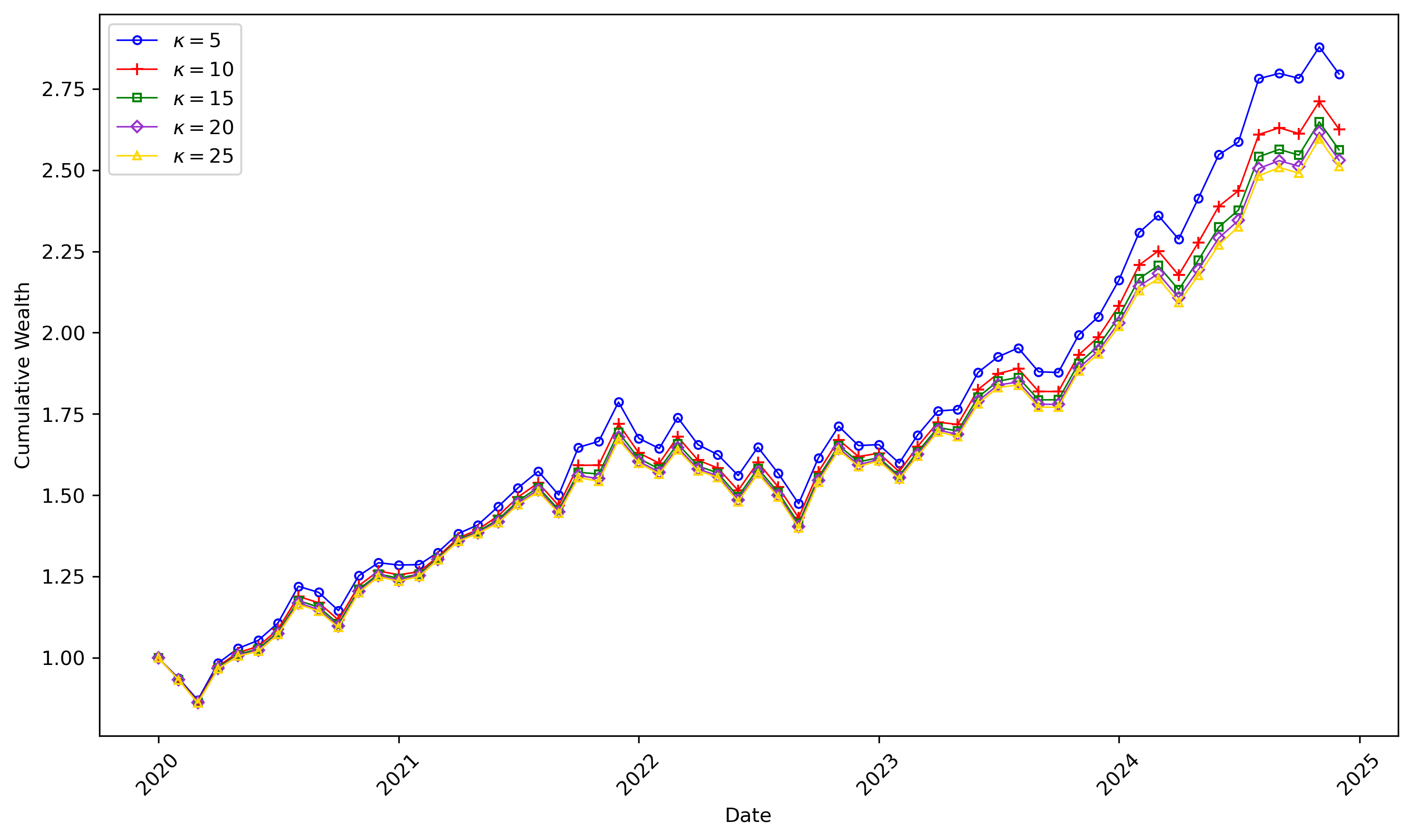}
        \caption{Cumulative wealth of the M-V model without short selling}
    \end{subfigure}
    \begin{subfigure}{0.45\textwidth}
        \centering
        \includegraphics[width=\linewidth]{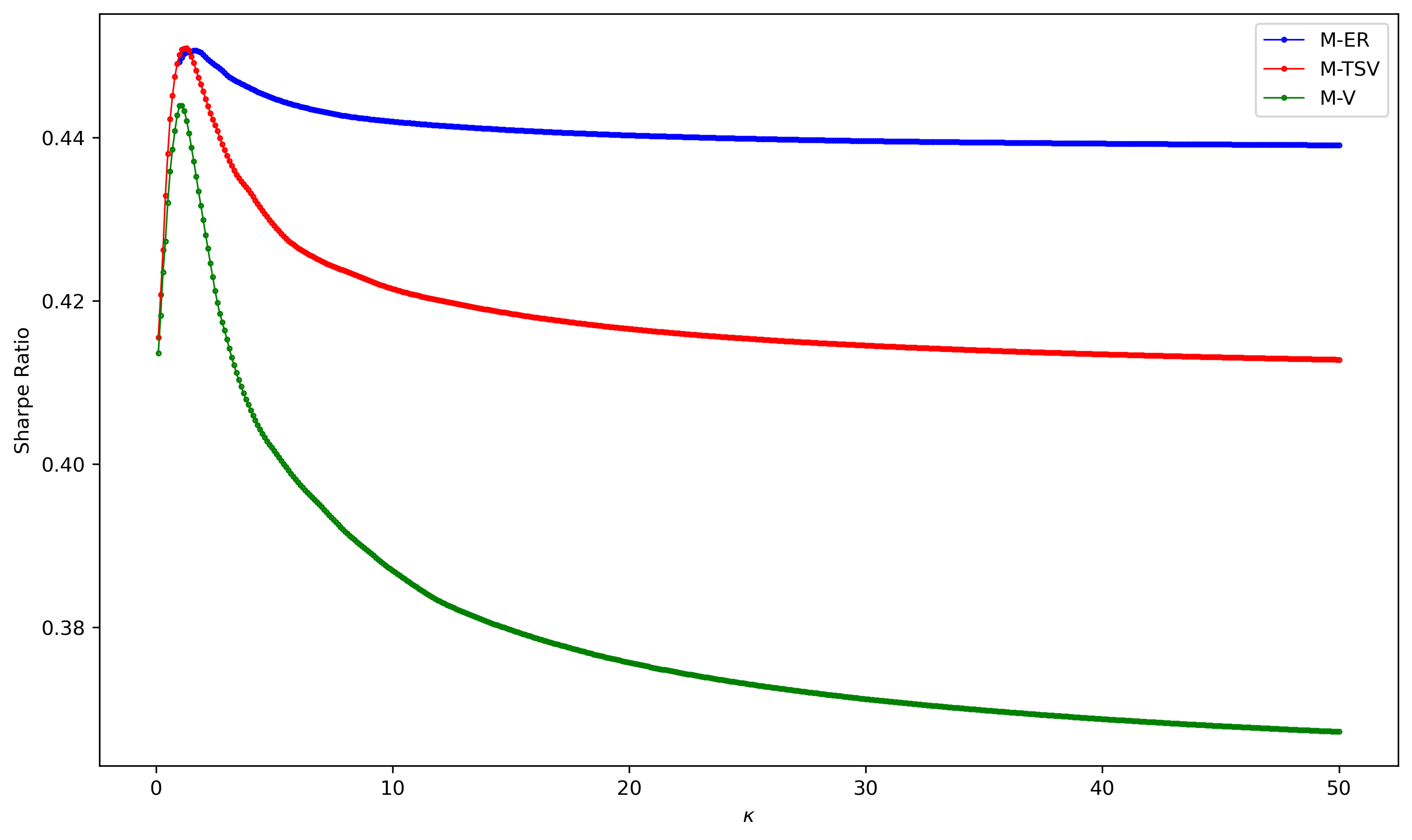}
        \caption{Sharpe Ratio of the three models without short selling}
    \end{subfigure}
    
    \caption{Cumulative wealth  and Sharpe ratios 
    of portfolios based on the M-ER, M-TSV and M-V models  under the uncertainty set ${\cal D}_{n+1}(\gamma_1, \gamma_2)$,   without short selling, plotted against the weighting  coefficient $\kappa$, for 
     $\alpha=-0.04$, $r=-0.01$,  $\gamma_1=0.21$, and  $\gamma_2=0.025$.
        }
    \label{fig: CW-Sharpe-diff-k-D-no-SS}
\end{figure} 

\begin{figure}[H]
    \centering
    \begin{subfigure}{0.45\textwidth}
        \centering
        \includegraphics[width=\linewidth]{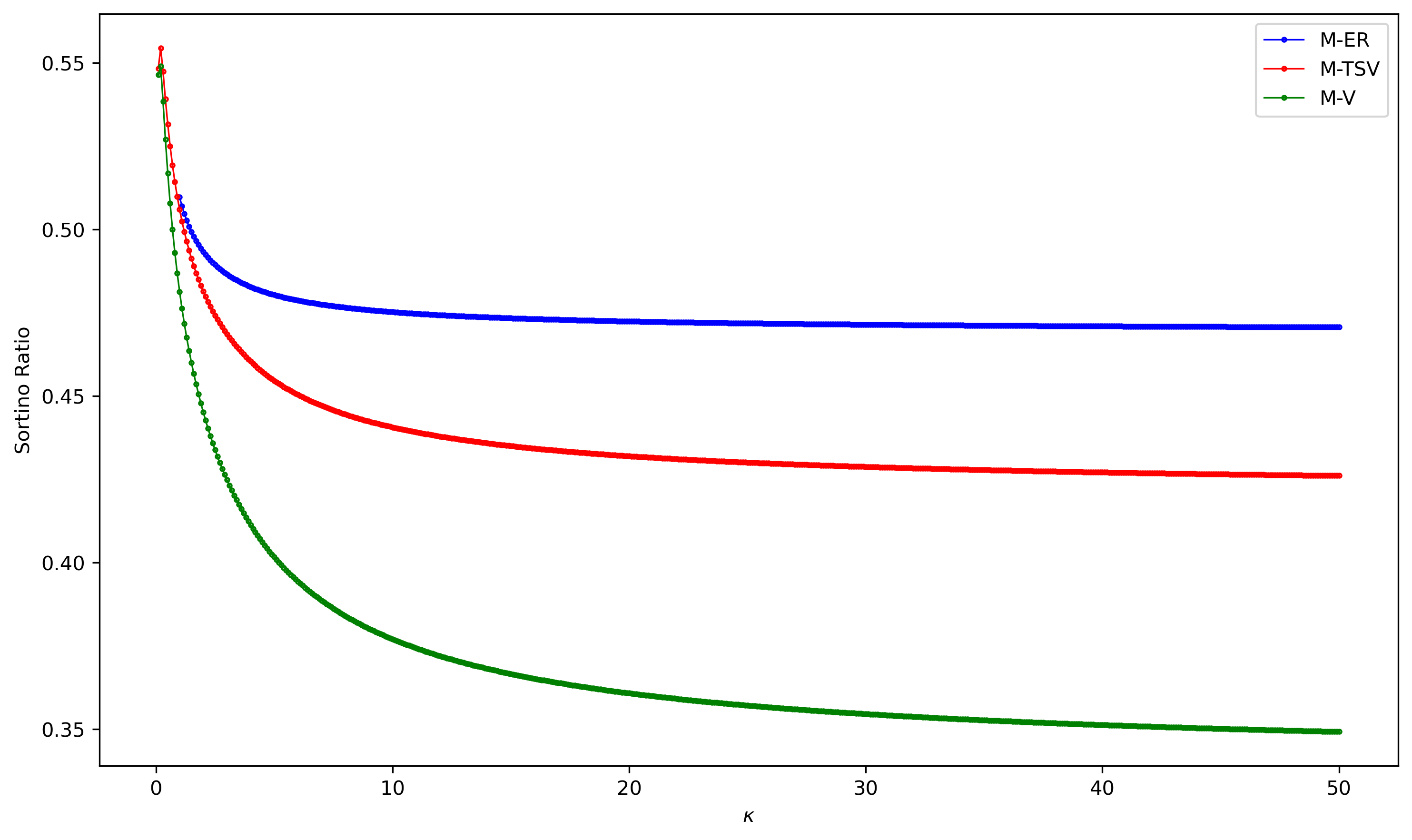}
        \caption{Sortino ratios of the three models with short selling }
    \end{subfigure}
    \begin{subfigure}{0.45\textwidth}
        \centering
        \includegraphics[width=\linewidth]{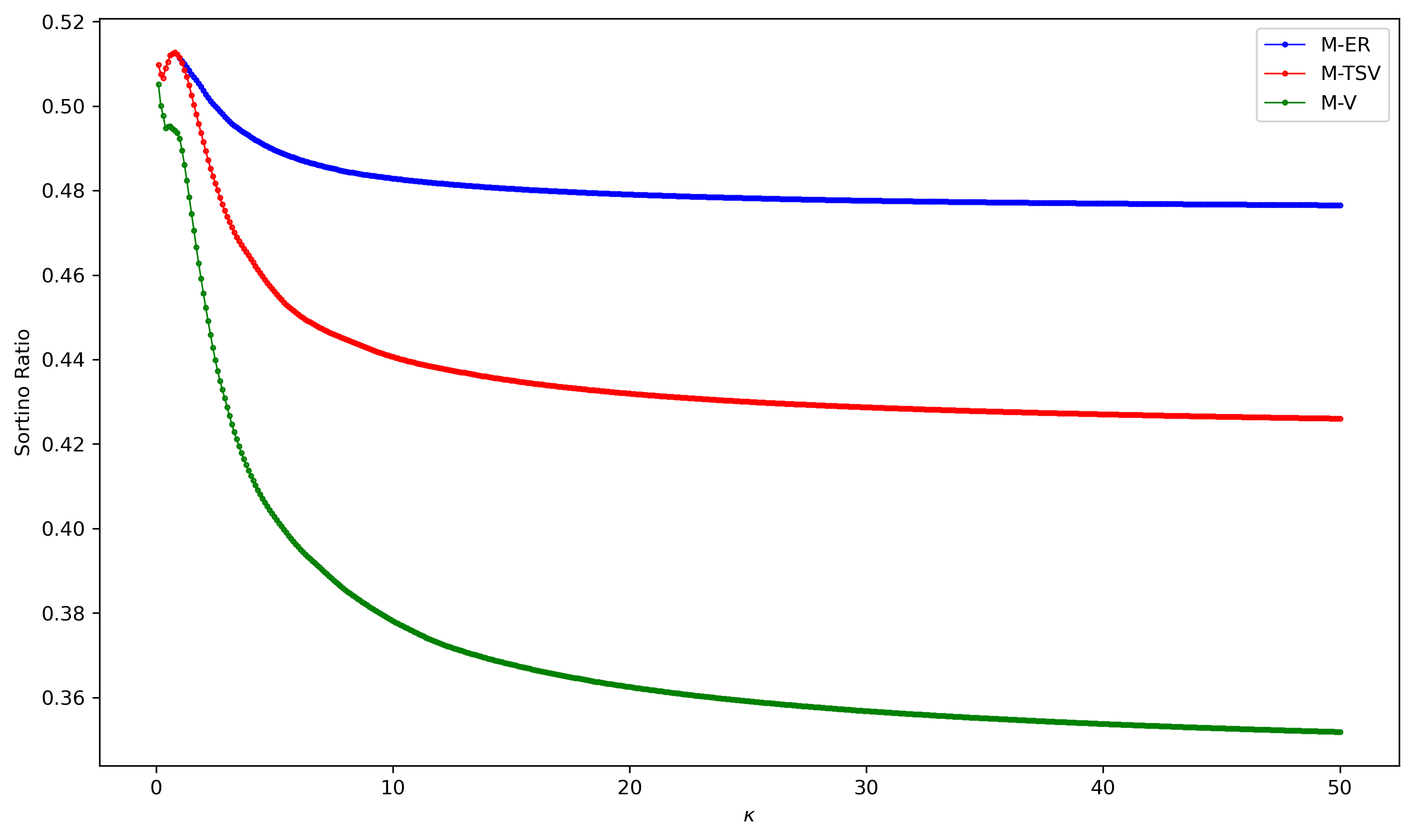}
        \caption{Sortino ratios of the three models without short selling}
    \end{subfigure}
    
    \begin{subfigure}{0.45\textwidth}
        \centering
        \includegraphics[width=\linewidth]{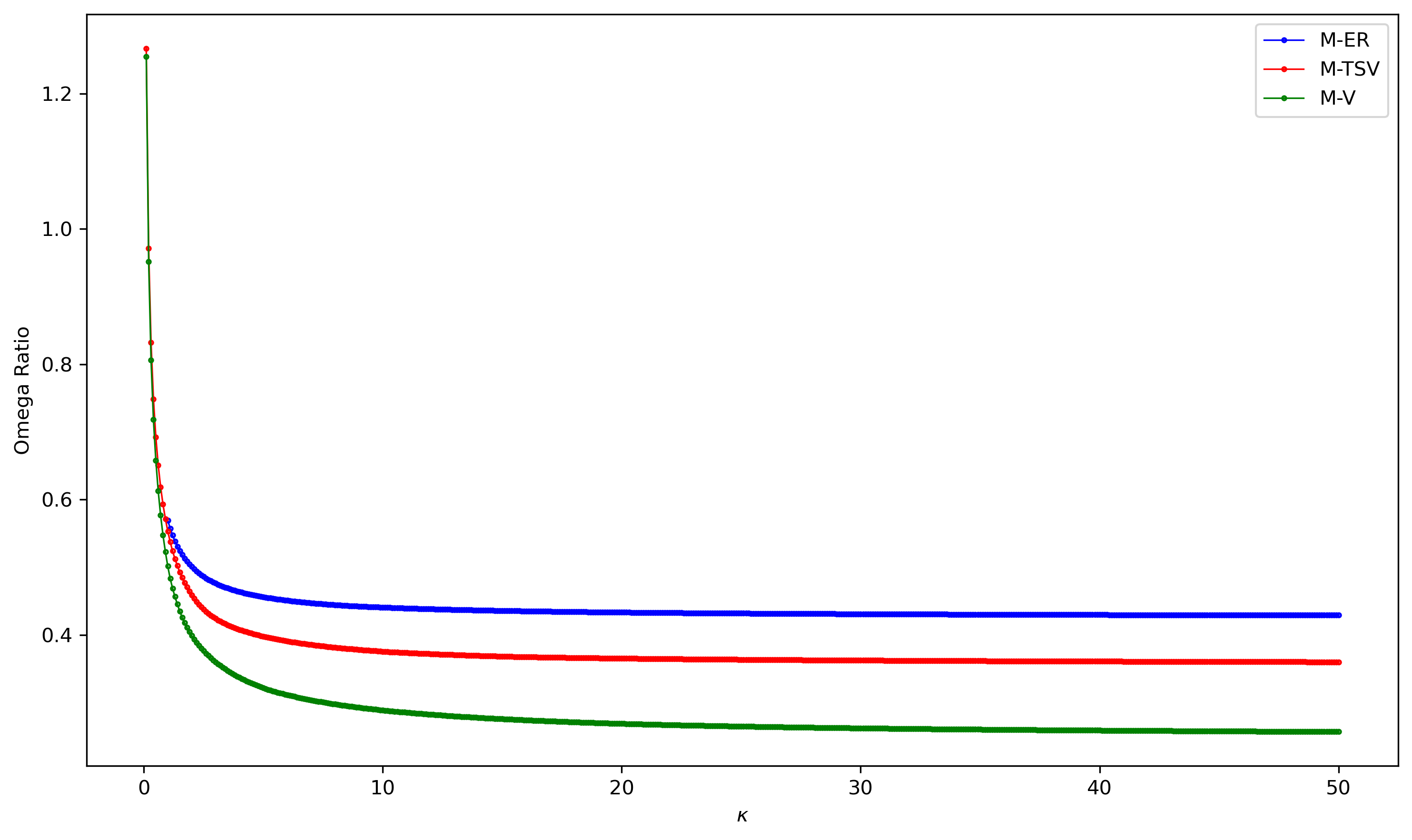}
        \caption{Omega ratios of the three models with short selling}
    \end{subfigure}
    \begin{subfigure}{0.45\textwidth}
        \centering
        \includegraphics[width=\linewidth]{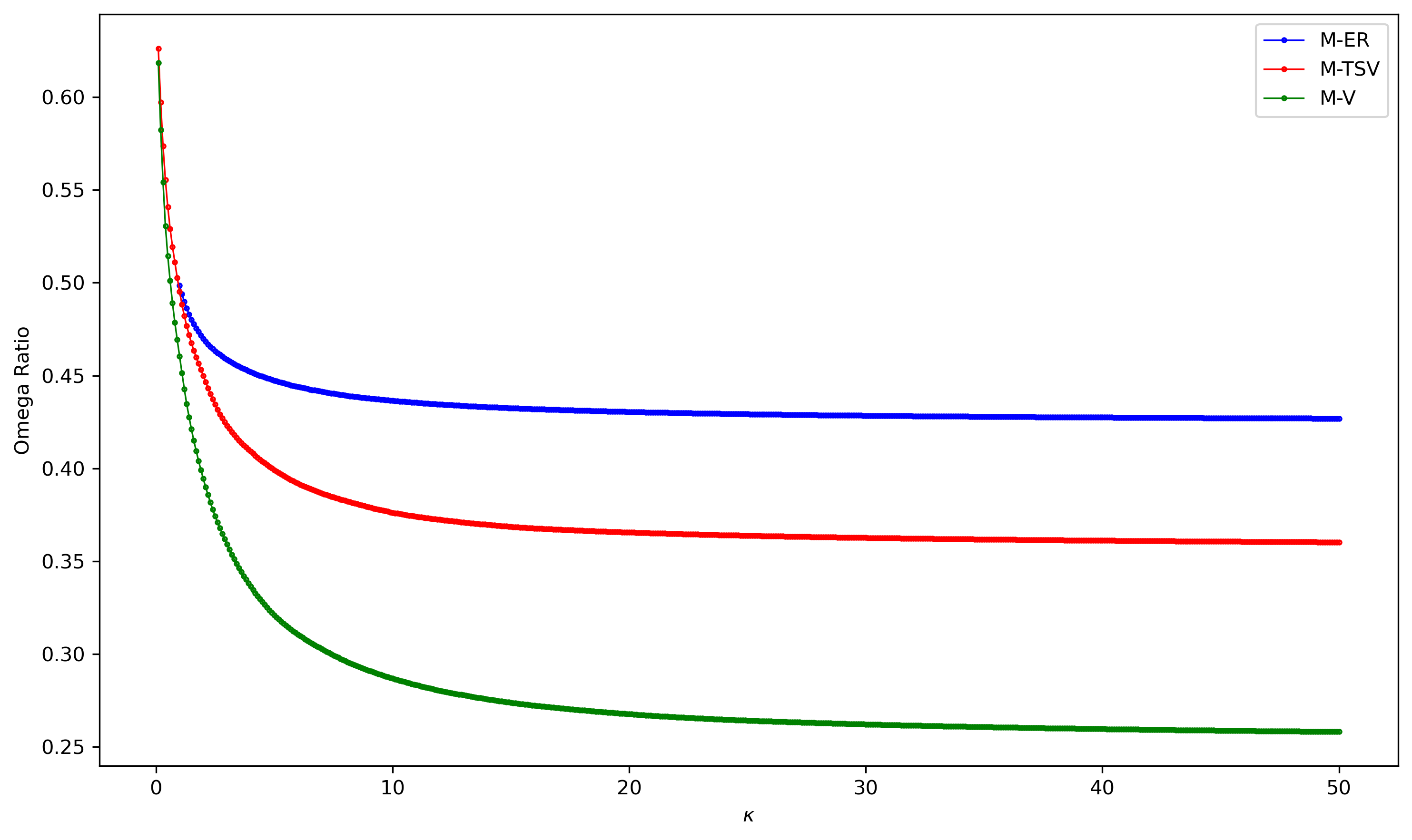}
        \caption{Omega ratios of the three models without short selling} 
    \end{subfigure}

    \caption{Sortino  and Omega ratios 
    of portfolios based on the M-ER, M-TSV and M-V models under the uncertainty set ${\cal D}_{n+1}(\gamma_1, \gamma_2)$, 
      with and  without short selling, plotted against the weighting  coefficient $\kappa$, for
     $\alpha=-0.04$, $r=-0.01$,  $\gamma_1=0.21$, and  $\gamma_2=0.025$.
     }
    \label{fig: S-O-ratio-diff-k-D}
\end{figure}

\section{Conclusion}
\label{sec-6}

 In this paper, we propose robust enhanced index tracking (EIT) portfolio selection models under distributional uncertainty 
in the joint distribution of asset and index losses. Our objective is to balance the trade-off between mean loss and 
downside risk (or variance) in portfolio selection. Based on Pareto optimality, we develop models that 
minimize a convex combination of worst-case mean loss and downside risk (or variance), subject to a 
constraint on worst-case mean return. The resulting portfolios are Pareto optimal, meaning that no 
alternative achieves both lower mean loss and lower downside risk (or variance).

Empirical results based on real market data show that the proposed models outperform the equally weighted portfolio, the robust mean-variance model, and the robust variance model in terms of out-of-sample cumulative wealth, as well as Sharpe, Sortino, and Omega ratios. These metrics are widely used to assess a portfolio's ability to balance mean return and downside risk (or variance), where downside risk is measured by the expected regret and target semi-variance. Specifically, portfolios from models that minimize a convex combination of worst-case mean loss and downside risk (or variance) consistently outperform those from models that minimize worst-case downside risk (or variance) alone.

The results further demonstrate the effectiveness of the proposed EIT models in balancing out-of-sample 
mean return and downside risk (or variance). In addition, we provide empirical guidance for 
selecting optimal weighting coefficients to maximize Sharpe, Sortino, and Omega ratios.

\appendix\label{A}
 \section{Appendix} \label{appendix}

\begin{proof}[\textbf{Proof of Proposition \ref{pro-Pareto}}]
Suppose that   $\bm{\omega}^*_\kappa$ is not  a Pareto-optimal portfolio. Then, there exits a portfolio $\bm{\omega}_\kappa \in \cal W$ such that 
       $ \rho(\boldsymbol{\omega}_\kappa, \bm{X}, X_0) \le \rho(\boldsymbol{\omega}^*_\kappa, \bm{X}, X_0)$ and 
          $ \varphi(\boldsymbol{\omega}_\kappa, \bm{X}, X_0) \le \varphi(\boldsymbol{\omega}^*_\kappa, \bm{X}, X_0),$
    with  at least one of the inequalities  being  strict.  Therefore, it follows  that  
    $   
      \rho(\boldsymbol{\omega}_\kappa, \bm{X}, X_0) +\kappa \,  \varphi(\boldsymbol{\omega}_\kappa, \bm{X}, X_0)   <  
       \rho(\boldsymbol{\omega}^*_\kappa, \bm{X}, X_0) +\kappa \,  \varphi(\boldsymbol{\omega}^*_\kappa, \bm{X}, X_0),   
        $
   which contradicts the fact that    $ \bm{\omega}^*_\kappa \in  \argmin_{\bm{\omega} \in {\cal W}} \ 
  \rho(\boldsymbol{\omega}, \bm{X}, X_0) +\kappa \,  \varphi(\boldsymbol{\omega}, \bm{X}, X_0)$.  Hence,   $\bm{\omega}^*_\kappa$ is a Pareto-optimal portfolio.           
\end{proof}

\begin{proof}[\textbf{Proof of Theorem \ref{theo-mean-upm}}]
For any  $(G,F)\in \mathcal{M}_{n+1}(\bm{\mu},\bm{\Sigma})$  
and any $\bm{\omega} \in \R^n$,  define vector $\bm{a}^{\top}=(\bm{\omega}, -1)$, it is easy to verify    that 
    \be
\mathbb{E}[\boldsymbol{\omega}^{\top}\bm{X}^G-X_0^F] &=&  \bm{a}^{\top}\bm{\mu}, \nonumber\\
{\rm Var}(\boldsymbol{\omega}^{\top}\bm{X}^G-X_0^F) &=& \bm{a}^{\top}\bm{\Sigma}\bm{a} \ = \  \bm{\omega}^{\top}\bm{\Sigma}_n\bm{\omega}-2\bm{b}^{\top}\bm{\omega}+\sigma_0^2  \, >  \, 0. \label{Var-TE}
\ee
Denote $t^2(s)$ and $t(s)$  by  the variance and standard deviation  of $\boldsymbol{\omega}_s^{\top}\bm{X}^G-X^F_0$, 
where 
$\boldsymbol{\omega}_s$ is given in \eqref{w_s}.  By \eqref{Var-TE}, we obtain that for any $s \in \R$,
 \be
 \label{ts}
 t(s)=\sqrt{b_0s^2-2(c_3b_0-c_4b_1+b_1)s+b_0c_3^2-2b_1c_3(c_4-1)+b_2(c_4-1)^2-c_5+\sigma_0^2}.
\ee

Thus, by the general projection property (see, e.g.,  Theorem 1 of \citet{07PI} or Lemma 2.4 of \citet{chen2011tight}), we have that 
\be
\label{M=L}
\sup_{(G,F)\in \mathcal{M}_{n+1}(\bm{\mu},\bm{\Sigma})} \, \mathbb{E}\big[(\boldsymbol{\omega}^{\top}\bm{X}^G-X_0^F-\alpha)_+^m\big]=\sup_{F\in \mathcal{L}(\bm{a}^{\top}\bm{\mu},\ \sqrt{\bm{a}^{\top}\bm{\Sigma}\bm{a}})} \, \mathbb{E}\big[(X^F-\alpha)_+^m\big].
 \ee
Further, $\bm{\omega}^{\top}\bm{\mu}_n = \bm{a}^{\top}\bm{\mu}+\mu_0$. Thus,  
 problem  \eqref{mean-upm-M} reduces to the following problem:
\begin{equation}
\label{min-mean-f_m}
\begin{aligned} 
  &  \min_{ \bm{\omega} \in \R^n} \bm{a}^{\top}\bm{\mu}+\mu_0+\kappa \, f_{m} \Big (\bm{a}^{\top}\bm{\mu},\sqrt{\bm{a}^{\top}\bm{\Sigma}\bm{a}} \Big ) \\
& \ \text{s.t.}  \quad  \ -\bm{\omega}^{\top}\bm{\mu}  \ge r_0, \ \ \ \ \boldsymbol{\omega}^{\top}\boldsymbol{e}=1,   
      \end{aligned} 
        \end{equation}
which is equivalent to the following problem: 
\begin{equation}
\label{min-h_m}
    \begin{aligned}
        &\min_{\boldsymbol{\omega}\in\mathbb{R}^n}\left\{h_{m, \kappa} \Big(\boldsymbol{a}^{\top}\boldsymbol{\mu},\sqrt{\boldsymbol{a}^{\top}\boldsymbol{\Sigma}\boldsymbol{a}}\Big): \  -\bm{\omega}^{\top}\bm{\mu}_n  \ge r_0, \  \  \boldsymbol{\omega}^{\top}\boldsymbol{e}=1\right\}  \\
       \end{aligned}
\end{equation}
as $\mu_0$ is a given constant in $ \mathcal{M}_{n+1}(\bm{\mu},\bm{\Sigma})$.

Note that  $ \bm{a}^{\top}\bm{\mu}=\bm{\omega}^{\top}\bm{\mu}_n - \mu_0$ and that $h_{m, \kappa} \big(\boldsymbol{a}^{\top}\boldsymbol{\mu},\sqrt{\boldsymbol{a}^{\top}\boldsymbol{\Sigma}\boldsymbol{a}}\big) =h_{m, \kappa} \big(\bm{\omega}^{\top}\bm{\mu}_n - \mu_0, \sqrt{\boldsymbol{a}^{\top}\boldsymbol{\Sigma}\boldsymbol{a}}\big)  $ is increasing and continuous in $\bm{\omega}^{\top}\bm{\mu}_n$ by Proposition \ref{pro-f-h}.
  Thus,   by introducing $s=\bm{\omega}^{\top}\bm{\mu}_n$, problem \eqref{min-h_m}
reduces to the following problem: 
\begin{equation} 
\label{min-min-s} 
           \min_{ s \le -r_0} \  \min_{\bm{\omega}^{\top}\bm{\mu}_n=s,\ \bm{\omega}^{\top}\bm{e}=1} h_{m, \kappa} \Big (s-\mu_0,\sqrt{\bm{a}^{\top}\bm{\Sigma}\bm{a}} \Big ).
         \end{equation}

Further, by Proposition \ref{pro-f-h},  for any fixed  $x \in\mathbb{R}$, function $h_m(x,y)$ is increasing and continuous  in $y \in \R^+$. 
Thus, 
\begin{equation*}
   \min_{\bm{\omega}^{\top}\bm{\mu}_n=s,\ \bm{\omega}^{\top}\bm{e}=1} 
   h_{m,\kappa}\Big (s-\mu_0,\sqrt{ \bm{a}^{\top}\bm{\Sigma}\bm{a}} \Big )
   =h_{m,\kappa}\Big(s-\mu_0,\ \sqrt{ \min_{\bm{\omega}^{\top}\bm{\mu}_n=s,\ \bm{\omega}^{\top}\bm{e}=1}\bm{a}^{\top}\bm{\Sigma}\bm{a}}\Big).
       \end{equation*}
    By Lemma  \ref{theo-w*} and \eqref{Var-TE} we obtain
$
\min_{\bm{\omega}^{\top}\bm{\mu}_n=s,\ \bm{\omega}^{\top}\bm{e}=1}\bm{a}^{\top}\bm{\Sigma}\bm{a}
 =
\bm{\omega}_s^{\top}\bm{\Sigma}_n\bm{\omega}_s-2\bm{b}^{\top}\bm{\omega}_s+\sigma_0^2, 
 $
where $\bm{\omega}_s$ is given in \eqref{w_s}. Thus, problem \eqref{min-min-s} reduces to  
\be
\label{min-s}
  \min_{ s \le -r_0} \  h_{m, \kappa} \Big (s-\mu_0,\sqrt{\bm{\omega}_s^{\top}\bm{\Sigma}_n\bm{\omega}_s-2\bm{b}^{\top}\bm{\omega}_s+\sigma_0^2 } \Big )
\ee
and  the optimal solution $\bm{\omega}^*$ to  problem  \eqref{mean-upm-M} is equal to $\bm{\omega}_{s^*}$, where $s^*$ is the solution to 
 problem \eqref{min-s}. In other words, 
 \be
 \label{s*<-r0-gmk} 
   s^*=\argmin_{s \le -r_0} \ g_{m, \kappa} (s).
 \ee
 where the function $g_{m,\kappa}(s)$ is defined as 
\be
g_{m,\kappa}(s) = h_{m, \kappa} \Big (s-\mu_0,\sqrt{\bm{\omega}_s^{\top}\bm{\Sigma}_n\bm{\omega}_s-2\bm{b}^{\top}\bm{\omega}_s+\sigma_0^2 } \Big )
=h_{m, \kappa} \big (s-\mu_0, \, t(s) \big ).
\ee

 \begin{enumerate}[(i)]
   \item Consider problem \eqref{mean-upm-M} for $m=1$. By Lemma \ref{lemma-jan-1977} and the definition of   $h_{m,\kappa}(s,t)$, we have 
   \by 
   g_{1, \kappa}(s)=s-\mu_0+\frac{\kappa}{2}\Big (s-\alpha^*+\sqrt{(s-\alpha^*)^2+t^2(s)} \Big ).
   \ey
          Note  that  $(s-\alpha^*)^2+t^2(s) = (b_0+1)s^2-2(b_0c_3-b_1c_4+b_1+\alpha^*)s+C,$
where 
$$C=b_0c_3^2-2b_1c_3(c_4-1)+b_2(c_4-1)^2-c_5+\sigma_0^2+\alpha^{*2}.$$
  By differentiating $g_{1, \kappa}(s)$, we obtain 
\begin{equation}\label{h'}
    g_{1,\kappa}'(s)=1+\frac{\kappa}{2}\Big(1+\frac{(b_0+1)s-(b_0c_3-b_1c_4+b_1+\alpha^*)}{\sqrt{(s-\alpha^*)^2+t^2(s)  }}\Big)
\end{equation}

\begin{enumerate}[(a)]
    \item If  $\kappa> \kappa_0$,      it is easy to verify that  $g_{1,\kappa}''(s)>0$ for all $s\in\mathbb{R}$. Moreover,
    the following  equation 
      \begin{equation}\label{eq:h'}
        g_{1,\kappa}'(s)=1+\frac{\kappa}{2}\Big(1+\frac{(b_0+1)s-(b_0c_3-b_1c_4+b_1+\alpha^*)}{\sqrt{(s-\alpha^*)^2+t^2(s)  }}\Big)=0
\end{equation}
has a unique solution \begin{equation}\label{sol-m1-s}
    s_{1,\kappa} =\frac{c_3b_0-b_1(c_4-1)+\alpha^*}{b_0+1}-\frac{\kappa+2}{b_0+1}\sqrt{\frac{q}{b_0\kappa^2-4(\kappa+1)}},
\end{equation}
where $q$ is given below equation  \eqref{w1*}.
 Hence, $s_{1,\kappa}$ is the global minimizer of $g_{1,\kappa}(s)$ on $\R$. Therefore, $s^* = \min\{-r_0, s_{1, \kappa}\}$ is 
 the unique global optimal solution of problem \eqref{s*<-r0-gmk} since $g_{1,\kappa}$ is strictly decreasing on $(-\infty, s_{1,\kappa})$ and strictly increasing on $(s_{1,\kappa}, \ \infty)$. Thus, we obtain \eqref{w1*}. 
    
        \item 
If  $\kappa\leq \kappa_0$, we  verify that  $g_{1,\kappa}'(s)>0$ for all $s\in\mathbb{R}$; 
 in other words, $g_{1,\kappa}(s)$ is strictly increasing in $s\in\mathbb{R}$. Therefore,   problem \eqref{s*<-r0-gmk} has  no  solution.

\end{enumerate}
Combining the above cases, we complete the proof for  $m=1$.

\item Consider problem \eqref{mean-upm-M} for  $m=2$. By Lemma \ref{lemma-jan-1977} and the definition of   the function $h_{m,\kappa}(s)$, 
it holds  that 
$
g_{2,\kappa}(s)=s-\mu_0+\kappa \Big ((s-\alpha^*)_+^2+t^2(s)\Big).
$
 Since the leading coefficient $b_0$ of  the  quadratic function $t^2(s)$ is positive,  $t^2(s)$ is strictly convex for all $s \in \R$.
 
 In addition,  $(s-\alpha^*)_+^2$ is convex in $s \in \R$, and the function $s$ is linear and thus convex. Therefor,  $g_{2,\kappa}(s)$  is strictly convex in 
 $s \in \R$. 
            Furthermore, for any $s \in \R$, 
       $g_{2, \kappa}(s) \geq s-\mu_0 +\kappa \ t^2(s).$
   Since   the   leading coefficient $\kappa b_0$ of  the  quadratic function $s-\mu_0+ \kappa \ t^2(s)$ is positive, we have that  
   $s-\mu_0+ \kappa \ t^2(s)$ is coercive 
     on $s \in \R$.    Consequently,  $g_{2,\kappa}(s)$ is also coercive on $s\in\mathbb{R}$. Therefore,  $g_{2,\kappa}(s)$ has a unique 
      global minimizer $s_{2, \kappa}$ on  $\mathbb{R}$ as it is coercive and strictly convex on $\R$.
      
      Note that 
    \begin{equation}
    \begin{aligned}
        g_{2,\kappa}(s)=\begin{cases}
       u(s) \triangleq  s+\kappa \ t^2(s), \ &s\leq\alpha^*,\\
       v(s) \triangleq  s+\kappa \ \big((s-\alpha^*)^2+t^2(s)\big), \ &s>\alpha^*,
    \end{cases}
    \end{aligned}
\end{equation}

 If $s_{2,\kappa} \le \alpha^*$, then $s_{2,\kappa} = \alpha^*$  or   $s_1$, where $s_1$ is the unique solution  to the equation 
  $
  u'(s)=1+2\kappa\big(b_0s-(b_0c_3-b_1c_4+b_1)\big)=0, 
  $
  which gives 
$$s_1 = \frac{2b_0c_3\kappa-2(c_4-1)b_1\kappa-1}{2b_0\kappa}.$$

 If $s_{2,\kappa} > \alpha^*$, then $s_{2,\kappa} = s_2$ is  the unique solution  to equation 
  $
v'(s)=1+2\kappa\big((b_0+1)s-(b_0c_3-b_1c_4+b_1+\alpha^*)\big)=0,
$
which yields 
 $$s_2 = \frac{2b_0c_3\kappa-2(c_4-1)b_1\kappa+2\kappa\alpha^*-1}{2\kappa(b_0+1)}.$$

Note that if $s_1> \alpha^*$, then $s_{2,\kappa}=s_2$ and
\by
s_1-\frac{(s_1-\alpha^*)_+}{b_0+1} =s_1-\frac{s_1-\alpha^*}{b_0+1} =s_2.
\ey
Therefore, the unique 
      global minimizer $s_{2,\kappa}$ of function $g_{2,\kappa}(s)$ over $\mathbb{R}$ can be expressed as 
\by
s_{2,\kappa}
=s_1-\frac{(s_1-\alpha^*)_+}{b_0+1}.
\ey
 Consequently, $s^* = \min\{-r_0, s_{2, \kappa}\}$ is 
 the unique  solution of problem \eqref{s*<-r0-gmk},  since $g_{2,\kappa}$ is strictly decreasing on $(-\infty, s_{2,\kappa})$ and strictly increasing on $(s_{2,\kappa}, \ \infty)$. Thus, we obtain \eqref{w2*} by  denoting  $s_1$ as $s_\kappa$ to indicate its  
 dependence on $\kappa$.  
       \end{enumerate}
\end{proof}

\begin{proof}[\textbf{Proof of Theorem \ref{pro-upm-M-sol}}]
The proof of Proposition \ref{pro-upm-M-sol} follows from the same arguments as those used 
for Theorem \ref{theo-mean-upm}, so the details are omitted. In fact, Proposition \ref{pro-upm-M-sol}(i) 
and (ii) correspond to the limits of Theorem \ref{theo-mean-upm}(i) and (ii), respectively, as $\kappa \to \infty$.
\end{proof}

\begin{proof}[\textbf{Proof of Theorem \ref{th-EIT-upm-cop}}]
First,  by  the definitions of  sets  ${\cal D}_{n+1}(\gamma_1, \gamma_2)$,  
 $ \mathcal{M}_{n+1}(\bm{\mu},\bm{\Sigma})$, and   $ \mathcal{U}_{(\hat{\bm{\mu}},\hat{\bm{\Sigma}})}$, 
 it is easy to see  that 
    \be
    \label{fun-D}
        && \sup_{(G,F)\in {\cal D}_{n+1}(\gamma_1, \gamma_2)}\mathbb{E}[\bm{\omega}^{\top}\bm{X}^G]+\kappa\  \sup_{(G,F)\in {\cal D}_{n+1}(\gamma_1, \gamma_2)}  \mathbb{E}\big[(\bm{\omega}^{\top}\bm{X}^G-X_0^F-\alpha)_{+}^m\big] \nonumber \\
        &=&\sup_{(\bm{\mu},\bm{\Sigma})\in \mathcal{U}_{(\hat{\bm{\mu}},\hat{\bm{\Sigma}})}} \sup_{(G,F) \in \mathcal{M}_{n+1}(\bm{\mu},\bm{\Sigma})} \mathbb{E}[\bm{\omega}^{\top}\bm{X}^G]+\kappa \ \sup_{(\bm{\mu},\bm{\Sigma})\in \mathcal{U}_{(\hat{\bm{\mu}},\hat{\bm{\Sigma}})}} \sup_{(G,F) \in \mathcal{M}_{n+1}(\bm{\mu},\bm{\Sigma})}   \mathbb{E}\big[(\bm{\omega}^{\top}\bm{X}^G-X_0^F-\alpha)_{+}^m\big] \nonumber \\
      &=&\sup_{(\bm{\mu},\bm{\Sigma})\in \mathcal{U}_{(\hat{\bm{\mu}},\hat{\bm{\Sigma}})}} \bm{\omega}^{\top}\bm{\mu}_n+\kappa  \sup_{(\bm{\mu},\bm{\Sigma})\in \mathcal{U}_{(\hat{\bm{\mu}},\hat{\bm{\Sigma}})}}  \ f_m \Big (\bm{a}^{\top}\bm{\mu},\ \sqrt{\bm{a}^{\top} {\bm{\Sigma}}\bm{a}} \Big )  \nonumber\\
       & =& \bm{\omega}^{\top}\hat{\bm{\mu}}_n+\sqrt{\gamma_1}\sqrt{\bm{\omega}^{\top}\hat{\bm{\Sigma}}_n\bm{\omega}}
+  \kappa \,  \sup_{\bm{\mu}\in \mathcal{U}_{\hat{\bm{\mu}}}} \sup_{\bm{\Sigma}\in \mathcal{U}_{\hat{\bm{\Sigma}}}} f_m \Big (\bm{a}^{\top}\bm{\mu},\ \sqrt{\bm{a}^{\top}{\bm{\Sigma}}\bm{a}} \Big ),
\ee
where the second equality follows from the definition of function $f_m(\mu, \sigma)$ defined in \eqref{f-mu-sigma} and the general projection property, and the last equality follows from equation (7) in \citet{kang2022robust}. 
  
For any given $\bm{\mu} \in\mathcal{U}_{\hat{\bm{\mu}}}$,  
the function $f_m \Big (\bm{a}^{\top}\bm{\mu},\ \sqrt{\bm{a}^{\top}\bm{\Sigma}\bm{a}} \Big )$ is 
increasing and continuous in  $\bm{a}^{\top}\bm{\Sigma}\bm{a}$. Hence, 
\begin{equation*}
    \sup_{\bm{\Sigma}\in \mathcal{U}_{\hat{\bm{\Sigma}}}} 
        f_{m} \Big (\bm{a}^{\top}\bm{\mu}, \ 
       \sqrt{\bm{a}^{\top}\bm{\Sigma}\bm{a}} \Big ) =
        f_{m} \Big (\boldsymbol{a}^{\top}\boldsymbol{\mu}, \ 
        \sqrt{\sup_{\bm{\Sigma}\in \mathcal{U}_{\hat{\bm{\Sigma}}}} \, \bm{a}^{\top}\bm{\Sigma}
        \boldsymbol{a}} \Big )
        =
         f_{m}\Big(\bm{a}^{\top}\bm{\mu}, \ 
        \sqrt{\bm{a}^{\top}(\hat{\bm{\Sigma}}+\gamma_2\bm{I}_{n+1})
        \bm{a}} \Big),
\end{equation*}
where  the last equality follows from equation (8) in \citet{kang2022robust}.  Thus, equation  \eqref{fun-D} reduces  to 
\begin{equation}
 \bm{\omega}^{\top}\hat{\bm{\mu}}_n+\sqrt{\gamma_1}\sqrt{\bm{\omega}^{\top}\hat{\bm{\Sigma}}_n\bm{\omega}}
+  \kappa \,   \sup_{\bm{\mu}\in \mathcal{U}_{\hat{\bm{\mu}}}} f_{m}\Big(\bm{a}^{\top}\bm{\mu}, \ 
        \sqrt{\bm{a}^{\top}(\hat{\bm{\Sigma}}+\gamma_2\bm{I}_{n+1})
        \bm{a}} \Big ).
\end{equation}
Moreover,  function $f_{m}\Big(\bm{a}^{\top}\bm{\mu}, \ 
        \sqrt{\bm{a}^{\top}(\hat{\bm{\Sigma}}+\gamma_2\bm{I}_{n+1})
        \bm{a}} \Big )$
        is increasing and continuous in  $\bm{a}^{\top}\bm{\mu}$, hence, 
\by
    \sup_{\boldsymbol{\mu}\in \mathcal{U}_{\hat{\boldsymbol{\mu}}}} 
        f_{m}\Big (\boldsymbol{a}^{\top}\boldsymbol{\mu}, \ 
       \sqrt{\bm{a}^{\top}(\hat{\bm{\Sigma}}+\gamma_2\bm{I}_{n+1})
        \bm{a}} \Big ) &=&
        f_{m} \Big (\sup_{\bm{\mu}\in \mathcal{U}_{\hat{\bm{\mu}}}}\bm{a}^{\top}\bm{\mu}, \ 
        \sqrt{\bm{a}^{\top}(\hat{\bm{\Sigma}}+\gamma_2\bm{I}_{n+1})
        \bm{a}} \Big )\\
      &=&f_{m} \Big (\bm{\omega}^{\top}\hat{\bm{\mu}}_n+\sqrt{\gamma_1}\sqrt{\bm{\omega}^{\top}\hat{\bm{\Sigma}}_n\bm{\omega}}-\hat{\mu}_0, \ 
        \sqrt{\bm{a}^{\top}(\hat{\bm{\Sigma}}+\gamma_2\bm{I}_{n+1})
        \bm{a}} \Big ),
        \ey
 where the second equation comes from $\bm{a}^{\top}\bm{\mu} = \bm{\omega}^{\top}\bm{\mu}_n-\hat{\mu}_0$, 
 together with equation (7) of   \citet{kang2022robust}.   
    
  For the constraint on the worst-case portfolio return in ${\cal W}_r$ for problem \eqref{mean-rho-D},
for any given $\bm{\mu} \in \mathcal{U}_{\hat{\bm{\mu}}}$, we have  
\by
 \min_{(G, F) \in {\cal D}_{n+1}(\gamma_1, \gamma_2)}\mathbb{E}[-\boldsymbol{\omega}^{\top}\boldsymbol{X}^G] 
 &=& -   \max_{(G, F) \in {\cal D}_{n+1}(\gamma_1, \gamma_2)}\mathbb{E}[\boldsymbol{\omega}^{\top}\boldsymbol{X}^G] \\
&=& -  \sup_{(\bm{\mu},\bm{\Sigma})\in \mathcal{U}_{(\hat{\bm{\mu}},\hat{\bm{\Sigma}})}} \sup_{(G,F) \in \mathcal{M}_{n+1}(\bm{\mu},\bm{\Sigma})} 
     \mathbb{E}[\boldsymbol{\omega}^{\top}\boldsymbol{X}^G] \\
     &=& -\sup_{(\bm{\mu},\bm{\Sigma})\in \mathcal{U}_{(\hat{\bm{\mu}},\hat{\bm{\Sigma}})}} \bm{\omega}^{\top}\bm{\mu}_n
      =   - \bm{\omega}^{\top}\hat{\bm{\mu}}_n-\sqrt{\gamma_1}\sqrt{\bm{\omega}^{\top}\hat{\bm{\Sigma}}_n\bm{\omega}}. 
                \ey

Therefore,  problem \eqref{mean-rho-D} is equivalent to the following problem:
\begin{equation}
\label{min-mean-f_m} 
\begin{aligned}       
    \min_{\bm{\omega} \in \cal W}  \, &  \bm{\omega}^{\top}\hat{\bm{\mu}}_n+\sqrt{\gamma_1}\sqrt{\bm{\omega}^{\top}\hat{\bm{\Sigma}}_n\bm{\omega}}
+  \kappa \, f_{m}\Big (\bm{\omega}^{\top}\hat{\bm{\mu}}_n+\sqrt{\gamma_1}\sqrt{\bm{\omega}^{\top}\hat{\bm{\Sigma}}_n\bm{\omega}}-\hat{\mu}_0, \ 
        \sqrt{\bm{a}^{\top}(\hat{\bm{\Sigma}}+\gamma_2\bm{I}_{n+1})
        \bm{a}} \Big ) \\
    \textrm{s.t.} \ 
  &- \bm{\omega}^{\top}\hat{\bm{\mu}}_n-\sqrt{\gamma_1}\sqrt{\bm{\omega}^{\top}\hat{\bm{\Sigma}}_n\bm{\omega}} \geq r,\ \ \  \bm{\omega}^{\top}\bm{e}=1.
\end{aligned}
\end{equation}

By the definition of  function $h_m(\mu, \sigma)$ given in \eqref{h_m_k} and noting that $\hat{\mu}_0$ is constant, 
we have that  problem \eqref{min-mean-f_m} is equivalent 
 to the following problem:
       \begin{equation}
       \label{min-h_m_k-D}
\begin{aligned}       
   & \min_{\bm{\omega} \in {\cal W}}  \, 
 h_{m, \kappa}\Big (\bm{\omega}^{\top}\hat{\bm{\mu}}_n+\sqrt{\gamma_1}\sqrt{\bm{\omega}^{\top}\hat{\bm{\Sigma}}_n\bm{\omega}}-\hat{\mu}_0, \ 
        \sqrt{\bm{a}^{\top}(\hat{\bm{\Sigma}}+\gamma_2\bm{I}_{n+1})
        \bm{a}} \Big ) \\
  &   \textrm{s.t.} \ - \bm{\omega}^{\top}\hat{\bm{\mu}}_n-\sqrt{\gamma_1}\sqrt{\bm{\omega}^{\top}\hat{\bm{\Sigma}}_n\bm{\omega}} \geq r,\ \  \bm{\omega}^{\top}\bm{e}=1.
\end{aligned}
\end{equation}

By Proposition \ref{pro-f-h},  the function $h_{m,\kappa}(\mu, \sigma) $ is  increasing and continuous  in each argument.
Thus, by introducing auxiliary variables $s, t \in\mathbb{R}^+$ and noting that $\alpha^*=\alpha+\hat{\mu}_0$,  problem 
\eqref{min-h_m_k-D}  is equivalent to  
  \begin{equation}
       \label{min-h_m_k-D-s-t}
\begin{aligned}       
    \min_{\bm{\omega} \in {\cal W}, \, s, \, t \in \R^+} \, 
 & \ h_{m, \kappa}\Big (\bm{\omega}^{\top}\hat{\bm{\mu}}_n+\sqrt{\gamma_1} \, t -\hat{\mu}_0, \ 
        s \Big ) \\
      \textrm{s.t.}  & \ \sqrt{\bm{a}^{\top}(\hat{\bm{\Sigma}}+\gamma_2 \bm{I}_{n+1})\bm{a}} \leq s,\\
        & \ \sqrt{\bm{\omega}^{\top}\hat{\bm{\Sigma}}_n\bm{\omega}}\leq t,\\
        & \  - \bm{\omega}^{\top}\hat{\bm{\mu}}_n-\sqrt{\gamma_1}\sqrt{\bm{\omega}^{\top}\hat{\bm{\Sigma}}_n\bm{\omega}} \geq r,\  \ \bm{\omega}^{\top}\bm{e}=1.      
        \end{aligned}
\end{equation}

By Proposition \ref{pro-f-h},  the function $h_{m,\kappa}(\mu, \sigma) $ is convex in $(\mu,\sigma) \in \R \times \R^+$. 
Moreover,  
convexity is preserved under  affine mappings. Therefore, the objective  function in problem \eqref{min-h_m_k-D-s-t}, namely, 
$h_{m, \kappa}\big (\bm{\omega}^{\top}\hat{\bm{\mu}}_n+\sqrt{\gamma_1} \, t -\hat{\mu}_0, \ 
        s \big ) $, is convex in $(\bm{\omega}, s, t)$.
    
        In addition,            for $m=1$, we have 
\begin{equation}
\label{min-h1k}
    \begin{aligned}
        &h_{1, \kappa} (\bm{\omega}^{\top}\hat{\bm{\mu}}_n+\sqrt{\gamma_1} \, t-\hat{\mu}_0, s)\\
        =& \  \bm{\omega}^{\top}\hat{\bm{\mu}}_n+\sqrt{\gamma_1} \,  t + \frac{\kappa}{2}(\bm{\omega}^{\top}\hat{\bm{\mu}}_n+\sqrt{\gamma_1} \, t-\hat{{\mu}}_0-\alpha)+\frac{\kappa}{2}\sqrt{(\bm{\omega}^{\top}\hat{\bm{\mu}}_n+\sqrt{\gamma_1} \, t-\hat{\mu}_0-\alpha)^2+s^2}\\
        =& \  (1+\frac{\kappa}{2})(\bm{\omega}^{\top}\hat{\bm{\mu}}_n+\sqrt{\gamma_1} \, t) - \frac{\kappa}{2}\alpha^*+\frac{\kappa}{2}\sqrt{(\bm{\omega}^{\top}\hat{\bm{\mu}}_n+\sqrt{\gamma_1} \, t-\alpha^*)^2+s^2}.
    \end{aligned}
\end{equation}
Note that  $\alpha^*$ and $\kappa >0$ are  constants, minimizing the objective function in the second equation of \eqref{min-h1k} 
is therefore equivalent to minimizing  
$ \bm{\omega}^{\top}\hat{\bm{\mu}}_n+\sqrt{\gamma_1} \, t +\frac{\kappa}{2+\kappa}\sqrt{(\bm{\omega}^{\top}\hat{\bm{\mu}}_n+\sqrt{\gamma_1} \, t-\alpha^*)^2+s^2},$ since minimizing a function is equivalent to minimizing any of 
its strictly increasing transformations. Together with \eqref{min-h_m_k-D-s-t}, this  yields \eqref{ob-m=1}. 

For  $m=2$, we have 
   $      h_{2, \kappa} (\bm{\omega}^{\top}\hat{\bm{\mu}}_n+\sqrt{\gamma_1} \,  t-\hat{\mu}_0, s) 
        = \ \bm{\omega}^{\top}\hat{\bm{\mu}}_n+\sqrt{\gamma_1} \, t + \kappa \big((\bm{\omega}^{\top}\hat{\bm{\mu}}_n+\sqrt{\gamma_1} \, t-\alpha^*)_+^2+s^2\big),
        $
which, together with \eqref{min-h_m_k-D-s-t}, yields \eqref{ob-m=2}. This completes 
 the proof of Theorem \ref{th-EIT-upm-cop}.
\end{proof}

\begin{proof}[\textbf{Proof of Proposition \ref{th-EIT-upm-cop-only}}]
The proof of Proposition \ref{th-EIT-upm-cop-only}    follows from the same arguments as those used for 
Theorem \ref{th-EIT-upm-cop}, so the details are omitted. In fact, Proposition \ref{th-EIT-upm-cop-only}(i) and (ii) correspond to the limits of Theorem  \ref{th-EIT-upm-cop}
 (i) and (ii), respectively, as $\kappa \to \infty$.
\end{proof}

\begin{proof}[\textbf{Proof of Proposition \ref{pro-ob-MV-model-D-g1-g2}}]
(i)  First, by  the definitions of  ${\cal D}_{n+1}(\gamma_1, \gamma_2)$,  
 $ \mathcal{M}_{n+1}(\bm{\mu},\bm{\Sigma})$, and   $ \mathcal{U}_{(\hat{\bm{\mu}},\hat{\bm{\Sigma}})}$, for any given $\bm{\omega} \in \cal W$, 
 it holds that 
    \by
        && \sup_{(G,F)\in {\cal D}_{n+1}(\gamma_1, \gamma_2)} 
        \mathbb{E}[\bm{\omega}^{\top}\bm{X}^G] 
        + \kappa\  \sup_{(G,F)\in {\cal D}_{n+1}(\gamma_1, \gamma_2)} 
        Var(\bm{\omega}^{\top}\bm{X}^G) \nonumber \\
        &=&\sup_{(\bm{\mu},\bm{\Sigma})\in \mathcal{U}_{(\hat{\bm{\mu}},\hat{\bm{\Sigma}})}} \sup_{(G,F) \in \mathcal{M}_{n+1}(\bm{\mu},\bm{\Sigma})} \mathbb{E}[\bm{\omega}^{\top}\bm{X}^G] 
        +\kappa \ \sup_{(\bm{\mu},\bm{\Sigma})\in \mathcal{U}_{(\hat{\bm{\mu}},\hat{\bm{\Sigma}})}} \sup_{(G,F) \in \mathcal{M}_{n+1}(\bm{\mu},\bm{\Sigma})}   
        Var(\bm{\omega}^{\top}\bm{X}^G) \nonumber \\
      &=&\sup_{(\bm{\mu},\bm{\Sigma})\in \mathcal{U}_{(\hat{\bm{\mu}},\hat{\bm{\Sigma}})}} \bm{\omega}^{\top}\bm{\mu}_n+\kappa  \sup_{(\bm{\mu},\bm{\Sigma})\in \mathcal{U}_{(\hat{\bm{\mu}},\hat{\bm{\Sigma}})}} 
      \   \bm{\omega}^{\top}\bm{\Sigma}_n\bm{\omega}  \nonumber \\
             & =& \bm{\omega}^{\top}\hat{\bm{\mu}}_n+\sqrt{\gamma_1}\sqrt{\bm{\omega}^{\top}\hat{\bm{\Sigma}}_n\bm{\omega}}
+  \kappa \, 
 \bm{\omega}^{\top}\big (\hat{\bm{\Sigma}}_n
+\gamma_2\bm{I}_{n} \big )\bm{\omega}  = f_\kappa \Big (\bm{\omega}^{\top} \hat{\bm{\mu}}_n+\sqrt{\gamma_1}\sqrt{\bm{\omega}^{\top}\hat{\bm{\Sigma}}_n\bm{\omega}}, \ 
        \sqrt{\bm{\omega}^{\top}(\hat{\bm{\Sigma}}_n+\gamma_2\bm{I}_{n})
        \bm{\omega}} \Big ),
        \ey
where the   third  equality follows from (7) in \citet{kang2022robust}, and function $f_\kappa(x, y)$ is defined as $f_\kappa(x, y) = x+\kappa y^2$ for $(x, y) \in \R \times \R^+$.   Therefore,  problem \eqref{MV-model-D-g1-g2} is equivalent to the following problem: 
  \begin{equation}
\label{min-f_k-D} 
\begin{aligned}       
    \min_{\bm{\omega} \in {\cal W}}  \, &  \, f_{\kappa}\Big (\bm{\omega}^{\top}\hat{\bm{\mu}}_n+\sqrt{\gamma_1}\sqrt{\bm{\omega}^{\top}\hat{\bm{\Sigma}}_n\bm{\omega}}, \ 
        \sqrt{\bm{\omega}^{\top}(\hat{\bm{\Sigma}}_n+\gamma_2\bm{I}_{n})
        \bm{\omega}} \Big )\\
    \textrm{s.t.} \ 
  &- \bm{\omega}^{\top}\hat{\bm{\mu}}_n-\sqrt{\gamma_1}\sqrt{\bm{\omega}^{\top}\hat{\bm{\Sigma}}_n\bm{\omega}} \geq r,\ \ \  \bm{\omega}^{\top}\bm{e}=1.
\end{aligned}
\end{equation}
Note that  function $f_\kappa(x, y)$ is  increasing and continuous  in each of its arguments on  $(x, y) \in \R \times \R^+$. 
Thus, by introducing auxiliary variables $s, t \in\mathbb{R}^+$,   problem 
\eqref{min-f_k-D} is reformulated as 
  \begin{equation}
       \label{min-f_k-MV-D} 
\begin{aligned}       
    \min_{\bm{\omega} \in {\cal W}, \, s, \, t \in \R^+}  \, 
 & \ f_{\kappa}\Big (\bm{\omega}^{\top}\hat{\bm{\mu}}_n+\sqrt{\gamma_1} \, t, \ 
        s \Big ) \\
      \textrm{s.t.}  & \ \sqrt{\boldsymbol{\omega}^{\top}(\hat{\boldsymbol{\Sigma}}_n+\gamma_2 \boldsymbol{I}_{n})\boldsymbol{\omega}} \leq s,\\
        & \ \sqrt{\bm{\omega}^{\top}\hat{\bm{\Sigma}}_n\bm{\omega}}\leq t,\\
        & \  - \bm{\omega}^{\top}\hat{\bm{\mu}}_n-\sqrt{\gamma_1}\sqrt{\bm{\omega}^{\top}\hat{\bm{\Sigma}}_n\bm{\omega}} \geq r, \ \ \boldsymbol{\omega}^{\top}\boldsymbol{e}=1.      
        \end{aligned}
\end{equation}

Obviously, function $f_{\kappa}(\mu, \sigma) $ is convex in $(\mu,\sigma)$.  Moreover,  
convexity is preserved under  affine mappings. Therefore, the objective  function 
$f_{\kappa}\big (\bm{\omega}^{\top}\hat{\bm{\mu}}_n+\sqrt{\gamma_1} \, t, \ 
        s \big ) $ in problem \eqref{min-f_k-MV-D}  is convex in $(\bm{\omega}, s, t)$.
       Note that  
      $
        f_{\kappa} (\bm{\omega}^{\top}\hat{\bm{\mu}}_n+\sqrt{\gamma_1} \, t, s)
        = \  \bm{\omega}^{\top}\hat{\bm{\mu}}_n+\sqrt{\gamma_1} \, t + \kappa s^2.$ 
          This completes the proof of Proposition \ref{pro-ob-MV-model-D-g1-g2}(i).   
  
   (ii) The proof of Proposition  \ref{pro-ob-MV-model-D-g1-g2}(ii)   
       follows from the same arguments as those used for  Proposition  \ref{pro-ob-MV-model-D-g1-g2}(i), 
       so the details are omitted. In fact,  Proposition  \ref{pro-ob-MV-model-D-g1-g2}(ii)     
       corresponds to the limit of   Proposition  \ref{pro-ob-MV-model-D-g1-g2}(i)   as $\kappa \to \infty$.   
\end{proof}

\end{document}